\documentclass{article}

\usepackage{arxiv}

\usepackage{nicematrix}
\usepackage{amsmath}
\usepackage[utf8]{inputenc} 
\usepackage[T1]{fontenc}    
\usepackage{amsthm}
\usepackage{amssymb}
\usepackage{amsfonts}
\usepackage{bbm}
\usepackage{tikz}
\usepackage{thmtools,thm-restate}
\usetikzlibrary{trees}
\usetikzlibrary{calc}

\usepackage[style=apa, backend=biber]{biblatex}
\usepackage{hyperref}       
\usepackage{bookmark}       
\usepackage{url}            
\usepackage{booktabs}       
\usepackage{amsfonts}       
\usepackage{nicefrac}       
\usepackage{microtype}      
\usepackage{cleveref}       
\usepackage{graphicx}
\usepackage{doi}
\usepackage{comment}
\usepackage{multicol}
\usepackage{multirow}
\usepackage{adjustbox}
\usepackage{pdfpages}
\usepackage{listings}

\author{Florian Stijven}

\theoremstyle{definition}
\newtheorem{definition}{Definition}[section]
\newtheorem{remark}{Remark}
\newtheorem{assumption}{Assumption}[section]

\newtheorem{theorem}{Theorem}[section]
\newtheorem{example}{Example}[section]



\title{Evaluation of Surrogate Endpoints Based on Meta-Analysis With Surrogate Indices}

\date{27 July 2026}

\newif\ifuniqueAffiliation
\uniqueAffiliationtrue

\ifuniqueAffiliation 
\author{ 
	\href{https://orcid.org/0000-0002-4574-8261}{\includegraphics[scale=0.06]{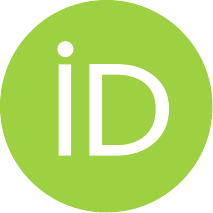}\hspace{1mm}Florian Stijven}
	\\
	I-BioStat \\
	KU Leuven\\
	Leuven, Belgium \\
	\texttt{florian.stijven@kuleuven.be} \\
	\and
	\href{https://orcid.org/0000-0002-2662-9427}{\includegraphics[scale=0.06]{orcid.pdf}\hspace{1mm}Peter B.~Gilbert} \\
	Vaccine and Infectious Disease Division \\
	Fred Hutchinson Cancer Center \\
	Seattle, WA, USA \\
	\texttt{pgilbert@fredhutch.org} \\
}
\else
\usepackage{authblk}

\newbox{\orcid}\sbox{\orcid}{\includegraphics[scale=0.06]{orcid.pdf}} 
\author[1]{%
	\href{https://orcid.org/0000-0000-0000-0000}{\usebox{\orcid}\hspace{1mm}David S.~Hippocampus\thanks{\texttt{hippo@cs.cranberry-lemon.edu}}}%
}
\author[1,2]{%
	\href{https://orcid.org/0000-0000-0000-0000}{\usebox{\orcid}\hspace{1mm}Elias D.~Striatum\thanks{\texttt{stariate@ee.mount-sheikh.edu}}}%
}
\affil[1]{Department of Computer Science, Cranberry-Lemon University, Pittsburgh, PA 15213}
\affil[2]{Department of Electrical Engineering, Mount-Sheikh University, Santa Narimana, Levand}
\fi

\renewcommand{\shorttitle}{Evaluation of trial-level surrogacy using a surrogate index}

\hypersetup{
	pdftitle={Evaluation of Surrogate Endpoints Based on Meta-Analysis With Surrogate Indices},
	pdfsubject={Surrogate endpoints; Meta-analysis},
	pdfauthor={Florian Stijven, Peter Gilbert},
	pdfkeywords={Surrogate endpoints; Meta-analysis; Surrogate index; Data-adaptive target parameter; COVID-19},
	pdfpagemode=UseOutlines,
	bookmarksopen=true,
	bookmarksopenlevel=2,
	bookmarksnumbered=true,
}

\bookmarksetup{depth=2}

\begin{document}
	\maketitle
	\begin{abstract}
		The meta-analytic (MA) framework is the gold standard for evaluating putative surrogate endpoints but it is not well-suited for complex surrogates. 
		We address this limitation by considering real-valued summaries of complex surrogates as alternative, univariate putative surrogates. We focus on the surrogate index as a specific summary measure with desirable properties. 

		We first formalize the data-generating mechanism underlying the MA framework, making explicit the assumptions required for valid inferences in any evaluation of trial-level surrogacy. These assumptions are often left implicit in the MA framework.
		Building on this formalization, we show that, under certain conditions, the surrogate index maximizes the trial-level surrogacy among all real-valued summaries. 
		When the surrogate index is estimated, we show that valid inference about the trial-level surrogacy of this estimated surrogate index is possible.

		The proposed approach can be implemented using standard software and is illustrated using COVID-19 vaccine efficacy trials, where antibody markers are evaluated as candidate surrogate endpoints.
	\end{abstract}


	\tableofcontents

	\section{Introduction}

	In clinical trials, the primary endpoint is ideally a clinical endpoint, defined as an endpoint that ``describes or reflects how an individual feels, functions, or survives'' \parencite{fda2016}. However, clinical endpoints are often costly, time-consuming, and difficult to measure, prompting trialists to use so-called surrogate endpoints as primary endpoints. Surrogate endpoints are biomarkers or intermediate outcomes that may substitute for clinical endpoints. While using a surrogate endpoint can increase a trial's efficiency, it should be evaluated before it replaces a clinical endpoint. \textcite{prentice1989surrogate} was the first to propose statistical criteria for evaluating surrogate endpoints. Since then, numerous approaches have been developed, including the proportion explained \parencite{freedman1992statistical}, principal surrogacy \parencite{frangakis2002principal}, the information-theoretic causal inference approach \parencite{alonso2015relationship}, and the meta-analytic (MA) framework \parencite{buyse2000validation, alonso2016applied}. 
	\textcite{joffe2009related} provide an overview of these approaches, including causal extensions of \textcite{prentice1989surrogate}, and how they are related to each other. 

	The meta-analytic framework is the gold standard for evaluating surrogate endpoints \parencite{fleming2012biomarkers, IOM2011, weir2022informed, vanderweele2013surrogate}. 
	However, the canonical MA framework has an important limitation: 
	it is generally restricted to the evaluation of univariate surrogates. Although there is some work on multivariate surrogates---see, e.g., \textcite{gabriel2019optimizing, alonso2004validation}---such methods can only be applied for low-dimensional surrogates (e.g., bi- or trivariate surrogates) or when additional structure can be imposed on the multivariate surrogate.
	This is inconsistent with clinical reality, where promising surrogates may be complex.
	
	In recent work, \textcite{athey2025surrogate} and \textcite{gilbert2025surrogate} developed methods for estimating treatment effects on an unobserved clinical endpoint in a randomized trial, given only a surrogate and baseline covariates from that trial and using observational data where the clinical endpoint was measured together with the surrogate and baseline covariates.
	These methods are based on the so-called \textit{surrogate index}: a real-valued function of baseline covariates and the surrogate that is defined through a regression.
	The goal in \textcite{athey2025surrogate,gilbert2025surrogate} differs from surrogate evaluation and relies on strong identifying assumptions (e.g., the surrogate should satisfy Prentice's ``full mediation'' condition). Nonetheless, the ability of the surrogate index to estimate treatment effects on an unobserved clinical endpoint highlights its potential utility within the MA framework as a putative surrogate. Furthermore, because the surrogate index is defined via a regression function, it can be estimated using flexible machine-learning techniques and it can therefore reduce a complex (e.g., high-dimensional) surrogate to a univariate putative surrogate.

	In this paper, we extend the meta-analytic framework to evaluate complex surrogates by considering an (estimated) surrogate index as the candidate surrogate endpoint. 
	Our main contributions are: 
	\begin{enumerate}
		\item We derive conditions under which valid inference about the trial-level surrogacy of an estimated surrogate index is possible. These conditions still allow flexible machine-learning estimators for the surrogate index, making the meta-analytic evaluation feasible for complex surrogates that cannot be handled by existing methods.
		\item We formalize the data-generating mechanism underlying the MA framework and make explicit the untestable assumptions required for inference about trial-level surrogacy---assumptions that are often left implicit in the meta-analytic literature.
	\end{enumerate}

	The remainder of this paper is structured as follows: First, we introduce the data-generating mechanism underlying the MA framework in Section \ref{sec:ma-framework} together with a general discussion of relevant models and parameters. In Section \ref{sec:ma-surrogate-index}, we introduce parameters, and corresponding identifying assumptions, that are relevant for evaluating surrogates in the MA framework. Some of these parameters are data-adaptive, inference about such parameters is discussed in Section \ref{sec:inference-data-adaptive-estimand}. 
	Section \ref{sec:simulations} summarizes the results of simulations. 
	The proposed approach is applied to a set of COVID-19 vaccine efficacy trials in Section \ref{sec:data-application} where antibody markers are evaluated as putative surrogates for symptomatic infection.
	We end with some concluding remarks in Section \ref{sec:conclusions}.
	
	\section{Individual-Participant Data in a Meta-Analysis}\label{sec:ma-framework}

	In this section, we formalize the data-generating mechanism underlying the MA framework for evaluating surrogate endpoints. This data-generating mechanism is applicable to the canonical MA framework but is more general and formal than what is typically presented \parencite[e.g,][]{buyse2000validation, alonso2016applied}. This degree of generality is required for the novel approach proposed in this paper. We first introduce the data-generating mechanism for the individual-participant data from a single trial and then extend this to a population of trials. We then discuss models and target parameters from a general perspective. 


	\subsection{Data from a Single Trial}\label{sec:single-trial-data}

	We assume there exists a set of trials $\mathbb{T}$, which contains all possible trials in the area of interest. We use the notation $T \in \mathbb{T}$ to refer to a single trial from this set. The object $T$ contains all information about the corresponding trial, some of which is observable (e.g., trial protocol, treatment class, and sample size) and some of which is unobservable (e.g., the true observed-data distribution in that trial).

	\begin{example}[set of trials $\mathbb{T}$]
		In the data application in Section \ref{sec:data-application}, the set of trials $\mathbb{T}$ is the set of all phase 3 COVID-19 vaccine trials that could have been conducted in 2020--2021. 
	\end{example}

	In any trial $T \in \mathbb{T}$, the observable participant-level random variables are $(X, Z, S, Y)$, where $X \in \mathcal{X}$ is a set of baseline covariates; $Z \in \mathcal{Z} := \{0, 1\}$ is the assigned treatment where $Z = 0$ indicates control treatment and $Z = 1$ indicates experimental treatment; $S \in \mathcal{S}$ is the surrogate endpoint; and $Y \in \mathcal{Y} \subseteq \mathbb{R}$ is the clinical endpoint. 
	While the control and experimental treatments may differ across trials, we assume there is always an unambiguous control and experimental treatment.
	Unless mentioned otherwise, we also assume that there are no missing data on $(X, Z, S, Y)$ and that no intercurrent events occur that preclude the measurement of $S$ or $Y$.

	We let $\mathcal{M}^*$ be a model 
	for the distribution of $(X, Z, S, Y)$ for all $T \in \mathbb{T}$. We assume simple randomization in each trial throughout this paper; hence, $X \perp Z$ for all distributions in $\mathcal{M}^*$.
	We let $\mathcal{P}: \mathbb{T} \to \mathcal{M}^*$ be the mapping that extracts the observed-data distribution from a trial, further denoted by $P^T := \mathcal{P}(T)$.

	The trial's sample size is information contained in the object $T$. Hence, we let $n: \mathbb{T} \to \mathbb{N}$ be the mapping that extracts the sample size from $T$. Given that trial $T$ is included in the meta-analysis, the observed data for trial $T$ are $n(T)$ i.i.d.~samples from $P^{T}$. The corresponding empirical within-trial distribution is denoted by $\mathbb{P}^{T}_{n} \in \mathcal{M}^*_{\text{NP}}$, where $\mathcal{M}^*_{\text{NP}}$ is a non-parametric model for the distribution of $(X, Z, S, Y)$ that contains all possible empirical distributions. 
	Note that, for a fixed trial $T$, $\mathcal{P}(T)$ and $n(T)$ are non-random elements in, respectively, $\mathcal{M}^*$ and $\mathbb{N}$. By contrast, for the same fixed trial, $P^{T}_{n}$ is a random element in $\mathcal{M}^*_{\text{NP}}$ whose distribution is determined by $P^T$ and $n(T)$.

	\subsection{Data from a Population of Trials}\label{sec:ma-setting-notation}
	
	In the MA framework, we have individual-participant data on $(X, Z, S, Y)$ from $N$ clinical trials. A clinical trial is treated as the unit of analysis in this framework; hence, $T$ is treated as a random element. We now describe the data-generating mechanism where trials are sampled from a population of trials. 
	
	\subsubsection{Population of Trials}\label{sec:population-of-trials}

	The observed data are sampled hierarchically in the MA framework (as also visualized in Figure \ref{fig:sampling-mechanism-diagram}): 
	\begin{enumerate}
		\item Trials are distributed according to the \textit{full-data} distribution $P^F \in \mathcal{M}^F$ where $\mathcal{M}^F$ is a model for the distribution of trials, which is left unspecified. From now on, we consider $T$ as a random element distributed according to $P^F$. 
		
		We further assume that the $N$ available trials are i.i.d.~sampled from $P^F$ and denote the $j$'th sampled trial by $T_j$.
		\item We are only interested in particular characteristics of the trials $T$, i.e., particular mappings from $\mathbb{T}$. Because we now treat $T$ as a random element with distribution $P^F$, $P^T := \mathcal{P}(T)$ and $n := n(T)$ are also random elements whose distribution is determined by $P^F$.
		
		For the $j$'th sampled trial, we let $P^{T_j}$ and $n_j$ be the within-trial observed-data distribution and sample size, respectively.
		\item While $n$ is an observable random variable, $P^T$ is not observable; we only observe the empirical distribution $\mathbb{P}^{T}_{n}$.
		
		We again add subscript $j$ to refer to the $j$'th trial: $\mathbb{P}^{T_j}_{n_j}$ is the observed empirical distribution in the $j$'th trial. We further denote the observed data for participant $i$ in trial $j$ as $(X_{ji}, Z_{ji}, S_{ji}, Y_{ji})'$.
	\end{enumerate}

	Putting everything together, the observed-data distribution is denoted by $P \in \mathcal{M}$ where $\mathcal{M}$ is a model for the distribution of the observed trial-level data $(n, \mathbb{P}^T_n) \in \mathbb{N} \times \mathcal{M}^*_{\text{NP}}$. Furthermore, the observed-data distribution $P$ is determined by the full-data distribution $P^F$. This is represented by the mapping $h: \mathcal{M}^F \to \mathcal{M}$ where $h(P^F)$ is the observed-data distribution implied by $P^F$.
		
	We assume that the observed data $\left( n_j, \mathbb{P}^{T_j}_{n_j} \right)_{j=1}^N$ are $N$ i.i.d.~samples from $P_0 := h(P^F_0)$ where the subscript $0$ refers to the true distribution.
	In addition, we let $\mathbb{P}_N \in \mathcal{M}_{\text{NP}}$ be the empirical measure based on the observed trial-level data:
	\begin{equation*}
		\mathbb{P}_N(A) := \frac{1}{N} \sum_{j = 1}^N \mathbbm{1} \! \left\{ (n_j, \mathbb{P}_{n_j}^{T_j}) \in A \right\} \text{ for measurable } A \subseteq \mathbb{N} \times \mathcal{M}_{\text{NP}}^{*}.
	\end{equation*}

	\begin{remark}[random trial $T$]
		For the methods presented later on, the trial random element $T$ is superfluous because the observed data and target parameters only depend on $T$ through $P^T$ and $n$. Still, introducing $T$ clarifies the data-generating mechanism and its i.i.d.~assumption, and facilitates criticism of the meta-analysis. This matters because all ensuing inferences rely on the plausibility of the observed trials being an i.i.d.~sample from a population of trials.
	\end{remark}

	\begin{example}[i.i.d.~sampling of trials]
		In the data application in Section \ref{sec:data-application}, $P^F$ is the distribution of all phase 3 COVID-19 vaccine trials that could have been conducted in 2020--2021. Although this population is vague, it is required for inferences in the MA framework and determines where surrogacy conclusions generalize. If $S$ is a good trial-level surrogate in this population, it may justify using $S$ as a primary endpoint in future trials sampled from $P^F$, but not in trials outside this population (e.g., with different vaccine mechanisms or circulating variants).
	\end{example}

	\begin{figure}
		\centering
		\begin{tikzpicture}[edge from parent/.style={draw,-latex}]
			\tikzstyle{level 1}=[sibling distance=27.5mm, level distance = 13mm] 
			\tikzstyle{level 2}=[sibling distance=10mm, level distance = 13mm] 
			\node (full) {$P^F$}
			child {node {$T_1$}
				child {node (a1) {$P^{T_1}$} edge from parent[dashed]}
				child {node [draw] (a2) {$n_1$} edge from parent[dashed]}
			}
			child {node {$\ldots$} edge from parent[draw=none]
				child {node (b1) {$\ldots$} edge from parent[draw=none]}
			}
			child {node {$T_N$}
				child {node (c1) {$P^{T_N}$} edge from parent[dashed]}
				child {node [draw] (c2) {$n_N$} edge from parent[dashed]}
			};
			
			\coordinate (CENTERa) at ($(a1)!0.5!(a2)$);
			\coordinate (CENTERc) at ($(c1)!0.5!(c2)$);
			
			\node (a3) at (CENTERa) [yshift=-13mm, draw] {$(X_{1i}, Z_{1i}, S_{1i}, Y_{1i})_{i = 1}^{n_1}$};
			\node (b3) at (b1) [yshift=-13mm] {$\ldots$};
			\node (c3) at (CENTERc) [yshift=-13mm, draw] {$(X_{Ni}, Z_{Ni}, S_{Ni}, Y_{Ni})_{i = 1}^{n_N}$};
			
			\draw [->, -latex] (a1) -- (a3);
			\draw [->, -latex, dashed] (a2) -- (a3);
			\draw [->, -latex] (c1) -- (c3);
			\draw [->, -latex, dashed] (c2) -- (c3);
			
			\node (a4) at (a3) [yshift=-13mm, draw] {$\mathbb{P}^{T_1}_{n_1}$};
			\node (b4) at (b3) [yshift=-13mm] {$\ldots$};
			\node (c4) at (c3) [yshift=-13mm, draw] {$\mathbb{P}^{T_N}_{n_N}$};
			
			\node (empirical) at (b4) [yshift=-13mm, draw] {$P_N$};
			
			\node [right of =full, xshift=50mm] (fullmodel) {$P^F \in \mathcal{M}^F$};
			\node (distribution) at (fullmodel|-c2) {$P^{T_j} \in \mathcal{M}^*$};
			\node (trialempdistribution) at (fullmodel|-c4) {$\mathbb{P}^{T_j}_{n_j} \in \mathcal{M}^*_{NP}$};
			\node (trialempdistribution) at (fullmodel|-empirical) {$\mathbb{P}_N \in \mathcal{M}_{NP}, \; P = h(P^F) \in \mathcal{M}$};
			
			\draw [->, -latex, dashed] (a3) -- (a4);
			\draw [->, -latex, dashed] (c3) -- (c4);
			
			\draw [->, -latex, dashed] (a4) -- (empirical);
			\draw [->, -latex, dashed] (c4) -- (empirical);
		\end{tikzpicture}
		\caption{Data-generating mechanism underlying the meta-analytic framework. 
		Full arrows represent sampling.
			Dashed arrows represent mappings.  
			Objects surrounded by a rectangle represent observed data. $\mathcal{M}^F$: full-data model; $\mathcal{M}^*$: within-trial model; $\mathcal{M}^*_{\text{NP}}$: non-parametric within-trial model; $\mathcal{M}$: observed-data model; $\mathcal{M}_{\text{NP}}$: non-parametric observed-data model.}\label{fig:sampling-mechanism-diagram}
	\end{figure}
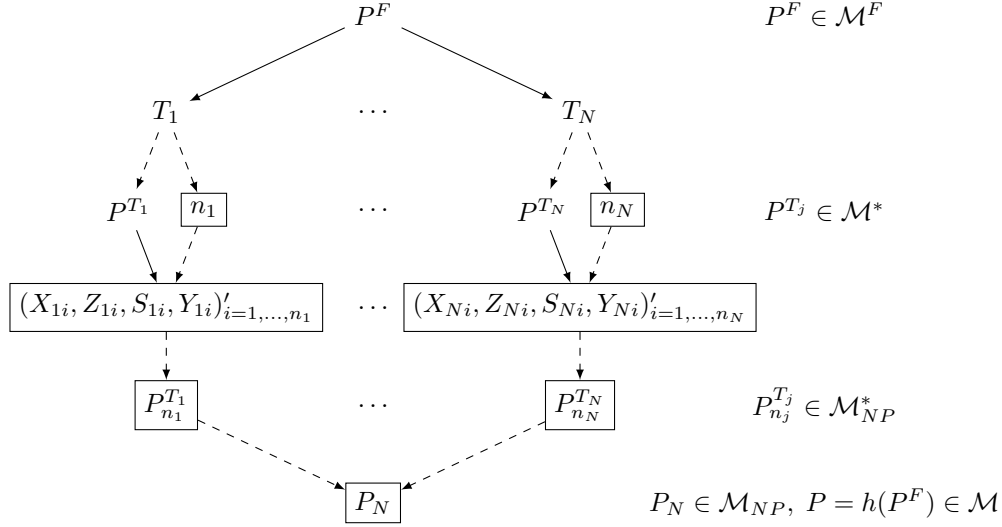

	\subsubsection{Trial-Level Treatment Effects}\label{sec:trial-level-treatment-effects}
	
	Two trial-level random variables are of special interest in a meta-analysis: the (unobserved) true trial-level treatment effects and the (observed) estimated trial-level treatment effects. Figure \ref{fig:trial-level-sampling-mechanism-diagram} illustrates how these two variables fit into the data-generating mechanism. 
	We extend the canonical MA framework by considering trial-level treatment effects on $g(X, S)$ for an arbitrary function $g: \mathcal{X} \times \mathcal{S} \to \mathbb{R}$, where $g(X, S)$ is termed the transformed surrogate. The canonical MA framework is a special case where $g(X, S) = S$ (and typically $\mathcal{S} \subset \mathbb{R}$). 
	Below, we introduce functionals and variables indexed by $g$; we omit the index $g$ if $g(X, S) = S$.
	
	\paragraph{True Treatment Effects}

	We let $\alpha^g: \mathcal{M}^* \to \mathbb{R}$ and $\beta: \mathcal{M}^* \to \mathbb{R}$ be functionals that map the within-trial distribution $P^T$ to the true treatment effects on the transformed surrogate and the clinical endpoint, respectively. In what follows, we will consider functionals of the following form, which only depend on conditional first moments of the within-trial distribution:
	\begin{equation}\label{eq:alpha-beta}
		\begin{split}
			\alpha^g(P^T) & := f_{\alpha} \left[ E_{P^T} \! \left\{ g(X, S) \mid Z = 1 \right\} , E_{P^T} \! \left\{ g(X, S) \mid Z = 0 \right\} \right], \text{ and} \\
			\beta(P^T) & := f_{\beta} \left[ E_{P^T} \! \left\{ Y \mid Z = 1 \right\}, E_{P^T} \! \left\{ Y \mid Z = 0 \right\} \right],
		\end{split}
	\end{equation}
	where $f_{\alpha}, f_{\beta}: \mathbb{R}^2 \to \mathbb{R}$ are contrast functions and $E_{P^T}$ is the expectation under $P^T$.
	
	Further on, we will interchangeably use $\alpha^g$ and $\beta$ to refer to the above functionals or to the corresponding random variables $\alpha^g := \alpha^g(P^T)$ and $\beta := \beta(P^T)$ for $T \sim P^F$. Note that these random variables are not observable.
	
	
	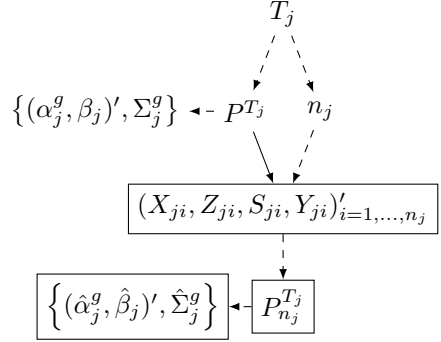
\begin{figure}
		\centering
		\begin{tikzpicture}[edge from parent/.style={draw,-latex}]
			\tikzstyle{level 1}=[sibling distance=10mm, level distance = 13mm] 
			\node (Tj) {$T_j$}
			child {
				node (PTj) {$P^{T_j}$} edge from parent[dashed]
			}
			child{
				node (nj) {$n_j$} edge from parent[dashed]
			};
			
			\coordinate (center) at ($(PTj)!0.5!(nj)$);
			
			\node (data) at (center) [yshift=-13mm, draw] {$(X_{ji}, Z_{ji}, S_{ji}, Y_{ji})'_{i = 1, \dots, n_j}$};
			\node (empirical) at (data) [yshift=-13mm, draw] {$\mathbb{P}^{T_j}_{n_j}$};
			
			\draw [->, -latex] (PTj) -- (data);
			\draw [->, -latex, dashed] (nj) -- (data);
			
			\node (alphabeta) [left of =PTj, xshift=-10mm] {$\left\{ (\alpha^g_j, \beta_j)', \Sigma^g_j \right\}$};
			\node (alphabetahat) [left of =empirical, xshift=-10mm, draw] {$\left\{ (\Hat{\alpha}^g_j, \Hat{\beta}_j)', \Hat{\Sigma}^g_j \right\}$};
			
			\draw [->, -latex, dashed] (data) -- (empirical);
			
			\draw [->, -latex, dashed] (PTj) -- (alphabeta);
			\draw [->, -latex, dashed] (empirical) -- (alphabetahat);
		\end{tikzpicture}
		\caption{True and estimated treatment effects as functionals of the true and empirical within-trial distributions, respectively. 
			Dashed arrows represent mappings. Full arrows represent sampling. 
			Objects surrounded by a rectangle represent the observed data.
			The dependence of functionals on the sample size is ommitted for clarity.}\label{fig:trial-level-sampling-mechanism-diagram}
	\end{figure}

	\paragraph{Estimated Treatment Effects}
	
	We let $\Hat{\alpha}^g: \mathbb{N} \times \mathcal{M}^*_{\text{NP}} \to \mathbb{R}$ and $\Hat{\beta}: \mathbb{N} \times \mathcal{M}^*_{\text{NP}} \to \mathbb{R}$ be functionals that map the empirical within-trial distribution $\mathbb{P}^{T}_{n}$ to the estimated treatment effects on the transformed surrogate and the clinical endpoint, respectively. 
	As before, we interchangeably use $\Hat{\alpha}^g$ and $\Hat{\beta}$ to refer to the functionals or to the corresponding random variables $\Hat{\alpha}^g := \Hat{\alpha}^g(n, \mathbb{P}^{T}_{n})$ and $\Hat{\beta} := \Hat{\beta}(n, \mathbb{P}^{T}_{n})$ for $(n, \mathbb{P}^T_n) \sim P$. Note that these random variables are observable. 

	A related mapping is $\Hat{\Sigma}^g: \mathbb{N} \times \mathcal{M}^*_{\text{NP}} \to \mathbb{R}^{2 \times 2}$, which maps the empirical within-trial distribution to the estimated covariance matrix of the root-$n$ standardized estimated treatment effects on the transformed surrogate and the clinical endpoint. As before, we interchangeably use $\Hat{\Sigma}^g$ to refer to the mapping or to the corresponding random matrix $\Hat{\Sigma}^g := \Hat{\Sigma}^g(n, \mathbb{P}^{T}_{n})$ for $(n, \mathbb{P}^T_n) \sim P$. Note that this random matrix is also observable, but that this is the \textit{estimated} covariance matrix and not the true covariance matrix, which we denote by $\Sigma^g: \mathbb{N} \times \mathcal{M}^* \to \mathbb{R}^{2 \times 2}$ and define as 
	\begin{equation*}
		\Sigma^g(n, P^T) := n \cdot \operatorname{Var}_{P^T, n} \! \left(\Hat{\alpha}^g, \Hat{\beta}  \right) ,
	\end{equation*}
	where $\operatorname{Var}_{P^T, n}$ is the sampling variance of the estimators when $n$ i.i.d.~samples are drawn from the true within-trial distribution $P^T$.
	The distinction between $\Hat{\Sigma}^g$ and $\Sigma^g$ will become important later on.

	
	\subsection{Models}

	We will work under a non-parametric observed-data model $\mathcal{M}$ unless stated otherwise. Hence, we will not make any testable assumptions, but we will need to make untestable identifying assumptions to identify certain target parameters. Although these assumptions cannot be tested with the observed data, some can be argued to hold approximately using known properties of certain estimators. 
	By contrast, in the canonical MA framework, the observed-data model $\mathcal{M}$ is often parametric and the target parameters can be identified without making additional identifying assumptions, as shown in Example \ref{example:parametric-trial-level-model} below. Semi-parametric models are also common in the canonical MA framework, where the identification relies on additional identifying assumptions, as shown in Example \ref{example:burzykowski} below.
	
	\begin{example}[linear mixed model]\label{example:parametric-trial-level-model}
		Consider a setting without baseline covariates and $\mathcal{S} = \mathbb{R}$. One could assume that the data in each trial are bivariate normal:
		\begin{equation*}
			(S_{ji}, Y_{ji})' \mid Z_{ji} = z \sim \mathcal{N}(\boldsymbol{\mu}_{j}^z, C_j),
		\end{equation*}
		where $\boldsymbol{\mu}_{j}^z$ is the trial-specific mean vector for treatment $z$ and $C_j$ is the trial-specific covariance matrix. 
		The within-trial model $\mathcal{M}^*$ is parametric here. 
		One could also make parametric assumptions about the distribution of the trial-specific parameters, for example,
		\begin{equation*}
			\begin{pmatrix}
				\boldsymbol{\mu}_j^0 \\
				\boldsymbol{\mu}_j^1
			\end{pmatrix}
			\sim \mathcal{N}(\boldsymbol{\theta}, D) \quad \text{and} \quad C_j = C.
		\end{equation*}
		Here, $\mathcal{M}$ is a parametric model with parameters $\left\{ \boldsymbol{\theta}, D, C \right\}$. This is in fact a linear mixed model, with random effects $\boldsymbol{\mu}_j^0$ and $\boldsymbol{\mu}^1_j$. This model has been proposed for evaluating surrogate endpoints \parencite{alonso2016applied, buyse2000validation}. The distribution of $(\alpha, \beta)'$, defined as mean differences, follows immediately as $(\alpha_j, \beta_j)' = \boldsymbol{\mu}_j^1 - \boldsymbol{\mu}_j^0 \sim \mathcal N\left( A \boldsymbol{\theta}, A D A^\top \right)$ for $A$ an appropriate contrast matrix.
	\end{example}


	\begin{example}[two-stage approach]\label{example:burzykowski}
		\textcite{burzykowski2001validation} assume that $\mathcal{M}^*$ is a semi-parametric Cox proportional hazards model.
		Let $(\Hat{\alpha}_j, \Hat{\beta}_j)'$ be the estimated log hazard ratios with corresponding estimated covariance matrix $\Hat{\Sigma}_j$.  \textcite{burzykowski2001validation} further assume that the true trial-level log hazard ratios $(\alpha_j, \beta_j)'$ are bivariate normal: $(\alpha_j, \beta_j)' \sim \mathcal N(\boldsymbol{\theta}, D)$.
		The parameters $\boldsymbol{\theta}$ and $D$ are not identified from the observed data without additional assumptions, however. \textcite{burzykowski2001validation}, therefore, assume that the sampling distribution $(\Hat{\alpha}_j - \alpha_j, \Hat{\beta}_j - \beta_j)'$ in a given trial $j$ is bivariate normal with mean zero and covariance matrix $\Sigma_j$, where additionally $\Hat{\Sigma}_j = \Sigma_j$. 
	\end{example}


	\subsection{Parameters and Estimands}

	As in the causal-inference and missing-data literature, we distinguish between full-data parameters as functionals of the full-data distribution, and statistical parameters as functionals of the observed-data distribution \parencite{hernan}. The values obtained by evaluating these functionals in the true distributions ($P^F_0$ and $P_0$, respectively) are called estimands.
	Full-data parameters are of primary interest but are not identified with the observed data. Therefore, untestable identifying assumptions are needed to equate statistical parameters with full-data parameters.
	After making these identifying assumptions, estimation of the full-data parameters becomes a statistical problem.

	We discuss general full-data and statistical parameters next. The specific parameters targeted by our methods will be introduced in Section \ref{sec:ma-surrogate-index}.

	\subsubsection{Full-Data Parameters}

	A full-data parameter is represented by the following mapping: $\Psi^F: \mathcal{M}^F \to \mathbb{R}^p$.
	The full-data parameters in meta-analyses only depend on the full-data distribution $P^F$ through its implied distribution of the trial-level treatment effects. For instance, the variance of $\beta$ is a measure of between-trial heterogeneity:
	\begin{equation}\label{eq:full-data-estimand-variance}
		\Psi^F(P^F) = E_{P^F} \! \left\{\beta - E_{P^F} \beta \right\}^2.
	\end{equation}
	Note that the above parameter is well defined under a non-parametric full-data model whenever $E_{P^F}\beta^2 < \infty$. Other parameters may only be defined under parametric assumptions, as illustrated in the next example.

	\begin{example}
		Assuming the linear mixed model from Example \ref{example:parametric-trial-level-model}, a full-data parameter could be the parameters of the random effects distribution: $\Psi^F(P^F) = (\boldsymbol{\theta}, D)$. Note that $\boldsymbol{\theta}$ and $D$ are implicitly functionals of $P^F$ and the parameter $\Psi^F$ is only defined under the parametric full-data model $\mathcal{M}^F$ that satisfies the linear mixed model assumptions.
	\end{example}

	\subsubsection{Statistical Parameters}

	A statistical parameter is represented by the following mapping: $\Psi: \mathcal{M} \to \mathbb{R}^p$.
	Most statistical parameters are only meaningful if they equal a relevant full-data parameter. This typically requires identifying assumptions (i.e., restrictions on the full-data model $\mathcal{M}^F$). More specifically, we say that the full-data parameter $\Psi^F$ is identified by the observed-data distribution if there exists a statistical parameter such that $\Psi^F(P^F) = \Psi(h(P^F))$ for all $P^F \in \mathcal{M}^F$.
	Continuing with the full-data parameter in (\ref{eq:full-data-estimand-variance}), we let $\hat{\sigma}_{\Hat{\beta}}^2$ be the element of $\Hat{\Sigma}^g$ coresponding to $\Hat{\beta}$. Then, the following statistical parameter:
	\begin{equation}\label{eq:statistical-estimand-variance}
		\Psi(P) = E_{P} \! \left\{ \Hat{\beta} - E_{P} \Hat{\beta} \right\}^2 - E_{P} \left( n^{-1} \hat{\sigma}_{\Hat{\beta}}^2 \right)
	\end{equation}
	identifies the full-data parameter if the following identifying assumptions hold for all $P^F \in \mathcal{M}^F$:
	\begin{equation*}
		E_{P^F}\beta = E_{h(P^F)} \Hat{\beta} \quad \text{and} \quad E_{P^F} \! (\Hat{\beta} - \beta)^2 = E_{h(P^F)} \! \left( n^{-1} \hat{\sigma}_{\Hat{\beta}}^2 \right),
	\end{equation*}
	which follows from adding and subtracting $\beta$ inside the first expectation in (\ref{eq:statistical-estimand-variance}).

	\section{Meta-Analysis with an (Optimal) Transformed Surrogate}\label{sec:ma-surrogate-index}

	The main measure of surrogacy in the canonical MA framework is the (squared) Pearson correlation between the treatment effects on the surrogate and the clinical endpoint. This is the \textit{trial-level correlation}, which quantifies the \textit{trial-level surrogacy} \parencite[see, e.g.,][]{sargent2005disease, weir2022informed, buyse2008individual}. 
	
	In the previous section, we introduced the treatment effect on the transformed surrogate, $\alpha^g$, for an arbitrary function $g: \mathcal{X} \times \mathcal{S} \to \mathbb{R}$. 
	One could define the optimal function in terms of best performance as a trial-level surrogate as the function that maximizes the correlation between $\alpha^g$ and $\beta$: $g_{\text{opt}} := \arg \max_g \text{Corr}(\alpha^g, \beta)$ where $\text{Corr}(\alpha^g, \beta)$ is the Pearson correlation between $\alpha^g$ and $\beta$.
	However, implementing this is not practically feasible because the number of independent trials is usually very small (e.g., there are six independent trials in our application) and, hence, overfitting will occur. 
	An alternative approach to finding a (near-)optimal $g$ is needed.

	As mentioned in the introduction, \textcite{athey2025surrogate} proposed a method that allows for inference about the treatment effect on $Y$ in a given trial where only $(X, Z, S)$ were observed. This method relies on two key identifying assumptions, which we will informally translate to the MA framework: (i) surrogacy and (ii) comparability. 
	We only present these assumptions and the corresponding identification arguments informally because they serve as a justification of our choice for an ``optimal" $g$, but our approach remains valid (but may no longer be optimal) under violations of surrogacy and comparability.
	\begin{assumption}[surrogacy]
		$\forall \; T \in \mathbb{T}: Y \perp Z \mid X, S \quad \text{for} \quad (X, Z, S, Y) \sim P^T$.
	\end{assumption}
	The comparability assumption, stated below, involves the surrogate index $g_0$ which we define as
	\begin{equation*}
		g_0(x, s) := E \! \left\{ w(x, s; T) E_{P^T}  (Y \mid X = x, S = s )\right\},
	\end{equation*}
	where $w(x, s; T)$ is an arbitrary trial-level weight function such that $E \! \left\{ w(x, s; T) \right\} = 1$ for all $(x, s) \in \mathcal{X} \times \mathcal{S}$.
	The surrogate index defined above is thus a pooled regression function where the pooling across trials and treatment groups depends on the weight function. 
	Two special cases are $w(x, s; T) = 1$ for all $T$ and $w(x, s; T) = n(T) / E\{n(T)\}$. The former corresponds to pooling the individual-participant data across all trials while reweighting participants inversely by their trial size (i.e., each trial receives the same weight), whereas the latter corresponds to pooling without reweighting (i.e., each participant receives the same weight).
	In the remainder of this paper, we consider pooling without reweighting.
	\begin{assumption}[comparability]
		$\forall \; T \in \mathbb{T}: E_{P^T}(Y \mid X = x, S = s) = g_0(x, s)$ for $P^T$ almost every $(x, s)$.
	\end{assumption}
	Under these two assumptions, the mean of $Y$ under treatment $z$ in an arbitrary trial $T$ can be identified from the distribution of $g_0(S, X)$ in the same trial:
	\begin{align*}
		E_{P^T} \! \left\{ Y \mid Z = z \right\} & = E_{P^T} \! \left\{ E_{P^T}(Y \mid X, Z = z, S) \mid Z = z \right\} &&~\mbox{(iterated expectation)} \\
		& = E_{P^T} \! \left\{ E_{P^T}(Y \mid X, S) \mid Z = z  \right\} &&~\mbox{(surrogacy)} \\
		& = E_{P^T} \! \left\{ g_0(X, S) \mid Z = z \right\} &&~\mbox{(comparability)}
	\end{align*}
	If surrogacy and comparability hold, $g_0$ is the optimal function in the sense that $\text{Corr}(\alpha^{g_0}, \beta) = 1$ for treatment effect measures that only depend on the first conditional moments as in (\ref{eq:alpha-beta}).
	Even if surrogacy and comparability do not hold---which is expected in most applications---$g_0(S, X)$ is well defined and is expected to be a good summary of $(X, S)$ as long as surrogacy and comparability hold approximately.
	The key difference between $g_{\text{opt}}$ defined before and $g_0$, however, is that $g_0$ can be estimated by regressing $Y$ onto $X$ and $S$ by pooling the individual-participant data from all trials. Overfitting will be less of an issue because the number of participant-level observations is much larger than the number of trials. 
	Nevertheless, some caution regarding overfitting is warranted; see Appendix \ref{appendix:overfitting-of-the-surrogate-index-estimator} for further discussion.

	In the remainder of this section, we first discuss full-data and statistical parameters, and corresponding identifying assumptions. 
	We then discuss inference about these parameters when (i) $g$ is known a priori and (ii) when $g$ has been estimated with the observed data $\left( n_j, \mathbb{P}_{n_j}^{T_j} \right)_{j = 1}^{N}$.

	\subsection{Parameters and Estimands}

	For every $g$, the parameters of interest are features of the distribution of $(\alpha^g, \beta)$. Since this distribution is indexed by $g$, the parameters are themselves also indexed by $g$ and the question arises which parameter is most relevant.
	We divide the set of parameters into three classes according to the nature of $g$:
	First, $g$ may be selected a priori (e.g., estimated with external data). We discuss this as a starting point for the more difficult data-adaptive parameters.  
	Second, one may be interested in the trial-level surrogacy of the unknown surrogate index $g_0$. We will not target such parameters because (i) inference about parameters indexed by $g_0$ is only possible under strong conditions on the rate at which the \textit{estimated surrogate index}, denoted by $g_N$, converges to $g_0$ and (ii) the practical relevance of this parameter is questionable since $g_0$ will remain unknown.
	Third, one may be interested in the trial-level surrogacy of the estimated surrogate index $g_N$. We prefer this parameter over the previous one because (i) inference is possible under weaker conditions and (ii) $g_N$ is estimated from the data and $g_N(X, S)$ is thus usable as a surrogate endpoint.

	\subsubsection[A Priori Selected g]{A Priori Selected $g$}\label{sec:full-data-estimands-a-priori-g}

	
	The canonical MA framework assumes a bivariate normal model for $(\alpha^g, \beta)'$ where $g(X, S) = S \in \mathbb{R}$. While this could hold for a single fixed $g$, extending it to the current setting would require bivariate normality for every $g \in \mathcal{G}$, which is impossible when $\mathcal{G}$ is sufficiently large. 
	We therefore avoid parametric assumptions on the distribution of $(\alpha^g, \beta)'$. The bivariate normal model remains useful as a reference because its parameters are easy to interpret.

	Following \textcite{vansteelandt2022assumption}, we define full-data parameters that coincide with the bivariate normal model parameters when those assumptions hold, while remaining interpretable when they do not.
	For a given $g$, this yields the following full-data parameters:
	\begin{equation}\label{eq:full-data-estimand}
		\begin{pmatrix}
			\Psi^{F, g}_{1}(P^F) \\
			\Psi^{F, g}_{2}(P^F)
		\end{pmatrix} =
		\begin{pmatrix}
			E_{P^F} \alpha^g \\
			E_{P^F} \beta 
		\end{pmatrix} \quad \text{and} \quad
		\begin{pmatrix}
			\Psi^{F, g}_{3}(P^F) & \Psi^{F, g}_{5}(P^F) \\
			\Psi^{F, g}_{5}(P^F) & \Psi^{F, g}_{4}(P^F)
		\end{pmatrix} =
		\operatorname{Var}_{P^F} \left\{ \begin{pmatrix}
			\alpha^g \\
			\beta
		\end{pmatrix}\right\},
	\end{equation}
	where we let $\Psi^{F, g}(P^F) := (\Psi^{F, g}_{1}(P^F), \dots, \Psi^{F, g}_{5}(P^F))'$.
	Note that we are mainly interested in the trial-level correlation between $\alpha^g$ and $\beta$. However, it is more convenient to work with the above target parameters because they are very simple functionals of the full-data distribution and the trial-level correlation is a function of these parameters anyhow. 

	
	Since we do not observe $(\alpha^g, \beta)'$ directly, the full-data parameters are not identifiable without identifying assumptions. 
	We further let $\mathcal{G}$ be a set of real-valued functions on $\mathcal{X} \times \mathcal{S}$ and let $E_{P^T, n}$ be the expectation over the empirical distributions $\mathbb{P}^T_n$ generated by sampling $n = n(T)$ i.i.d.~observations from $P^T$.
	Under the following two identifying assumptions, the full-data parameters are identified with the observed-data distribution.
	\begin{assumption}[unbiased estimation of trial-level treatment effects]\label{assumption:unbiased-trt-effect}
		The trial-level treatment effect estimators are unbiased in the sense that $\forall \; T \in \mathbb{T}, g \in \mathcal{G}: \; E_{P^T, n} \! \left\{ (\Hat{\alpha}^g, \Hat{\beta})' \right\} = \left(\alpha^g(P^T), \beta(P^T)\right)'$.
	\end{assumption}
	\begin{assumption}[unbiased estimation of within-trial sampling variability]\label{assumption:unbiased-sampling-variance}
		The sampling variability of the trial-level treatment effect estimators is unbiasedly estimated in the sense that $\forall \; T \in \mathbb{T}, g \in \mathcal{G}: \; E_{P^T, n}(\Hat{\Sigma}^g ) = \Sigma^g(n, P^T)$.
	\end{assumption}	
	\begin{restatable}{proposition}{propositionidmeanvariance}
		\label{proposition:identification-mean-variance}
		Define the statistical parameter $\Psi^g := (\Psi^g_1, \dots, \Psi^g_5)'$ as follows:
		\begin{equation}\label{eq:identification-mean-variance}
			\begin{pmatrix}
				\Psi^g_1(P) \\
				\Psi^g_2(P)
			\end{pmatrix} = 
			\begin{pmatrix}
				E_P \Hat{\alpha}^g \\
				E_P \Hat{\beta}
			\end{pmatrix} \quad \text{and} \quad
			\begin{pmatrix}
				\Psi^g_3(P) & \Psi^g_5(P) \\
				\Psi^g_5(P) & \Psi^g_4(P)
			\end{pmatrix} =
			\operatorname{Var}_P \left\{ \begin{pmatrix}
				\Hat{\alpha}^g \\
				\Hat{\beta}
			\end{pmatrix}\right\} - E_P \left( n^{-1} \Hat{\Sigma}^g \right),
		\end{equation}
		where $P \in \mathcal{M}$ is an observed-data distribution. 
		Under Assumptions \ref{assumption:unbiased-trt-effect} and \ref{assumption:unbiased-sampling-variance}, the full-data parameter (\ref{eq:full-data-estimand}) is identified in the sense that $\Psi^{F, g}(P^F) = \Psi^g(h(P^F))$ for all $P^F \in \mathcal{M}^F$ and $g \in \mathcal{G}$. 
	\end{restatable}
	Note that Assumptions \ref{assumption:unbiased-trt-effect} and \ref{assumption:unbiased-sampling-variance} are stated for a set of possible transformations $\mathcal{G}$. If one is interested in only one specific transformation $g$, we can set $\mathcal{G} = \{g\}$. Later on when the transformation will be estimated, these assumptions will be required to hold for $\mathcal{G}$ being a large class of functions. 

	\begin{remark}[relaxing unbiasedness assumptions]
		Assumptions \ref{assumption:unbiased-trt-effect} and \ref{assumption:unbiased-sampling-variance} impose strict unbiasedness on the trial-level treatment effect estimators. This can be weakened by adopting an alternative asymptotic regime in which within-trial sample sizes increase with the number of sampled trials $N$ (Appendix \ref{appendix:alternative-asymptotic-framework}). This change does not alter the estimation procedure. It only replaces exact unbiasedness in Assumptions \ref{assumption:unbiased-trt-effect} and \ref{assumption:unbiased-sampling-variance} with the requirement that within-trial biases converge to zero sufficiently fast as the within-trial sample size increases. In most meta-analyses, where $n \gg N$, these rate conditions are mild.
		Appendix \ref{appendix:simulation-results} includes simulations that illustrate this point.
	\end{remark}

	\subsubsection[Estimated Surrogate Index gN]{Estimated Surrogate Index $g_N$}

	The parameters introduced above are indexed by $g \in \mathcal{G}$, and as argued before, the surrogate index $g_0$ is a good choice for $g$. However, $g_0$ is unknown and needs to be estimated.
	If the surrogate index is estimated with the observed data, the target parameters become data-adaptive: These parameters are indexed by the estimated surrogate index $g_N$: $\Psi^{F, g_N}$ and $\Psi^{g_N}$, where $\Psi^{F, g_N}$ is the full-data parameter and $\Psi^{g_N}$ is the corresponding statistical parameter defined as before but with $g$ replaced by $g_N$.

	To obtain valid inferences about the data-adaptive parameter $\Psi^{g_N}$, we will impose Donsker conditions on $g_N$ in Section \ref{sec:inference-data-adaptive-estimand}.
	This limits the complexity of the allowed surrogate index estimators even though many machine learning methods can still be used. Methods for inference about data-adaptive target parameters that avoid Donsker conditions through cross-fitting exist (see, e.g., \textcite[Section 7]{van2015targeted}), but they cannot target $\Psi^{g_N}$. Instead, they target different types of data-adaptive parameters that are more difficult to interpret as the cross-fitting scheme itself defines the parameter \parencite{hubbard2016statistical}. We, therefore, do not further consider cross-fitting.

	\subsection{Estimators}

	\subsubsection{Plug-In Estimator}

	We consider the plug-in estimator for the (non-data-adaptive) statistical parameter $\Psi^g$ in (\ref{eq:identification-mean-variance}), denoted by $\Hat{\Psi}^g(P) = (\Hat{\Psi}_1^g(P), \dots, \Hat{\Psi}_5^g(P))'$ for $P \in \mathcal{M}_{\text{NP}}$. We consider the same estimator for the data-adaptive parameter $\Psi^{g_N}$ by simply replacing $g$ with $g_N$.
	This estimator, evaluated in the empirical distribution $\mathbb{P}_N$, has the following closed-form expression:
	\begin{equation}\label{eq:plug-in-estimator}
		\begin{pmatrix}
			\Hat{\Psi}_1^g(\mathbb{P}_N) \\
			\Hat{\Psi}_2^g(\mathbb{P}_N)
		\end{pmatrix} = 
		\begin{pmatrix}
			E_{\mathbb{P}_N} \Hat{\alpha}^g \\
			E_{\mathbb{P}_N} \Hat{\beta}
		\end{pmatrix} \quad \text{and} \quad
		\begin{pmatrix}
			\Hat{\Psi}_3^g(\mathbb{P}_N) & \Hat{\Psi}_5^g(\mathbb{P}_N) \\
			\Hat{\Psi}_5^g(\mathbb{P}_N) & \Hat{\Psi}_4^g(\mathbb{P}_N)
		\end{pmatrix} =
		\operatorname{Var}_{\mathbb{P}_N}\left\{ \begin{pmatrix}
			\Hat{\alpha}^g \\
			\Hat{\beta}
		\end{pmatrix}\right\} - E_{\mathbb{P}_N} \! \left( n^{-1} \Hat{\Sigma}^g \right).					
	\end{equation}
	This plug-in estimator solves the estimating equation $N^{-1} \sum_{j=1}^N \phi_{\boldsymbol{\theta}}^g(n_j, \mathbb{P}^{T_j}_{n_j}) = \boldsymbol{0}$, where $\phi_{\boldsymbol{\theta}}^g: \mathbb{N} \times \mathcal{M}_{\text{NP}}^* \to \mathbb{R}^5$ is indexed by the parameter $\boldsymbol{\theta}$. 
	In the simulations and data application, a finite-sample adjusted version of the plug-in estimator is used, where the variance of $(\Hat{\alpha}^g, \Hat{\beta})$ is estimated by dividing by $N - 1$ instead of $N$.
	Further details on the estimator and its estimating functions are given in Appendix \ref{appendix:plug-in-estimator}. For a non-data-adaptive target (fixed $g$), the estimator is asymptotically linear and inference can rely on the sandwich variance estimator. In the next section, we show that inference is also possible for the data-adaptive parameter $\Psi^{g_N}$ under conditions on $g_N$.


	\subsubsection{Efficiency}

	Assumptions \ref{assumption:unbiased-trt-effect} and \ref{assumption:unbiased-sampling-variance} restrict the full-data model $\mathcal{M}^F$ and, therefore, the induced observed-data model $\mathcal{M} := \{ h(P^F) : P^F \in \mathcal{M}^F \}$. Under regularity conditions, the model $\mathcal{M}$ is locally non-parametric at the true $P_0 = h(P^F_0)$. Consequently, any two regular asymptotically linear estimators are asymptotically equivalent and attain the semiparametric efficiency bound \parencite{kennedy2017semiparametric}. In particular, the plug-in estimator \eqref{eq:plug-in-estimator} is efficient and any regular asymptotically linear estimator will be asymptotically equivalent to it.

	\begin{restatable}{proposition}{propositionnonparametricmamodel}
		\label{proposition:non-parametric-ma-model}
		Under Assumptions \ref{assumption:unbiased-trt-effect} and \ref{assumption:unbiased-sampling-variance}, fix $g$ and assume that
		\begin{enumerate}
			\item[(i)] the support of $(\alpha^g, \beta) \mid (\Sigma^g, n)$ contains a nonempty open set in $\mathbb{R}^2$ $P^F_0$-a.s., and
			\item[(ii)] $\Sigma^g$ is positive definite $P^F_0$-a.s.
		\end{enumerate}
		Then the induced model for the observed data $\bigl(\Hat{\alpha}^g, \Hat{\beta}, \Hat{\Sigma}^g, n\bigr)$ is locally non-parametric at $P_0 = h(P^F_0)$.
	\end{restatable}

	Appendix \ref{appendix:semi-parametric-efficiency-analysis} further discusses how this result implies that intuitive weighting schemes (e.g., inverse-variance weighting) are valid only under additional restrictions on the observed-data model.

	\section{Inference About Trial-Level Surrogacy of the Estimated Surrogate Index}\label{sec:inference-data-adaptive-estimand}

	In this section, we discuss inference for data-adaptive target parameters, with particular focus on $\Psi^{g_N}$.
	We first introduce general notation that is not specific to the MA framework and summarize relevant existing results. We then apply these results, together with a new technical theorem in Appendix \ref{appendix:discussion-conditions-theorem-bootstrap-for-data-adaptive-estimands}, to the MA framework and the estimated surrogate index.

	\subsection{Notation}\label{sec:notation-inference}

	We first introduce general notation that is used throughout this section.
	Let $P_0 \in \mathcal{M}$ denote the true observed-data distribution and $\mathbb{P}_N \in \mathcal{M}_{\text{NP}}$ the empirical distribution based on $N$ i.i.d.~observations $O \in \Omega$ from $P_0$.
	For any integrable $f: \Omega \to \mathbb{R}$, we write $P_0 f := \int_{\Omega} f(o) dP_0(o)$ and $\mathbb{P}_N f := \int_{\Omega} f(o) dP_N(o)$. Convergence in distribution is denoted by $\overset{d}{\to}$.
	Let $g \in \mathbb{D}$ index the statistical parameter $\Psi^g: \mathcal{M} \to \mathbb{R}^p$, where $(\mathbb{D}, d)$ is a metric space.
	Let $g_N \in \mathbb{D}$ be data-adaptive (i.e., a function of $\mathbb{P}_N$); the corresponding target parameter is $\Psi^{g_N}$.


	\subsection{General Data-Adaptive Target Parameters}\label{sec:inference-estimated-g}

	Theorem 3 from \textcite{hubbard2016statistical}, reproduced below with slight changes in notation, shows when and how asymptotically valid inferences are possible for general data-adaptive target parameters $\Psi^{g_N}$. 
	\begin{theorem}\label{theorem:asymptotic-linearity-data-adaptive-estimand}
		Assume $\Hat{\Psi}^g(P_N)$ is an asymptotically linear estimator of $\Psi^g(P_0)$ at $P_0$ with influence function $IC^g(P_0): \Omega \to \mathbb{R}$ \textit{uniformly in the choice of the parameter} $g$ in the following sense:
		\begin{equation*}
			\Hat{\Psi}^{g_N}(\mathbb{P}_N) - \Psi^{g_N}(P_0) = (\mathbb{P}_N - P_0) IC^{g_N}(P_0) + R_N
		\end{equation*}
		where $R_N = o_P(N^{-1 / 2})$. In addition, assume $P_0 \left\{ IC^{g_N}(P_0) - IC^{g_0}(P_0) \right\}^2 = o_P(1)$ and $IC^{g_N}(P_0)$ is an element of a $P_0$-Donsker class with probability tending to 1. Then,
		\begin{equation}\label{eq:uniform-asymptotic-linearity}
				\Hat{\Psi}^{g_N}(\mathbb{P}_N) - \Psi^{g_N}(P_0) = (P_N - P_0) IC^{g_0}(P_0) + o_P(N^{-1 / 2})
		\end{equation}
		and thus $N^{1 / 2} \left\{ \Hat{\Psi}^{g_N}(\mathbb{P}_N) - \Psi^{g_N}(P_0) \right\}$ is asymptotically normally distributed with mean zero and variance $\sigma^2 = P_0 IC^{g_0}(P_0)^2$.
	\end{theorem}
	In each specific application with data-adaptive target parameters, one must verify that the conditions of this theorem hold. This is straightforward when the target parameter is a mean of a data-adaptively determined function as then $R_N = 0$, but it is more difficult for general data-adaptive target parameters.
	Theorem \ref{theorem:bootstrap-for-data-adaptive-estimands} in Appendix \ref{appendix:discussion-conditions-theorem-bootstrap-for-data-adaptive-estimands} specializes Theorem \ref{theorem:asymptotic-linearity-data-adaptive-estimand} to parameters defined as solutions to estimating equations without nuisance parameters. The resulting conditions are easier to verify in practice and, moreover, Theorem \ref{theorem:bootstrap-for-data-adaptive-estimands} includes a bootstrap result.

	\subsection{Plug-in Estimator for Data-Adaptive Target Parameter in MA Framework}

	Proposition \ref{proposition:sufficiency-conditions-theorem} in Appendix \ref{appendix:discussion-conditions-theorem-bootstrap-for-data-adaptive-estimands} provides a set of conditions (further referred to as conditions (a--f)) that are sufficient for the conditions in Theorem \ref{theorem:bootstrap-for-data-adaptive-estimands} to hold for the plug-in estimator (\ref{eq:plug-in-estimator}) of the data-adaptive target parameter $\Psi^{g_N}$ defined in (\ref{eq:identification-mean-variance}).
	We discuss the key conditions from Appendix \ref{sec:plug-in-estimator-statistical-estimand} only briefly here. See Appendix \ref{sec:plug-in-estimator-statistical-estimand} for the full set of conditions and a more detailed discussion.

	Condition (a) requires that the estimated surrogate index $g_N$ converges to a fixed function in the supremum norm.
	\begin{enumerate}
		\item[(a)] There exists a $g^* \in \mathcal{G}$ such that $\lVert g_N - g^* \rVert_{\infty} \overset{P}{\to} 0$ and $g_N \in \mathcal{G}$ with probability tending to one. 
	\end{enumerate}
	Though weak and satisfied by many regression methods, this condition has important finite-sample implications for overfitting and positivity violations, discussed in Appendix \ref{appendix:overfitting-of-the-surrogate-index-estimator}. 

	Conditions (c--d) restrict the complexity of the trial-specific treatment effect estimators and the estimator for $g_N$: $\Hat{\alpha}^g$ should be smooth in $g$ and $\mathcal{G}$ should not be too complex. 
	Conditions (c--d) can be shown using many strategies and are for instance satisfied when using parametric estimators for $g_N$. Otherwise, if we impose mild smoothness conditions on the estimator $(n, \mathbb{P}^T_n) \mapsto \Hat{\alpha}^g(n, \mathbb{P}^T_n)$, conditions (c-d) are satisfied whenever $\mathcal{G}$ is a $P_{a, 0}$-Donsker class, where $P_{a, 0}$ is the distribution on $\mathcal{X} \times \mathcal{S}$ induced by sampling $(n, \mathbb{P}^T_n)$ from $P_0$ and then sampling a random value from $\mathbb{P}^T_n$. This allows one to use many machine learning methods to estimate $g_N$. 

	Conditions (b) and (e--f) are minor regularity conditions which are expected to hold in any realistic application.

	Under conditions (a--f), $N^{1/2}\left( \Hat{\Psi}^{g_N}(\mathbb{P}_N) - \Psi^{g_N}(P_0) \right)$ is asymptotically normal with mean zero and covariance matrix 
	\begin{equation*}
		\Omega^{g^*}_0 = \left( \dot{\Phi}^{g^*}(\boldsymbol{\theta}^{g^*}_0) \right)^{-1} P_0 \! \left( \phi^{g^*}_{\boldsymbol{\theta}^{g^*}_0} \right)^{\otimes 2} \left( \dot{\Phi}^{g^*}(\boldsymbol{\theta}^{g^*}_0) \right)^{-1},
	\end{equation*}
	where $\boldsymbol{\theta}^{g^*}_0 := \Psi^{g^*}(P_0)$ and $\dot \Phi^{g}$ is the matrix of partial derivatives of $\boldsymbol{\theta} \mapsto \Phi^g(\boldsymbol{\theta}) := P_0 \phi^g_{\boldsymbol{\theta}}$.
	This covariance matrix can be consistently estimated by the sandwich estimator:
	\begin{equation}\label{eq:sandwich-estimator}
		\Hat{\Omega}^{g_N}_N = \left(\mathbb{P}_N \dot{\phi}_{\Hat{\boldsymbol{\theta}^{g_N}_N}}^{g_N}\right)^{-1} \mathbb{P}_N \! \left( \phi_{\Hat{\boldsymbol{\theta}^{g_N}_N}}^{g_N} \right)^{\otimes 2} \left(\mathbb{P}_N \dot{\phi}_{\Hat{\boldsymbol{\theta}^{g_N}_N}}^{g_N}\right)^{-1},
	\end{equation}
	where $\Hat{\boldsymbol{\theta}}^{g_N}_N := \Hat{\Psi}^{g_N}(\mathbb{P}_N)$ and $\dot \phi^g_{\boldsymbol{\theta}}(o)$ is the matrix of partial derivatives of $\boldsymbol{\theta} \mapsto \phi^g_{\boldsymbol{\theta}}(o)$. 
	In the simulations and data application, the sandwich estimator is finite-sample adjusted by multiplying it with $N/(N - 1)$ and confidence intervals based on the normal approximation and sandwich estimator are computed using the $t$ distribution with $N - 1$ degrees of freedom.

	We can alternatively construct confidence intervals using the bootstrap, which is valid by Theorem \ref{theorem:bootstrap-for-data-adaptive-estimands}.
	Because the number of independent trials $N$ is often small in practice, the multinomial bootstrap yields a highly discrete bootstrap distribution, which may harm finite-sample performance. We therefore use the Bayesian bootstrap \parencite{rubin1981bayesian}, a multiplier bootstrap in which the weights $W$ are sampled independently from a unit exponential distribution. In the simulations and data application, bootstrap confidence intervals are constructed using the bias-corrected and accelerated (BCa) method \parencite{efron1987better}.

	\section{Simulations}\label{sec:simulations}

	We conducted simulations to evaluate the finite-sample performance of the proposed methods. We considered two scenarios: a proof-of-concept scenario and a more realistic vaccine-trials scenario. 
	The proof-of-concept scenario uses a simple data-generating mechanism that demonstrates the advantage of using the estimated surrogate index as putative surrogate. Here, the trial-level correlations for the estimated surrogate indices (i.e., the data-adaptive parameters) were much closer to 1 than the trial-level correlations for the original surrogates. 
	The vaccine-trials scenario has a more realistic data-generating mechanism. We summarize the main findings of the simulations below and disucss the simulations in full detail in Appendix \ref{appendix:simulation-results}.

	Confidence intervals constructed using the bootstrap and the sandwich estimator showed reduced coverage in small-sample settings ($N = 6$); however, coverage improves as the number of trials increases. 
	For $N = 12$ and $N = 24$, the sandwich confidence intervals had near-nominal coverage, while the BCa confidence intervals had slight undercoverage. For $N = 6$, the sandwich confidence intervals do not always have closer to nominal coverage than the BCa confidence intervals. 
	These results indicate that there is no best method for constructing confidence intervals for $N = 6$ and we, therefore, recommend to use different types of confidence intervals simultaneously to increase robustness of the conclusions.

	We also studied confidence intervals based on a parametric Bayesian model.  Although it is not asymptotically valid for the data-adaptive target parameters because the parametric assumptions are violated, it can be more stable in small samples and is often used in practice.
	The resulting credible intervals had good coverage in small samples, confirming the above intuitions.

	The small-sample difficulties are not driven by the targeting of data-adaptive parameter alone, because similar issues arise when estimating the trial-level correlations of the original surrogates. 
	This is not suprising as all inferences are based on asymptotics, where the number of trials grows to infinity.

	\section{Application to COVID-19 vaccine efficacy trials}\label{sec:data-application}
	
	\subsection{Data Description}\label{sec:data-description}

	The surrogate endpoint analysis included five Phase 3 randomized, placebo-controlled COVID-19 vaccine efficacy trials: P3001 (Moderna), P3002 (AstraZeneca), P3003 (Janssen), P3004 (Novavax), and P3005 (Sanofi). The Sanofi trial consisted of two stages: stage 1 studied a monovalent vaccine and stage 2 studied a bivalent vaccine. These two stages, moreover, included a considerable number of non-naive subjects (i.e., subjects who have been previously infected with SARS-CoV-2), while the other trials included only naive subjects. We, therefore, split the Sanofi trial into four separate trials: (i) stage 1, naive only; (ii) stage 1, non-naive only; (iii) stage 2, naive only; and (iv) stage 2, non-naive only. We further refer to the latter as the Sanofi trials, ignoring that they are strictly speaking trial subunits.

	Two antibody markers were measured in these trials as putative surrogates: (i) log10 IgG binding antibody concentration against the D614 index strain Spike protein (IgG Spike) and (ii) log10 50\% inhibitory dilution neutralizing antibody titer (nAb ID50) against the D614G reference strain. These markers were measured at different time points post-enrollment: 57 days for Moderna, 57 days for AstraZeneca, 29 days for Janssen, 35 days for Novavax, and 43 days for Sanofi. The per-protocol cohort included participants who received all planned vaccinations without protocol deviations and, in the first four trials, who were SARS-CoV-2 negative at the last vaccination visit. The per-protocol cohort in the Sanofi trial includes subjects who have been previously infected as defined by an antibody or antigen test positive by the last vaccination visit. Eligibility for correlates analyses in the first four trials required participants to remain SARS-CoV-2 negative up to 6 days post marker-measurement visit, and for the Sanofi trials to not be infected between the last vaccination and 6 days post marker-measurement visit.

	All trials stratified randomization by age and prospective risk category (based on CDC guidelines for severe COVID-19 risk factors at the time of enrollment). Additionally, a harmonized case-cohort design with two-phase sampling was implemented across trials to measure IgG Spike and nAb ID50. The clinical endpoint was virologically confirmed symptomatic SARS-CoV-2 infection occurring between 7 and 80 days post marker-measurement visit. Table \ref{table:application} summarizes (i) the study cohort sizes and clinical endpoint counts by randomization arm and (ii) the number of subjects with antibody data by case and vaccine status.

	\begin{table}
		\centering
		\caption{Phase 3 COVID-19 vaccine efficacy trial data included in the surrogate endpoint evaluation. The numbers between brackets indicate the number of subjects from the corresponding group with antibody marker measurements. For the Sanofi trials, the first number inside the brackets indicates the number of IgG Spike measurements and the second one indicates the number of nAb ID50 measurements. Naive subjects in the placebo groups have a known negligible antibody response to SARS-CoV-2 equal to the lowest possible value (see Appendix \ref{appendix:naive-subject-markers}). The number of antibody marker measurements is therefore omitted in those groups.
		Note that the counts for the cases include all subjects with an event during the entire trial period (which includes follow-up beyond 80 days post marker-measurement visit). Non-cases are defined as per-protocol participants with no evidence of SARS-CoV-2 infection (i.e., never tested nucleic acid amplification/PCR positive) after enrollment up to the end of follow-up}
		\label{table:application} 
		\resizebox{\linewidth}{!}{%
		\begin{tabular}{lcccccc}	
			\toprule
			& \multicolumn{3}{c}{Placebo}  &  \multicolumn{3}{c}{Vaccine} \\ 
			\cmidrule(lr){2-4} \cmidrule(lr){5-7}
			Phase 3 Trial  & n &  Cases  & Non-Cases & n & Cases & Non-Cases \\ 
			\midrule
			Moderna &   13,682    &  656   &  13,026 & 13,979 & 47 (36) & 13,932 (1,000)  \\      
			AstraZeneca &  6,212    &  55   &  6,157 & 13,492 & 45 (25) & 13,447 (438)   \\ 
			Janssen &   17,821   &  818   &  17,003 & 18,155 & 420 (373) & 17,735 (816)    \\
			Novavax &   7,250   &  44   &  7,206 & 16,143 & 12 (12) & 16,131 (714)   \\
			Sanofi (stage 1, naive) &   820   &  84   &  736 & 772 & 77 (69, 72) & 695 (178, 268)   \\ 
			Sanofi (stage 1, non-naive) &   3,297   &  114 (102, 109)   &  3,183 (198, 238) & 3,274 & 72 (62, 67) & 3,202 (165, 203)   \\ 
			Sanofi (stage 2, naive) &   292   &  26   &  266 & 275 & 19 (17, 17) & 256 (134, 179)   \\
			Sanofi (stage 2, non-naive) &   4,660   &  72 (56, 56)   &  4,588 (236, 413) & 4,798 & 70 (60, 61) & 4,728 (245, 456)   \\
			\bottomrule
		\end{tabular}%
		}
	\end{table}

	The nAb ID50 in these trials quantifies neutralization against the reference strain; however, this was not the only circulating strain in the trials.
	Moreover, the trials were conducted in different regions and in different periods, leading to between-trial variability in circulating strains. 
	Ideally, the neutralization titer should be measured against the circulating strains to ensure that the surrogate endpoint matches the clinical endpoint (i.e., symptomatic infection with the circulating strains).
	This ``ideal'' surrogate endpoint is not available, however. Instead, we approximated this ideal surrogate endpoint using two information sources: (i) the strains circulating in the subject's region in the period that the subject was in the trial as obtained from GISAID \parencite{gisaid_variants} and (ii) the estimated geometric mean titer (GMT) ratios (reference strain to circulating strain) from the literature. 
	The so obtained variable is termed the \textit{adjusted nAb ID50}.
	Details on this adjustment are provided in Appendix \ref{appendix:adjustment-neutralization-titer}.

	\subsection[Estimating the surrogate index gN]{Estimating the surrogate index $g_N$}

	For each potential surrogate (IgG Spike, nAb ID50, and adjusted nAb ID50), we consider two estimators for the surrogate index based on data pooled across trials and treatment arms: a Cox proportional hazards (Cox PH) model and a SuperLearner. We briefly discuss these estimators next and provide more details in Appendix \ref{appendix:surrogate-index-estimators}.

	The Cox PH model includes, in addition to the potential surrogate, the following baseline covariates: age at randomization, sex, BMI category (BMI < 18.5, 18.5 $\le$ BMI < 25, 25 $\le$ BMI < 30, BMI $\ge$ 30), and an indicator of whether the subject was at high risk of severe COVID-19 \parencite{polack2020safety}.
	We stratified the Cox PH model by trial to account for different forces of infection across the trials. Predictions of the clinical endpoint (i.e., probability of symptomatic infection by 80 days post marker-measurement visit) are obtained by using the estimated regression coefficients and the estimated baseline hazard function from the Janssen trial.
	We did not use the baseline hazard function estimated from the subject’s own trial, because that would make the estimated surrogate index trial-dependent. We chose the Janssen trial because it was the largest trial. 
	Nonetheless, we expect the choice of baseline hazard function to have little impact on the trial-level surrogacy results. Changing the baseline hazard affects absolute risk predictions in the same way for both treatment arms, which has therefore little impact on the corresponding treatment effects.

	SuperLearner is an ensemble method that combines multiple candidate prediction algorithms into a single predictive model. It uses cross-validation to evaluate each candidate's performance and aggregates their predictions via a meta-learner that estimates optimal weights for combining predictions from the different candidates \parencite{van2007super}.
	To estimate the surrogate index, we use a SuperLearner with a set of logistic regression models (varying the included predictors and interaction terms) and a generalized additive logistic regression model as candidate prediction algorithms. The same covariates as in the Cox PH model are included in the candidate algorithms, but we additionally included the logit of the estimated trial-level COVID-19 probability in the control groups as a trial-level covariate to account for different forces of infection across trials. This probability is the estimated cumulative incidence of COVID-19 by 80 days post marker-measurement visit in the control group of each trial.
	The surrogate is included in each candidate algorithm.
	We further used a modified cross-validation scheme to avoid trial-level overfitting. Specifically, we use leave-one-trial-out cross-validation: each candidate prediction algorithm is fitted on data from all but one trial and then used to predict outcomes in the omitted trial. Repeating this process yields cross-validated predictions for all trials, which serve as inputs to the meta-learner. Finally, all candidate methods are refitted on the full data set, and their predictions are combined using the weights estimated by the meta-learner to produce final predictions.

	We have to address missingness in predictors (i.e., in the antibody marker) and in the outcome (i.e., censoring before 80 days) when estimating the surrogate index.
	For the Cox PH model, we weight obervations by the inverse probability of being sampled in the case-cohort design; censoring is automatically handled in the Cox PH model.
	For the SuperLearner, we weight observations by the product of the inverse probability of being sampled in the case-cohort design and the inverse probability of censoring weights (IPCWs), under a random censoring assumption. The IPCWs are estimated by the Kaplan--Meier estimator for the censoring time stratified by trial and treatment group.

	\subsection{Trial-Level Treatment Effects}

	The trial-level treatment effects on the clinical endpoint and the estimated surrogate index are estimated using standard methods. 
	The cumulative incidence rates for symptomatic infection by 80 days post marker-measurement visit are estimated with the Kaplan--Meier estimator.
	The estimated surrogate index means are estimated by first fitting a linear regression model with the estimated surrogate index as outcome and the treatment indicator, age, sex, BMI category, and the high-risk indicator as covariates. Note that the estimated surrogate index is missing whenever the antibody marker is missing. We, therefore, use the inverse probability of being sampled in the case-cohort design as weights in these linear regression models.
	The estimated means are obtained as the mean of the predicted values for all subjects in a given trial with the treatment indicator set to 1 (vaccine) and 0 (placebo), respectively. This approach is more efficient than simple sample means and is valid even if the linear model is misspecified \parencite{van2024covariate}. 

	We further consider the treatment effects on a transformed scale to ensure that they can take values on the entire real line. For the clinical endpoint, we consider the following transformation of vaccine efficacy (VE): $\log(1 - \operatorname{VE})$. 
	Similarly, for the estimated surrogate index, we consider the treatment effect on the following scale:
	\begin{equation*}
		\log \left( 1 - \widetilde{\operatorname{VE}} \right) \quad \text{for} \quad \widetilde{\operatorname{VE}} = 1 - \frac{E_{P^T}(g_N(X, S) \mid Z = 1)}{E_{P^T}(g_N(X, S) \mid Z = 0)}.
	\end{equation*}
	The within-trial covariance matrices of the estimated treatment effects are estimated using the non-parametric bootstrap (with 5,000 replications) in each trial. We keep the case-cohort weights fixed in the bootstrap samples; this is expected to lead to conservative standard errors. 

	\subsection{Results}

	We present the meta-analyses using all trials except the ones with non-naive subjects. We exclude the trials with non-naive subjects because the holistic set of immune responses differs substantially between naive and non-naive subjects and, consequently, trial-level surrogacy in a population of trials that mixes naive and non-naive subjects is not meaningful from an immunological perspective. 
	It is difficult to reach conclusive results about trial-level surrogacy because the number of independent trials is small. 
	We, therefore, also repeat the analyses on a pseudo-real data set with (i) 24 trials obtained by duplicating each of the original six trials four times, and (ii) the same six trials, but with a larger within-trial sample size.

	\subsubsection{Original Data}

	\paragraph{Trial-Level Surrogacy}

	The trial-level treatment effects on the three antibody markers (as differences in mean base-10 log antibody marker) and the estimated VEs are shown in Figure \ref{fig:ma-standard-naive-only}. This corresponds to a canonical MA for each marker separately. If we were to exclude the Sanofi trials, this plot would suggest a strong trial-level surrogacy. However, the Sanofi trials are clear outliers, which seems to suggest that IgG Spike and nAb ID50 are poor trial-level surrogates. The trial-level surrogacy for the adjusted nAb ID50 is less clear from Figure \ref{fig:ma-standard-naive-only}, however.
	The corresponding plot using the estimated surrogate indices is shown in Figure \ref{fig:ma-surrogate-index-naive-only}. 
	This plot leads to similar conclusions.

	Table \ref{table:trial-level-correlations-naive-only} presents the estimated trial-level correlation coefficients for both the untransformed antibody markers (corresponding to the canonical MA in Figure \ref{fig:ma-standard-naive-only}) and the estimated surrogate indices (corresponding to Figure \ref{fig:ma-surrogate-index-naive-only}), along with 95\% confidence intervals.
	Notably, the estimated trial-level correlations are higher for the estimated surrogate indices than the antibody markers. Furthermore, the estimated trial-level correlations for the adjusted nAb ID50 are higher than those for the unadjusted nAb ID50, indicating that the adjustment for the circulating strains improves the trial-level surrogacy.
	Nonetheless, the confidence intervals are wide, and the inferences should be interpreted with caution, as they are based on only six independent trials.

	\begin{figure}
		\centering
		\includegraphics[width=0.8\textwidth]{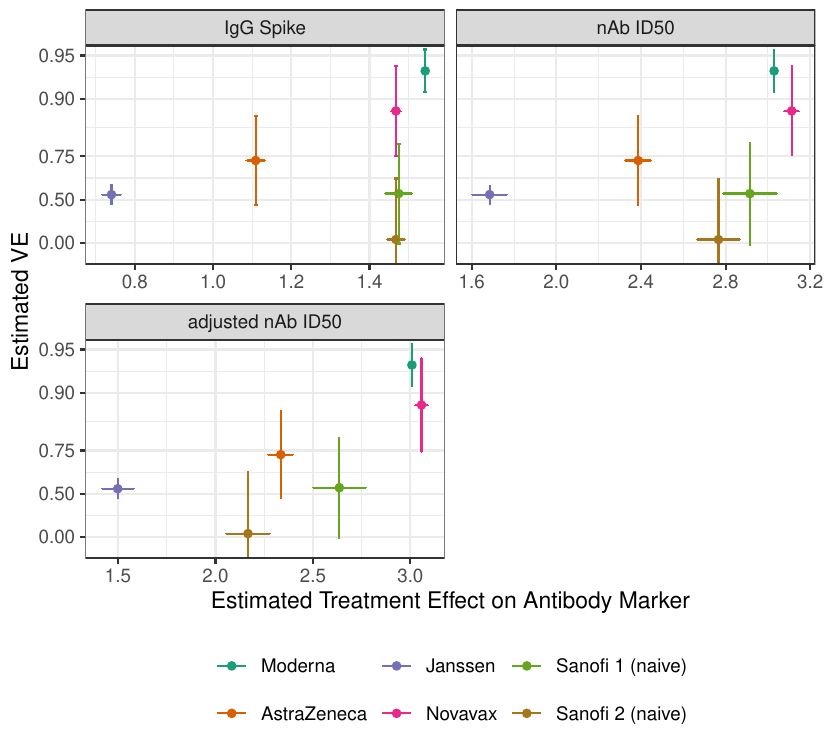}
		\caption{Estimated trial-level treatment effects on the antibody markers (difference in mean base-10 log units) and vaccine efficacies (VEs) with 95\% confidence intervals based on the bootstrap standard errors. These plots correspond to a canonical meta-analysis for each antibody marker separately. }
		\label{fig:ma-standard-naive-only}
	\end{figure}

	\begin{figure}
		\centering
		\includegraphics[width=0.8\textwidth]{"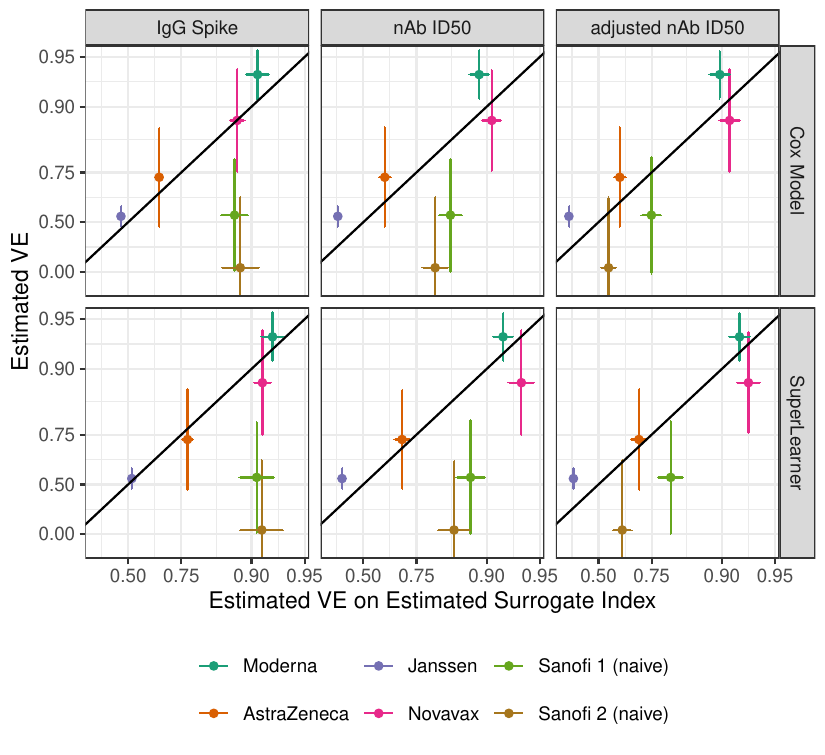"}
		\caption{Estimated trial-level treatment effects on the estimated surrogate indices and vaccine efficacies (VEs) with 95\% confidence intervals based on the bootstrap standard errors. The black line is the identity line.}
		\label{fig:ma-surrogate-index-naive-only}
	\end{figure}

\begin{table}
        \centering
        \caption{Estimated trial-level correlations with 95\% BCa bootstrap confidence intervals for the meta-analysis that includes all naive-cohort trials P3001-P3004.
        The Antibody Marker column indicates the meta-analyses where the antibody marker (without any transformation) is evaluated as a trial-level surrogate endpoint.}
        \label{table:trial-level-correlations-naive-only}
        \begin{tabular}{lccc}
            \toprule
            Surrogate & Antibody Marker & \multicolumn{2}{c}{Est.~Surrogate Index} \\
            & & Cox Model & SuperLearner \\
            \midrule
            IgG Spike & 0.33 (-0.37, 0.83) & 0.34 (-0.46, 0.85) & 0.33 (-0.45, 0.83) \\
            nAb ID50 & 0.47 (-0.28, 0.79) & 0.56 (-0.34, 0.83) & 0.57 (-0.32, 0.82) \\
            Adj.~nAb ID50 & 0.77 (0.12, 0.91) & 0.88 (0.34, 0.95) & 0.87 (0.25, 0.95) \\
            \bottomrule
        \end{tabular}
    \end{table}

	In Appendix \ref{appendix:bayesian-analysis}, we also present the results of a Bayesian analysis using weakly informative priors; this is the same Bayesian model as studied in the simulations. This alternative analysis leads to similar conclusions as the frequentist non-parametric analysis. We also repeated the Bayesian analysis assuming a proportional regression of $\beta$ on $\alpha^g$; this leads to considerably more precise results, but this gain in precision due to making additional parametric assumptions cannot be adequately justified given insufficient knowledge and empirical data for establishing the assumptions. 

	\paragraph{Prediction of Vaccine Efficacy in the Sanofi Trials of Non-Naive Subjects}

	The estimated treatment effects on the estimated surrogate index may serve as predictions for the treatment effects on the clinical endpoint (i.e., for the VE). We further refer to these estimated treatment effects as the predicted VEs. 
	The point estimates for the trial-level correlation suggest that these predictions should be accurate for the estimated surrogate indices based on adjusted nAb ID50 for trials sampled from the same population as the Moderna, AstraZeneca, Janssen, and Novavax trials, and the Sanofi trials of non-naive subjects (i.e., for trials sufficiently ``similar''). 
	To illustrate this, we show the corresponding predicted VEs for the Sanofi trials with non-naive subjects in Figure \ref{fig:ma-predictions}. These predictions are mainly for illustrative purposes, as the Sanofi non-naive trials are not sufficiently similar to the first six trials because, as mentioned before, the holistic set of immune responses differs substantially from that in naive subjects, caused in part by the fact that naive subjects were only exposed to the SARS-CoV-2 Spike protein (the protein included in the vaccine) whereas non-naive subjects, being infected, were exposed to all of the viral proteins. 
	Somewhat surprisingly, the predicted VEs are accurate for the estimated surrogate indices based on the adjusted nAb ID50. However, the confidence intervals for the estimated VEs are wide; hence, it is possible that the predicted VEs do not accurately predict the true VEs.

	\begin{figure}
		\centering
		\includegraphics[width=0.8\textwidth]{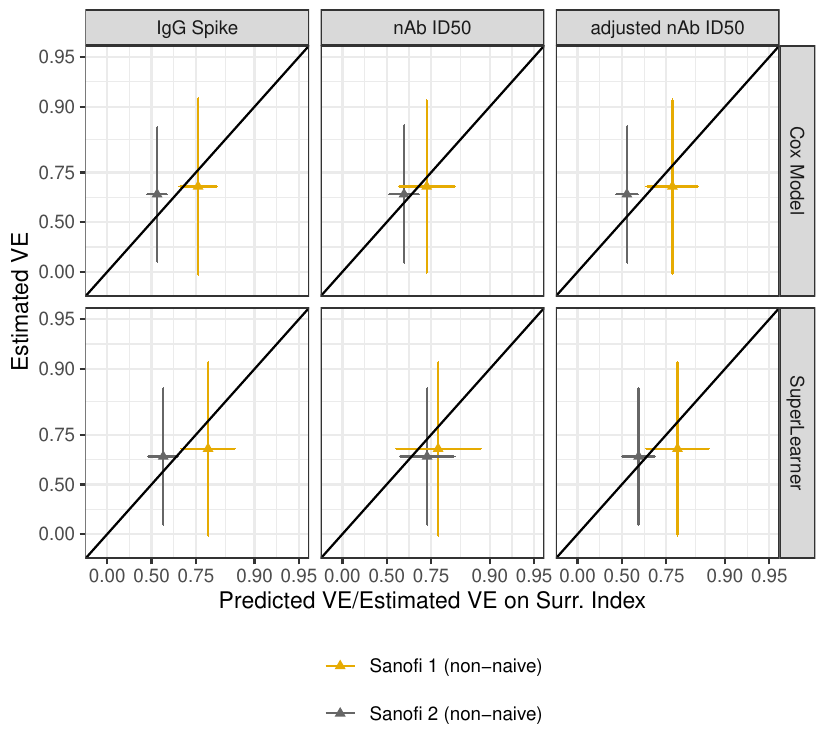}
		\caption{Predicted and estimated vaccine efficacies (VEs) for the non-naive Sanofi trials. The predicted vaccine efficacies are based on the estimated surrogate indices, estimated using the data from the six naive trials. The whiskers are the 95\% confidence intervals based on the bootstrap standard errors.
		The black line is the identity line.}
		\label{fig:ma-predictions}
	\end{figure}

	\subsubsection{Pseudo-Real Data}

	The results of the meta-analysis with the six naive trials are mostly inconclusive because the confidence intervals are wide. To further illustrate the potential of the proposed method, we repeated the analyses on a pseudo-real data set (i) with the same six naive trials as before, but dividing all within-trial variance matrix entries by four, and (ii) with 24 trials obtained by cloning each trial four times.
	The corresponding estimates and confidence intervals for the trial-level correlation coefficients are shown in Table \ref{table:trial-level-correlations-pseudo-real}.
	These results indicate that increasing the precision in each trial increases the precision of the trial-level correlation estimate only slightly; however, increasing the number of trials leads to a substantial increase in precision.

	\begin{table}
        \centering
        \caption{Estimated trial-level correlations with 95\% BCa bootstrap confidence intervals for the meta-analysis on the pseudo-real data.
        The Antibody Marker column indicates the meta-analyses where the antibody marker (without any transformation) is evaluated as a trial-level surrogate endpoint.}
        \label{table:trial-level-correlations-pseudo-real}
        \begin{tabular}{llccc}
            \toprule
            Pseudo-Real Data & Surrogate & Antibody Marker & \multicolumn{2}{c}{Est.~Surrogate Index} \\
            & & & Cox Model & SuperLearner \\
            \midrule
            \multirow{3}{*}{Increased Precision} 
                & IgG Spike & 0.32 (-0.29, 0.80) & 0.33 (-0.36, 0.81) & 0.31 (-0.36, 0.79) \\
                & nAb ID50 & 0.44 (-0.18, 0.76) & 0.54 (-0.21, 0.79) & 0.54 (-0.20, 0.79) \\
                & Adjusted nAb ID50 & 0.73 (0.18, 0.87) & 0.84 (0.43, 0.93) & 0.83 (0.38, 0.92) \vspace{0.1cm}\\
            \multirow{3}{*}{Increased Number of Trials} 
                & IgG Spike & 0.34 (0.04, 0.60) & 0.35 (0.01, 0.63) & 0.33 (0.02, 0.61) \\
                & nAb ID50 & 0.47 (0.23, 0.65) & 0.57 (0.31, 0.71) & 0.58 (0.32, 0.72) \\
                & Adjusted nAb ID50 & 0.77 (0.62, 0.85) & 0.89 (0.78, 0.93) & 0.88 (0.76, 0.92) \\
            \bottomrule
        \end{tabular}
    \end{table}

	\section{Conclusions and Discussion}\label{sec:conclusions}
	
	This paper extends the canonical MA framework in a manner that allows for complex surrogate endpoints. 
	While this approach may be useful in settings where the canonical MA framework can also be used---as, e.g., in the application in Section \ref{sec:data-application}---it may be especially useful in settings where the canonical MA framework cannot be used, such as in settings with complex surrogate endpoints. For instance, the approach as presented in this paper can be used to evaluate the trial-level surrogacy of the information available 3 years after randomization (including, e.g., fasting plasma glucose measurements, but also the presence of diabetes) as a trial-level surrogate for diabetes 4 years after randomization \parencite{agniel2024robust}. An open problem, however, is how to deal with intercurrent events that may preclude the measurement of the surrogate endpoint, such as death or dropout. 
	We addressed intercurrent events in the data application by restricting the analysis to a per-protocol cohort, but this is not entirely satisfactory as further discussed in \textcite{gilbert2025surrogate}.

	The proposed approach runs into the same limitations as the canonical MA framework, like the population of trials being ill defined and limited power due to only a small number of trials being available. Additionally, whereas the canonical MA framework can be used when appropriate summary-level statistics are available, the proposed approach always requires individual-participant data for all trials where the same covariates should be measured across trials. In the data application, such data were available because the trials were designed in a harmonized manner, but this is rare in practice.
	If the trial populations are similar, one could even apply this approach without any covariates; this may still lead to higher trial-level correlations than the canonical MA framework because treatment effects on $g_N(S)$ may still correlate better with treatment effects on the clinical endpoint than treatment effects on $S$.

	We have focused on a particular full-data parameter that reduces to the target parameter of the canonical MA framework when the parametric model is correctly specified. However, one may define different full-data parameters that are identical to the target parameter of the canonical MA framework when the parametric model is correctly specified, but that behave differently under misspecification. The full-data parameter we have considered is appealing because it is easy to interpret even when the parametric model is misspecified; however, it may be more difficult to estimate than alternative full-data parameters.
	For instance, one could define an alternative full-data parameter in terms of a projection onto the parametric model (used in the canonical MA framework) based on the Kullback--Leibler divergence. Such a projection parameter would be easier to estimate \parencite{van2023adaptive}, but it may be more difficult to interpret when the parametric model is misspecified.

	\section*{Acknowledgements}
  
	We gratefully acknowledge all data contributors, i.e., the Authors and their Originating laboratories responsible for obtaining the specimens, and their Submitting laboratories for generating the genetic sequence and metadata and sharing via the GISAID Initiative, on which this research is based.

	The authors thank Craig Magaret from the Fred Hutchinson Cancer Center for preparing the GISAID data for adjusting the neutralization titers.

	The authors also thank the participants and study teams that conducted the COVID-19 vaccine efficacy trials of the Moderna, Astra-Zeneca, Janssen, Novavax, and Sanofi SARS-CoV-2 vaccines that are analyzed in this article.  The authors thank the US Government Vaccine COVID-19 Vaccine Correlates of Protection Program from the Department of Health and Human Services [including the Administration for Strategic Preparedness and Response; Biomedical Advanced Research and Development Authority; and the National Institute of Allergy and Infectious Diseases (NIAID) of the National Institutes of Health], which co-conducted these clinical trials in public--private partnerships with the vaccine developers. The authors also thank the vaccine developers for these partnerships.

	This work was partially supported by NIAID, which provided grant funding to the HIV Vaccine Trials Network (award numbers UM1 AI68614, UM1 AI68635, UM1 AI68618). This work was also supported by NIAID award number R37AI054165 and by the Intramural Research Program of the NIAID Scientific Computing Infrastructure at Fred Hutch, under ORIP grant S10OD028685. The content is solely the responsibility of the authors and does not necessarily represent the official views of the National Institutes of Health.

	\section*{Supplementary Material}

	The code to reproduce the results in this paper is available at \url{https://github.com/florianstijven/ma-surrogate-index} (release v1.0.0).

	\section*{Data Availability}

	The individual participant data from the COVID-19 vaccine trials cannot be made available, but the GitHub repository includes a simulated data set that resembles the original data.

	The findings of this study (adjustment of nAb to circulating variants) are based on metadata associated with 5,123,874 sequences available on GISAID up to June 5, 2023, and accessible at \url{doi.org/10.55876/gis8.250401ez}.

	\printbibliography
	
	\appendix
	\makeatletter
	\newcommand{\appendixclearpage}{%
		\ifdim\pagetotal>0pt\relax
			\clearpage
		\else
			\ifx\@deferlist\@empty
				\ifx\@toplist\@empty
					\ifx\@botlist\@empty
					\else\clearpage\fi
				\else\clearpage\fi
			\else\clearpage\fi
		\fi
	}
	\makeatother

	\appendixclearpage
	\section{Surrogate Index}\label{sec:si-framework}

	This section summarizes the surrogate-index framework as proposed by \parencite{athey2025surrogate}. This framework approaches the identification of the treatment effect on the clinical endpoint from a missing data perspective, or equivalently from a data fusion perspective \parencite{li2023efficient}. Note that this is not a surrogacy-evaluation approach; perfect surrogacy is part of the identifying assumptions. This framework was used to justify the optimal transformation of the surrogate and baseline covariates in Section \ref{sec:ma-surrogate-index}.

	\subsection{Setting and Notation}

	In the surrogate-index framework, the full data for subject $i$ consist of the following variables, of which only a subset is observed:
	\begin{itemize}
		\item $D \in \{0, 1\}$: Data-source indicator, with $D = 0$ for the experimental study and $D = 1$ for the observational study.
		\item $Z \in \{0, 1\}$: Binary treatment assignment. Treatment need not be randomized if conditional exchangeability holds, but we assume simple randomization throughout.
		\item $X \in \mathcal{X}$: Baseline covariates.
		\item $S \in \mathcal{S}$: Surrogate endpoint.
		\item $Y \in \mathcal{Y}$ and $(Y(0), Y(1))' \in \mathcal{Y} \times \mathcal{Y}$: Observed clinical endpoint and its potential outcomes under treatment $z$. We assume consistency, so $Y = Z \cdot Y(1) + (1 - Z) \cdot Y(0)$.
	\end{itemize}
	These data are assumed to be i.i.d.~sampled from some distribution $P_0^F$.

	The observed data come from two sources, each with its own missing data pattern:
	\begin{itemize}
		\item In the \textit{experimental study} ($D = 0$), a binary treatment $Z$ is assigned to all subjects after measuring the baseline covariates $X$. The surrogate endpoint $S$ is measured for all subjects, but the clinical endpoint $Y$ is missing for all subjects.
		\item In the \textit{observational study} ($D = 1$), treatment information may or may not be available. We set $Z = 0$ for all subjects with $D = 1$. The same baseline covariates $X$ and surrogate endpoint $S$ are measured as in the experimental data, but additionally, the clinical endpoint $Y$ is measured.
	\end{itemize}

	The observed data are thus $\left( D, X, Z, S, D Y \right)'$, which are assumed to be i.i.d.~sampled from some distribution $P_0$, which is implied by the full-data distribution $P_0^F$.

	The full-data parameter is the average treatment effect (ATE) in the experimental study:
	\begin{equation*}
		\beta_0 := f \left\{ E_{P_0^F} \left( Y(1) \mid D = 0 \right), E_{P_0^F} \left( Y(0) \mid D = 0 \right) \right\},
	\end{equation*}
	where $f: \mathbb{R}^2 \to \mathbb{R}$ is a contrast function such that $f(x, y) = 0$ if $x = y$ and $f(x, y) \ne 0$ otherwise.

	\subsection{Identification}

	We discuss the identification of $E_{P_0^F}(Y(z) \mid D = 0)$ from a mean-imputation perspective because it clarifies the connection between the surrogate-index method and missing-data problems. A natural first step would be to use $E_{P_0^F}(Y \mid D = 0, X, Z, S)$ to impute the missing clinical endpoint in the experimental study, but this regression function is not identified because $Y$ is not observed in the experimental study. Identification therefore starts with the assumption that the surrogate captures the entire treatment effect.
	\begin{assumption}[surrogacy]
		$Y \perp Z \mid D, X, S$.
	\end{assumption}
	By the surrogacy assumption, $E_{P_0^F}(Y \mid D = a, X = x, Z, S = s) = E_{P_0^F}(Y \mid D = a, X = x, S = s)$. The latter quantity is the surrogate index of \textcite{athey2025surrogate}.
	\begin{definition}[surrogate index]\label{def:surrogate-index}
		The surrogate index is the conditional expectation of the clinical endpoint given the surrogate outcome and the baseline covariates, conditional on the data source:
		\begin{equation*}
			g_0(a, x, s) := E_{P_0^F} \left( Y \mid D = a, X = x, S = s \right).
		\end{equation*}
	\end{definition}
	This definition of the surrogate index differs from the one given in the main text, but they agree mutatis mutandis under the comparability assumption.
	\begin{assumption}[comparability of samples]
		Let $\phi(x, s) = P(D = 0 \mid X = x, S = s)$ be the sampling score. The comparability assumption requires that:
		\begin{enumerate}
			\item $D \perp Y \mid X, S$
			\item $\phi(x, s) < 1 \; \forall \; x \in \mathcal{X}, s \in \mathcal{S}$.
		\end{enumerate}
	\end{assumption}

	Comparability identifies the surrogate index in the experimental study with the surrogate index in the observational study. The potential outcome means in the application trial are then identified by replacing the unobserved $Y$ with the surrogate index:
	\begin{align*}
		E_{P_0^F}(Y(z) \mid D = 0) & = E_{P_0^F}(Y \mid D = 0, Z = z) &~\mbox{(simple randomization + consistency)} \\
		& = E_{P_0} \left\{ E_{P_0^F}(Y \mid D = 0, Z = z, X, S) \mid D = 0, Z = z \right\} &~\mbox{(law of iterated expectations)} \\
		& = E_{P_0} \left\{ g_0(1, X, S) \mid D = 0, Z = z \right\} &~\mbox{(surrogacy + comparability)}
	\end{align*}



	\subsection{Limitations}

	\textcite{athey2025surrogate} and \textcite{gilbert2025surrogate} note that the surrogacy and comparability assumptions are very strong. These authors suggest sensitivity analyses to assess the impact of violations of surrogacy and comparability on the inferences. 
	However, an informative sensitivity analysis typically requires one to limit the magnitude of departures from the identifying assumptions.
	Moving away from sensitivity analyses, one could quantify departures from the identifying assumptions in the following error:
	\begin{equation}
		e_0 = \beta_0 - f\left[ E_{P_0} \left\{ g_0(1, X, S) \mid D = 0, Z = 1 \right\}, E_{P_0} \left\{ g_0(1, X, S) \mid D = 0, Z = 0 \right\} \right].
	\end{equation}

	Surrogacy and comparability together imply that $e_0 = 0$, and $e_0 \ne 0$ implies that at least one of these assumptions is violated. This error could be estimated if $Y$ were to become available in the experimental study, allowing one to assess the accuracy of the surrogate-index based prediction. However, $e_0 = 0$ does not imply that similar surrogate-index based predictions are possible in other experimental studies.

	In the context of multiple clinical trials from the same population of trials, one could estimate this error for each trial and estimate the distribution of these errors. This distribution would reflect violations from surrogacy and comparability in this population of trials. Something similar has been suggested by \textcite[p.~27]{athey2025surrogate}:
	\begin{quote}
		``Building on this application, it would be useful to systematically establish surrogate indices that match the long-term treatment effects estimated in other experiments and quasi-experiments. Over time, this would allow researchers to collectively build a public library of surrogate indices for long-term outcomes that could be used to expedite the analysis of future interventions.''
	\end{quote}
	In the above quote, \textcite{athey2025surrogate} essentially propose to evaluate the trial-level surrogacy of surrogate indices, which is the focus of our paper. 

	\appendixclearpage
	\section{Canonical Parametric Meta-Analysis}\label{sec:ma-parametric-estimation}

	This appendix reviews the canonical parametric MA framework. We first state the standard assumptions used to estimate the joint distribution of trial-level treatment effects on the surrogate and clinical endpoints. We then show how this framework yields prediction intervals for treatment effects on the clinical endpoint in new trials. Finally, we describe a Bayesian implementation.

	\subsection{(Approximate) Observed-Data Likelihood}\label{appendix:ma-parametric-estimation-likelihood}
	
	Because the number of trials in a meta-analysis is often small, non-parametric methods can be unstable. In practice, the distribution of $(\Hat{\alpha}^{g}, \Hat{\beta})'$ is often modeled parametrically.
	Parametric methods are typically more stable, but they are sensitive to model misspecification. As argued in Section \ref{sec:full-data-estimands-a-priori-g}, misspecification is unavoidable when $g$ is not fixed a priori. Still, if the parametric model is a reasonable approximation, the gain in stability may justify this trade-off.
	We next outline the canonical parametric MA framework and make its assumptions explicit, while suppressing the dependence on $g$ in the notation.

	Let $f(\cdot; \boldsymbol{\theta})$ be the density of $(\alpha, \beta)' \in \mathbb{R}^2$ indexed by the finite-dimensional parameter $\boldsymbol{\theta} \in \Theta \subset \mathbb{R}^p$ where $\boldsymbol{\theta}_0$ is the true parameter. 
	In most meta-analyses, a bivariate normal model is assumed for $(\alpha, \beta)'$. Hence, $\boldsymbol{\theta} = (\boldsymbol{\mu}, D)'$ and $f(\cdot; \boldsymbol{\mu}, D)$ is a bivariate normal density with mean $\boldsymbol{\mu}$ and covariance $D$ where $\boldsymbol{\mu} = (\mu_{\alpha}, \mu_{\beta})$ and 
	\begin{equation*}
		D = \begin{pmatrix}
			d_{\alpha}^2 & \rho_{\text{trial}} d_{\alpha} d_{\beta} \\
			\rho_{\text{trial}} d_{\alpha} d_{\beta} & d_{\beta}^2
		\end{pmatrix}.
	\end{equation*}
	
	If $(\alpha, \beta)'$ were observed directly, estimation of $\boldsymbol{\theta_0}$ using (restricted) maximum likelihood would be trivial. However, we only observe corresponding estimates ($\Hat{\alpha}, \Hat{\beta})'$. Hence, we have to derive the observed-data likelihood, which is the likelihood for $O = (\boldsymbol{y}, n, \Sigma)$ where $\boldsymbol{y} = (\Hat{\alpha}, \Hat{\beta})'$ and $\Sigma$ is the true within-trial variance matrix. Indeed, $n$ and $\Sigma$ are part of the observed data. 
	The observed-data likelihood is given by:
	\begin{align*}
		L(O; \boldsymbol{\theta}) & = f(\boldsymbol{y} \mid n, \Sigma) \cdot f(n, \Sigma) \\
		& = f(n, \Sigma) \cdot \int f(\boldsymbol{y} \mid T, \alpha, \beta, n, \Sigma) \cdot f(T \mid \alpha, \beta, n, \Sigma) \cdot f(\alpha, \beta \mid n, \Sigma) \, d \nu(T, \alpha, \beta)  \\
		& = f(n, \Sigma) \cdot \int f(\boldsymbol{y} \mid T) \cdot f(T \mid \alpha, \beta, n, \Sigma) \cdot f(\alpha, \beta \mid n, \Sigma) \, d \nu(T, \alpha, \beta)  &~\mbox{($\alpha$, $\beta$, $n$, and $\Sigma$ are known given $T$)}
	\end{align*}
	where $f(\boldsymbol{y} \mid T)$ is the density of the sampling distribution of the treatment effect estimators given trial $T$, and $\nu$ is an appropriate dominating measure which we leave unspecified.
	This observed-data likelihood is too complex to use in practice due to the presence of $T$. This is solved in almost all meta-analytic methods by assuming that $f(\boldsymbol{y} \mid T)$ is also bivariate normal.
	
	\begin{assumption}[normal sampling distributions]\label{assumption:normal-sampling-distribution}
		For any given trial $T \in \mathbb{T}$, the sampling distribution of the treatment effect estimator is bivariate normal with mean $(\alpha, \beta)' = (\alpha(P^{T}), \beta(P^{T}))'$ and known covariance $\Sigma/n = \Sigma(T)/n(T)$: 
		\begin{equation}
			\boldsymbol{y} \mid T \sim \mathcal{N} \!\left( (\alpha, \beta)', n^{-1}\Sigma \right).
		\end{equation}
	\end{assumption}
	Under this assumption, the observed-data likelihood simplifies to:
	\begin{align*}
		L(O; \boldsymbol{\theta}) & = f(n, \Sigma) \cdot \int f(\boldsymbol{y} \mid \alpha, \beta, n^{-1} \Sigma) \cdot f(T \mid \alpha, \beta, n, \Sigma) \cdot f(\alpha, \beta \mid n, \Sigma) \, d \nu(T, \alpha, \beta) \\
		& = f(n, \Sigma) \cdot \int f(\boldsymbol{y} \mid \alpha, \beta, n^{-1} \Sigma) \cdot f(\alpha, \beta \mid n, \Sigma) \, d \nu(\alpha, \beta) &~\mbox{(integrate out $T$)}.
	\end{align*}
	The above expression is still unwieldy because of the presence of $f(\alpha, \beta \mid n, \Sigma)$ in the integral. This is solved by making the additional assumption that the within-trial covariance and sample size are unrelated to the true treatment effects. 
	\begin{assumption}[independence of treatment effects and covariance and sample size]\label{assumption:independence-treatment-effects-covariance}
		$(\alpha, \beta)' \perp \left(n, \Sigma \right)$.
	\end{assumption}
	The observed-data likelihood now simplifies to:
	\begin{align*}
		L(O; \boldsymbol{\theta})& = f(n, \Sigma) \cdot \int f(\boldsymbol{y} \mid \alpha, \beta, n^{-1} \Sigma) \cdot f(\alpha, \beta; \boldsymbol{\theta}) \, d \nu(\alpha, \beta).
	\end{align*}
	The presence of $f(n, \Sigma)$ in the above expression is a nuisance because we are not interested in the distribution  of the within-trial sample size and covariance. 
	This is solved by conditioning inference on the observed $n$ and $\Sigma$; which leads to the following conditional observed-data likelihood:
	\begin{equation}
		\begin{split}
			L(\boldsymbol{y}; \boldsymbol{\theta}, n^{-1}\Sigma) & = \int f(\boldsymbol{y} \mid \alpha, \beta, n^{-1} \Sigma) \cdot f(\alpha, \beta; \boldsymbol{\theta}) \, d \nu(\alpha, \beta).
		\end{split}
	\end{equation}
	In principle, we could now estimate $\boldsymbol{\theta}$ by maximizing the above likelihood, whatever the model $\left\{ f(\cdot; \boldsymbol{\theta}) : \boldsymbol{\theta} \in \Theta \right\}$ for the distribution of the true treatment effects. As mentioned before, a bivariate normal model is often used here. 
	\begin{assumption}[bivariate normality of treatment effects]\label{assumption:bivariate-normality-treatment-effects}
		The true treatment effects $(\alpha, \beta)'$ are bivariate normal with mean $\boldsymbol{\mu}$ and covariance $D$:
		\begin{equation*}
			(\alpha, \beta)' \sim \mathcal{N} \! \left( \boldsymbol{\mu}, D \right).
		\end{equation*}
	\end{assumption}
	Finally, under this additional normality assumption, the conditional observed-data likelihood simplifies to:
	\begin{equation}\label{eq:observed-data-likelihood}
		L(\boldsymbol{y}; \boldsymbol{\theta}, n^{-1} \Sigma) = f(\boldsymbol{y}; \boldsymbol{\mu}, D + n^{-1} \Sigma),
	\end{equation}
	where $f(\cdot; \boldsymbol{\mu}, D + n^{-1} \Sigma)$ is a bivariate normal density with mean $\boldsymbol{\mu}$ and covariance $D + n^{-1} \Sigma$.
	Estimation of $\boldsymbol{\theta}$ is now straightforward. 
	
	\begin{remark}
		In reality, $\Sigma$ is not observed, but it can be consistently estimated as $\Hat{\Sigma} = \Hat{\Sigma}(P_{n}^T)$. Therefore, the observed-data likelihood used in practice is the following \parencite{gail2000meta, korn2005assessing, van2002advanced, daniels1997meta}:
		\begin{equation}
			L(\boldsymbol{y}; \boldsymbol{\theta}, n^{-1} \Hat{\Sigma}) = f(\boldsymbol{y}; \boldsymbol{\mu}, D + n^{-1} \Hat{\Sigma}),
		\end{equation}
		where $\Sigma$ is replaced with $\Hat{\Sigma}$. This has also been termed the approximate likelihood by some authors \parencite{gabriel2016comparing}.
		This strategy has been implemented in R packages for multivariate meta-analysis such as \textit{mixmeta} \parencite{mixmeta}. However, ignoring the estimation error in $\Hat{\Sigma}$ generally leads to inconsistent estimators under the asymptotic regime of Section \ref{sec:ma-framework} (where within-trial sample sizes do not increase with the number of independent trials). Furthermore, the sampling distribution of estimators is generally not normal for a fixed sample size, but this is required by assumption \ref{assumption:normal-sampling-distribution}.
	\end{remark}

	\subsection{Prediction of \texorpdfstring{$\beta_0$}{beta0}}
	
	The two bivariate normality assumptions (on the level of the estimated and true treatment effects) simplify prediction of treatment effects in new trials. 
	For instance, suppose we have conducted a new trial, indicated by $j = 0$, where $S$ was measured but $Y$ was not. In this trial, we are interested in the treatment effect on $Y$, which cannot be estimated directly. Still, we can estimate the treatment effect on $S$ as $\Hat{\alpha}_0$ where $\Hat{\alpha}_0 \sim \mathcal{N}(\alpha_0, \sigma^2)$, or equivalently, $\Hat{\alpha}_0 - \alpha_0 \sim \mathcal{N}(0, \sigma^2)$. Because the sampling error is independent from $\alpha_0$ (viewed as random treatment effect), it can be shown that
	\begin{equation*}
		(\Hat{\alpha}_0, \beta_0)' \sim \mathcal{N}\left\{ \boldsymbol{\mu},
		\begin{pmatrix}
			d_{\alpha}^2 + \sigma^2 & \rho_{trial} d_{\alpha} d_{\beta} \\
			\rho_{trial} d_{\alpha} d_{\beta} & d_{\beta}^2
		\end{pmatrix} \right\}.
	\end{equation*}
	The conditional distribution $\beta \mid \Hat{\alpha}  = \Hat{\alpha}_0$, which can be used for empirical Bayes prediction and inference, follows directly as
	\begin{equation}\label{eq:empirical-bayes-prediction}
		\beta \mid \Hat{\alpha} = \Hat{\alpha}_0 \sim \mathcal{N} \left( \mu_{\beta} + \rho_{trial} \cdot \frac{d_{\beta}}{\sqrt{d_{\alpha}^2 + \sigma^2}} \cdot (\Hat{\alpha}_0 - \mu_{\alpha}), d_{\beta}^2 - \frac{(\rho_{trial} d_{\alpha} d_{\beta})^2}{d_{\alpha}^2 + \sigma^2} \right).
	\end{equation}
	Above, the prediction for $\beta_0$ is shrunk towards $\mu_{\beta}$; the amount of shrinkage depends on $\sigma^2$ (i.e., the sampling variability of $\Hat{\alpha}$). This shrinkage approach to prediction is applied, among others, in \textcite{gail2000meta, korn2005assessing} and is appropriate if the new trial and the trials in the meta-analysis are sampled independently from the same population of trials.
	
	\begin{remark}
		For a new trial $j = 0$, we could consider $\alpha_0$ as fixed and assume that the conditional distribution of $\beta_0 \mid \alpha_0$ is the same for the new trial and the trials in the meta-analysis. The shrinkage present in (\ref{eq:empirical-bayes-prediction}) is then not appropriate anymore because $(\alpha_0, \beta_0)'$ is not sampled from the distribution of $(\alpha, \beta)$ implied by $P^F_0$. 
		Instead, we start from the following conditional distribution with $\alpha$ known:
		\begin{equation*}
			\beta \mid \alpha = \alpha_0 \sim N \left( \mu_{\beta} + \rho_{trial} \cdot \frac{d_{\beta}}{d_{\alpha}} \cdot (\alpha_0 - \mu_{\alpha}), \, d_{\beta}^2 (1 - \rho_{trial}^2) \right).
		\end{equation*}
		In the above prediction, we can simply replace $\alpha_0$ with $\Hat{\alpha}_0$. This prediction is unbiased for a fixed $\alpha_0$:
		\begin{equation*}
			\begin{split}
				E (\mu_{\beta} + \rho_{trial} \cdot \frac{d_{\beta}}{d_{\alpha}} \cdot (\Hat{\alpha}_0 - \mu_{\alpha}) \mid \alpha_0)  & = 	E (\mu_{\beta} + \rho_{trial} \cdot \frac{d_{\beta}}{d_{\alpha}} \cdot (\Hat{\alpha}_0 - \alpha_0 + \alpha_0 - \mu_{\alpha}) \mid \alpha_0) \\
				& = \mu_{\beta} + \rho_{trial} \cdot \frac{d_{\beta}}{d_{\alpha}} \cdot (\alpha_0 - \mu_{\alpha}) + \rho_{trial} \cdot \frac{d_{\beta}}{d_{\alpha}} \cdot E(\Hat{\alpha}_0 - \alpha_0 \mid \alpha_0) \\
				& = E(\beta_0 \mid \alpha_0).
			\end{split}
		\end{equation*}
		This approach is used throughout \textcite{alonso2016applied}.
	\end{remark}
	

	\subsection{Bayesian Implementation}\label{appendix:bayesian}

	In the data application and the simulation study, we also consider a Bayesian implementation of the canonical parametric meta-analysis. 
	In this section, we describe the Bayesian implementation that was used in the data application and the simulation study.

	We implemented a Bayesian hierarchical model for bivariate trial-level treatment effects as follows. For $i = 1, \ldots, N$ trials, let $\mathbf{y}_i := (\Hat{\alpha}_i, \Hat{\beta}_i)'$ denote the observed estimated treatment effects on the surrogate and clinical endpoints, with known within-trial covariance matrix $\Sigma_i = \Hat{\Sigma}_i$ and sample size $n_i$.

	\begin{align*}
	\text{Data model:} \qquad
	& \mathbf{y}_i \mid \alpha_i, \beta_i, \Sigma_i, n_i \sim \mathcal{N} \! \left((\alpha_i, \beta_i)', \Sigma_i / n_i \right) \\
	\text{Latent effects:} \qquad
	& (\alpha_i, \beta_i)' \mid \boldsymbol{\mu}, D \sim \mathcal{N}(\boldsymbol{\mu}, D) \\
	\text{where} \qquad
	& \boldsymbol{\mu} = \begin{pmatrix} \mu_{\alpha} \\ \mu_\beta \end{pmatrix}, \quad
	D = \begin{pmatrix} d_{\alpha}^2 & \rho_{\text{trial}} d_{\alpha} d_{\beta} \\ \rho_{\text{trial}} d_{\alpha} d_{\beta} & d_{\beta}^2 \end{pmatrix}
	\end{align*}

	The hyperparameters have the following priors:
	\begin{align*}
	& \mu_{\alpha}, \mu_{\beta} \sim \mathcal{N}(0, 2^2) \\
	& d_{\alpha}, d_{\beta} \sim \mathcal{N}^+(0, 2^2) \quad \text{(truncated to } [0, \infty)) \\
	& \rho \sim \mathrm{Uniform}(-1, 1)
	\end{align*}

	In the data application, we also consider the analysis where the regression of $\beta$ on $\alpha$ is assumed to go through the origin. This is implemented by setting $\mu_{\beta} = \rho_{\text{trial}} \cdot d_{\beta} / d_{\alpha} \cdot \mu_{\alpha}$.

	\appendixclearpage
	\section{Plug-In Estimator and Efficiency}\label{appendix:plug-in-estimator}
	
	\subsection{Plug-In Estimator}

	In the following, we describe the estimating functions for the plug-in estimator as given in (\ref{eq:plug-in-estimator}) of the main text and corresponding finite-sample corrections.
	We also explain how approximate prediction intervals can be constructed in this setting.

	\subsubsection{Estimating Function}

	The plug-in estimator for $\Psi^g(P_0)$ in (\ref{eq:plug-in-estimator}) is defined by the estimating function $\phi^g_{\boldsymbol{\theta}}: \mathbb{N} \times \mathcal{M}^*_{\text{NP}} \to \mathbb{R}^5$, indexed by $\boldsymbol{\theta} := (\mu_{\alpha}^g, \mu_{\beta}^g, d_{\alpha}^g, d_{\beta}^g, d_{\alpha, \beta}^g)'$. 
	\begin{equation}\label{eq:estimating-function}
			\begin{split}
				\phi_{\boldsymbol{\theta}}^g & := \begin{pmatrix}
					\phi_{\boldsymbol{\theta}, 1}^g \\
					\phi_{\boldsymbol{\theta}, 2}^g \\
					\phi_{\boldsymbol{\theta}, 3}^g \\
					\phi_{\boldsymbol{\theta}, 4}^g \\
					\phi_{\boldsymbol{\theta}, 5}^g 
				\end{pmatrix} = \begin{pmatrix}
					\Hat{\alpha}^g - \mu_{\alpha}^g \\
					\Hat{\beta} - \mu_{\beta}^g \\
					(\Hat{\alpha}^g - \mu_{\alpha}^g)^2 - n^{-1} \Hat{\sigma}_{\Hat{\alpha}^g}^2 - d_{\alpha}^g \\
					(\Hat{\beta} - \mu_{\beta}^g)^2 - n^{-1} \Hat{\sigma}_{\Hat{\beta}}^2 - d_{\beta}^g  \\
					(\Hat{\alpha}^g - \mu_{\alpha}^g) \cdot (\Hat{\beta} - \mu_{\beta}^g) - n^{-1} \Hat{\sigma}_{\Hat{\alpha}^g, \Hat{\beta}} - d_{\alpha, \beta}^g
				\end{pmatrix}
			\end{split}
		\end{equation}
		where 
		\begin{equation*}
			\Hat{\Sigma}^g =: \begin{pmatrix}
				\Hat{\sigma}_{\Hat{\alpha}^g}^2 & \Hat{\sigma}_{\Hat{\alpha}^g, \Hat{\beta}} \\
				\Hat{\sigma}_{\Hat{\alpha}^g, \Hat{\beta}} & \Hat{\sigma}_{\Hat{\beta}}^2
			\end{pmatrix}
		\end{equation*}
		is the estimated within-trial covariance matrix. Its expectation, $\Phi^g(\boldsymbol{\theta}) := P_0 \phi^g_{\boldsymbol{\theta}}$, is
		\begin{equation*}
			\Phi^g(\boldsymbol{\theta}) = \begin{pmatrix}
				P_0 \Hat{\alpha}^g - \mu_{\alpha}^g \\
				P_0 \Hat{\beta} - \mu_{\beta}^g \\
				P_0 \left(\Hat{\alpha}^g - \mu_{\alpha}^g\right)^2 - P_0 \left( n^{-1} \Hat{\sigma}_{\Hat{\alpha}^g}^2 \right) - d_{\alpha}^g \\
				P_0 \left(\Hat{\beta} - \mu_{\beta}^g \right)^2 - P_0 \left( n^{-1} \Hat{\sigma}_{\Hat{\beta}}^2 \right) - d_{\beta}^g  \\
				P_0 \left\{ \left(\Hat{\alpha}^g - \mu_{\alpha}^g \right) \cdot \left(\Hat{\beta} - \mu_{\beta}^g \right) \right\} - P_0 \left( n^{-1} \Hat{\sigma}_{\Hat{\alpha}^g, \Hat{\beta}} \right) - d_{\alpha, \beta}^g
			\end{pmatrix}.
		\end{equation*}
		The Jacobian of $\Phi^g(\boldsymbol{\theta})$ with respect to $\boldsymbol{\theta}$ is
		\begin{equation}\label{eq:phi-dot}
			\begin{split}
				\dot{\Phi}^g(\boldsymbol{\theta}) & = \begin{pmatrix}
					-1 & 0 & 0 & 0 & 0 \\
					0 & -1 & 0 & 0 & 0 \\
					-2 \cdot (P_0 \Hat{\alpha}^g - \mu_{\alpha}^g) & 0 & -1 & 0 & 0 \\
					0 & -2 \cdot (P_0 \Hat{\beta} - \mu_{\beta}^g) & 0 & -1 & 0 \\
					-1 \cdot \left( P_0 \Hat{\beta} - \mu_{\beta}^g \right) & -1 \cdot \left( P_0 \Hat{\alpha}^g - \mu_{\alpha}^g \right) & 0 & 0 & - 1
				\end{pmatrix} 
			\end{split}
		\end{equation}
		At any solution $\boldsymbol{\theta}^g_0$ with $\Phi^g(\boldsymbol{\theta}^g_0) = 0$, this Jacobian is diagonal with $-1$ on the diagonal.

	\subsubsection{Finite-Sample Corrections}

	The estimator based on (\ref{eq:estimating-function}) uses the biased sample covariance for $(\Hat{\alpha}^g, \Hat{\beta})'$, which divides by $N$ rather than $N - 1$. A simple finite-sample correction is to rescale the covariance terms as follows:
	\begin{equation*}
			\begin{pmatrix}
				\phi_{\boldsymbol{\theta}, 3}^g \\
				\phi_{\boldsymbol{\theta}, 4}^g \\
				\phi_{\boldsymbol{\theta}, 5}^g
			\end{pmatrix} = \begin{pmatrix}
				\frac{N}{N - 1} (\Hat{\alpha}^g - \mu_{\alpha}^g)^2 - n^{-1} \Hat{\sigma}_{\Hat{\alpha}^g}^2 - d_{\alpha}^g \\
				\frac{N}{N - 1} (\Hat{\beta} - \mu_{\beta}^g)^2 - n^{-1} \Hat{\sigma}_{\Hat{\beta}}^2 - d_{\beta}^g  \\
				\frac{N}{N - 1} (\Hat{\alpha}^g - \mu_{\alpha}^g) \cdot (\Hat{\beta} - \mu_{\beta}^g) - n^{-1} \Hat{\sigma}_{\Hat{\alpha}^g, \Hat{\beta}} - d_{\alpha, \beta}^g
			\end{pmatrix}
	\end{equation*}
	
	A corresponding finite-sample adjusted sandwich estimator is
	\begin{equation*}
		\left(P_N \dot{\phi}_{\boldsymbol{\theta}^{g_N}_0}^{g_N}\right)^{-1} \frac{N}{N - 1} P_N \left( \phi_{\boldsymbol{\theta}_0^{g_N}}^{g_N} \right)^{\otimes 2} \left(P_N \dot{\phi}_{\boldsymbol{\theta}^{g_N}_0}^{g_N}\right)^{-1}.
	\end{equation*}

	Both corrections are used in the data application and the simulation study.

	\subsubsection{Approximate Prediction Intervals}
	
	The parameters above are appealing because they are well defined in a non-parametric model, but they are not directly useful for prediction in new trials where the clinical endpoint is unobserved. Given the small number of independent trials in practice, trial-level prediction requires strong parametric assumptions. Such assumptions are in principle testable, but with so few trials those tests have little power. More importantly, the assumptions are at best approximate. The resulting prediction intervals should therefore be viewed as approximations, and they are not guaranteed to attain nominal coverage, even as $N \to \infty$.
	
	Under surrogacy and comparability, treatment effects on the surrogate index $g_0$ equal treatment effects on the clinical endpoint: $\beta(T) = \alpha^{g_0}(T)$ for all $T \in \mathbb{T}$. In practice, these assumptions hold only approximately and the surrogate index is estimated. We therefore write $\beta = \alpha^{g_N} + \epsilon$ for the estimated surrogate index $g_N$, with $E(\epsilon \mid \alpha^{g_N}) \approx 0$.
	To construct prediction intervals, we further assume that $\epsilon \sim \mathcal{N}(0, \sigma^2_{g_N})$. Let $\Hat{\alpha}_{0}^{g_N}$ denote the estimated treatment effect in a new trial, with subscript $0$ indicating the new trial and estimated standard error $\Hat{\sigma}_{\Hat{\alpha}^{g_N}, 0}$. We further treat $\alpha_0^{g_N}$ and $g_N$ as fixed.
	If we additionally assume that $\Hat{\alpha}_{0}^{g_N} - \alpha_0^{g_N} \sim \mathcal{N}(0, \sigma_{\Hat{\alpha}^{g_N}, 0}^2/n_0)$, we have
	\begin{align*}
		\beta_0 - \Hat{\alpha}_{0}^{g_N} & = (\beta_0 - \alpha_0^{g_N}) + (\alpha_0^{g_N} - \Hat{\alpha}_{0}^{g_N}) \\
		& = \epsilon_0 + (\alpha_0^{g_N} - \Hat{\alpha}_{0}^{g_N}) \sim  \mathcal{N}(0, \sigma^2_{g_N} + \sigma_{\Hat{\alpha}^{g_N}, 0}^2/n_0)
	\end{align*}
	because $\epsilon_0$ and $\alpha_0^{g_N} - \Hat{\alpha}_{0}^{g_N}$ are independent by design. The corresponding approximate $1 - \alpha$ prediction interval for $\beta_0$ is
	\begin{equation*}
		\Hat{\alpha}_{0}^{g_N} \pm z_{1 - \alpha / 2} \cdot \sqrt{\Hat{\sigma}_{\Hat{\alpha}^{g_N}, 0}^2/n_0 + \Hat{\sigma}^2_{g_N}},
	\end{equation*}
	where $z_{1 - \alpha / 2}$ is the $1 - \alpha/2$ quantile of the standard normal distribution and $\Hat{\sigma}^2_{g_N}$ estimates $\sigma^2_{g_N}$.

	\begin{remark}[estimation of $\sigma^2_{g_N}$]
		In this setup, the full-data parameter $\sigma^2_{g_N}$ is defined as $\sigma^2_{g_N} := E_{P^F} \left\{ (\beta - \alpha^{g_N})^2 \right\}$. Under assumptions \ref{assumption:unbiased-trt-effect} and \ref{assumption:unbiased-sampling-variance}, it is identified from the observed-data distribution as
		\begin{equation*}
			\sigma^2_{g_N} = E_{P} \left\{ (\Hat{\beta} - \Hat{\alpha}^{g_N})^2 \right\} - E_{P} \left\{ \Hat{\sigma}_{\Hat{\beta}}^2/n + \Hat{\sigma}_{\Hat{\alpha}^{g_N}}^2/n - 2 \Hat{\sigma}_{\Hat{\alpha}^{g_N}, \Hat{\beta}}/n \right\}.
		\end{equation*}
		where $\Hat{\sigma}_{\Hat{\beta}}^2$, $\Hat{\sigma}_{\Hat{\alpha}^{g_N}}^2$, and $\Hat{\sigma}_{\Hat{\alpha}^{g_N}, \Hat{\beta}}$ are the corresponding entries of $\Hat{\Sigma}^{g_N}$. 
		An estimator of $\sigma^2_{g_N}$ is obtained by replacing $P$ with $\mathbb{P}_N$.
	\end{remark}

	\subsection{Semi-Parametric Efficiency Analysis}\label{appendix:semi-parametric-efficiency-analysis}

	\subsubsection{Results for General Multivariate Meta-Analysis}

	The following analysis is more general than the bivariate meta-analysis setting considered in the main text. We consider a multivariate meta-analysis where the estimated treatment effects are denoted by $\Hat{X} \in \mathbb{R}^k$ and the estimated variance matrices by $\Hat{\Sigma} \in \mathbb{R}^{k \times k}$. Here, $\Hat{\Sigma}$ estimates the within-trial variance of $\Hat{X}$ (i.e., without the the root-$n$ standardization as in the main text). For notational simplicity, we suppress within-trial sample sizes. All arguments are unchanged if one augments the observed data with $n$, since we do not impose any restrictions on the dependence between $n$ and $(\Hat{X}, \Hat{\Sigma})$ (and thus the implied observed-data model remains locally non-parametric).
	
	We first present a lemma relevant for the canonical meta-analysis setup where $\Hat{\Sigma} = \Sigma$ and $\Hat{X} \mid X, \Sigma \sim \mathcal{N}(X, \Sigma)$, with $X$ the vector of true treatment effects. Although these are strong identifying assumptions, the implied observed-data model is already locally non-parametric if no assumptions are imposed on the distribution of $(X, \Sigma)$. The results and proofs closely follow \textcite[Example 25.35]{van2000asymptotic} and \textcite[Appendix B]{chen2025empirical}.

	\begin{restatable}{proposition}{propositionnonparametricmamodelb}
		\label{proposition:non-parametric-ma-modelb}
		Let $\mathcal{M}^F$ be a model for the full-data distribution of $Y = (\Hat{X}, \Hat{\Sigma}, X, \Sigma)$. For every $P^F \in \mathcal{M}^F$, assume that
		\begin{equation*}
			\Hat{X} \mid X, \Sigma \sim \mathcal{N}(X, \Sigma) \quad \text{and} \quad \Hat{\Sigma} = \Sigma
		\end{equation*}
		If the support of $X \mid \Sigma$ contains a nonempty open set in $\mathbb{R}^k$ $P^F_0$-a.s. and $\Sigma$ is positive definite $P^F_0$-a.s., then the induced model for the observed data $(\Hat{X}, \Hat{\Sigma})$ is locally non-parametric at $P_0 = h(P^F_0)$.
	\end{restatable}

	Lemma \ref{proposition:non-parametric-ma-modelb} imposes stronger model restrictions than those used in the main text. The following corollary extends the result to the weaker assumptions adopted there, namely Assumptions \ref{assumption:unbiased-trt-effect} and \ref{assumption:unbiased-sampling-variance}.

	\begin{restatable}{corollary}{corollarynonparametricmamodel}
		\label{corollary:non-parametric-ma-model}
		Under the same support and positive-definiteness conditions as in Lemma \ref{proposition:non-parametric-ma-modelb}, but replacing the conditional normality assumption by the weaker moment conditions that, for every $P^F \in \mathcal{M}^F$, 
		\begin{equation*}
			E (\Hat{X} \mid X, \Sigma) = X \quad \text{and} \quad \operatorname{Var}(\Hat{X} \mid X, \Sigma) = \Sigma,
		\end{equation*}
		and $E(\Hat{\Sigma} \mid X, \Sigma) = \Sigma$, the induced model for the observed data $(\Hat{X}, \Hat{\Sigma})$ is locally non-parametric at $P_0 = h(P^F_0)$.
	\end{restatable}

	\subsubsection{Implications}

	Because the observed-data model $\mathcal{M}$ is locally non-parametric at $P_0$, the efficient influence function is unique. Hence, any two regular asymptotically linear estimators are asymptotically equivalent and, in particular, the plug-in estimator \eqref{eq:plug-in-estimator} is efficient. This is counterintuitive for two reasons.

	First, in finite samples, the plug-in estimator may return estimated variance matrices that are not positive semi-definite. Hence, in finite samples, the plug-in estimator may be improved by truncating estimated variance matrices to be positive semi-definite by finding the closest positive semi-definite matrix to the estimated matrix. This is a well-known problem in the literature and can be solved using the \texttt{nearPD} function in R. However, the plug-in estimator will return positive semi-definite estimates with probability tending to one as $N \to \infty$. 

	Second, one may try to improve the plug-in estimator by weighting the trials by, for example, $(\Hat{\Sigma}^g/n)^{-1}$. Such weighting is valid only under the additional assumption that $\Hat{\Sigma}^g/n$ is unrelated to the treatment effects; this assumption is arbitrary as illustrated in the example below. If we would assume this, the observed-data model $\mathcal{M}$ would no longer be locally non-parametric and, there would generally exist multiple regular asymptotically linear estimators. The plug-in estimator would be one of them, but not necessarily the most efficient one.

	\begin{example}
		Consider the notation of Proposition \ref{proposition:non-parametric-ma-modelb} and assume that the proposition's conditions hold. Additionally assume that $\Sigma \perp X$. We could redefine the treatment effects as $\tilde X := X^3$, where the power is applied element-wise. We then define the corresponding covariance matric $\tilde \Sigma$, which is approximately $\tilde \Sigma \approx 9 \cdot \text{diag}(X^4) \cdot \Sigma$ by the delta method; hence, $\tilde \Sigma \not \perp \tilde X$. 
		Changing the scale of the treatment effects can thus change the covariance matrix in a way that breaks independence. 
	\end{example}

	\appendixclearpage
	\section{Inference for Data-Adaptive Parameters via Estimating Equations}\label{appendix:discussion-conditions-theorem-bootstrap-for-data-adaptive-estimands}

	This appendix proceeds in three parts. 
	First, we present a general theorem for inference about data-adaptive parameters defined through estimating equations. 
	Second, we specialize this theorem to the meta-analytic setting and give sufficient conditions for this theorem to hold in this setting.
	Third, we explain how these sufficient conditions relate to overfitting of the surrogate-index estimator and how this can be mitigated in practice.

	\subsection{A General Theorem}\label{sec:inference-estimand-g-N}

	Theorem \ref{theorem:asymptotic-linearity-data-adaptive-estimand} established asymptotic linearity for general estimators of data-adaptive parameters, but its conditions can be difficult to verify in practice.
	Theorem \ref{theorem:bootstrap-for-data-adaptive-estimands} below specializes Theorem \ref{theorem:asymptotic-linearity-data-adaptive-estimand} to parameters defined through estimating equations and gives more interpretable conditions that are easier to verify in practice.
	
	For each $g \in \mathcal{G}$ and $\boldsymbol{\theta} \in \Theta$, let $\phi^g_{\boldsymbol{\theta}}: \Omega \to \mathbb{R}^p$ denote the estimating function and define its expectation by $\Phi^g(\boldsymbol{\theta}) := P_0 \phi^g_{\boldsymbol{\theta}}$. The target $\boldsymbol{\theta}_0^g$ is defined as any solution of the population estimating equation
	\begin{equation*}
		\Phi^g(\boldsymbol{\theta}) = 0.
	\end{equation*}
	Given observed data $\left(O_i \right)_{i=1}^N$, the empirical estimator $\Hat{\boldsymbol{\theta}}_N^{g_N}$ solves the empirical estimating equation
	\begin{equation*}
		\mathbb{P}_N \phi^{g_N}_{\boldsymbol{\theta}} = 0,
	\end{equation*}
	and the bootstrap estimator $\Tilde{\boldsymbol{\theta}}_N^{g_N}$ solves the corresponding bootstrap estimating equation
	\begin{equation*}
		\Tilde{\mathbb{P}}_N \phi^{g_N}_{\boldsymbol{\theta}} = 0,
	\end{equation*}
	up to negligible $o_P(N^{-1/2})$ remainder terms, where $\Tilde{\mathbb{P}}_N$ is defined below.
	
	Before stating the theorem, we introduce additional notation.
	We let $\Tilde{\mathbb{P}}_N$ be the empirical distribution of the bootstrap sample defined as $\Tilde{\mathbb{P}}_N := \frac{1}{N} \sum_{i = 1}^N W_i \cdot \delta_{O_i}$ where $\delta_{O_i}$ is the Dirac measure at $O_i$ and $W_i$ is the bootstrap weight. For the multiplier bootstrap, the bootstrap weights $W_i$ are i.i.d.~random variables with mean $0 < \mu < \infty$ and variance $0 < \tau^2 < \infty$ such that $\frac{\mu}{\tau^2} = 1$, and with $\lVert W \rVert_{2, 1} < \infty$\footnote{The norm $\lVert \cdot \rVert_{2, 1}$ is defined as $\int_0^{\infty} \sqrt{P(|W| > x)} \, dx$.}. 
	For the multinomial bootstrap, the bootstrap weights are obtained by sampling from a multinomial distribution with $N$ trials and $N$ categories, where the probability of each category is $1/N$.
	For both the multiplier and multinomial bootstrap, the bootstrap weights are independent of the observed data $O_1, \dots, O_N$.
	
	For the bootstrap result, we use an alternative definition of weak convergence for conditional bootstrap laws (see, e.g., \textcite[Section 2.2.3]{kosorok2008introduction} for details). Briefly, let $(\mathbb{E}, e)$ be a metric space with metric $e$ and let $\breve{X}_N$ be a sequence of bootstrapped processes in $\mathbb{E}$ with random weights $W$. For some tight process $X$ in $\mathbb{E}$, we write $\breve{X}_N \overset{P}{\underset{W}{\leadsto}} X$ to mean that
	\begin{equation*}
		\sup_{h \in \operatorname{BL}_1} \left| E_W h(\breve{X}_N) - E h(X) \right| \overset{P}{\to} 0, \quad \text{and for all } h \in \operatorname{BL}_1 \quad E_W h(\breve{X}_N)^* - E_W h(\breve{X}_N)_* \overset{P}{\to} 0, 
	\end{equation*}
	where 
	\begin{itemize}
		\item $\operatorname{BL}_1$ is the set of functions $h: \mathbb{E} \to \mathbb{R}$ with $\lVert h \rVert_{\infty} \le 1$ and Lipschitz constant at most $1$ with respect to the metric $e$.
		\item $E_W h(\breve{X}_N)$ is the expectation of $h(\breve{X}_N)$ with respect to the bootstrap weights $W$ given the remaining data.
		\item $h(\breve{X}_N)^*$ and $h(\breve{X}_N)_*$ are the measurable majorants and minorants with respect to the joint data $\left( O_1, \dots, O_N, W_1, \dots, W_N \right)$. 
	\end{itemize}

	We are now ready to state the theorem. 
	\begin{restatable}{theorem}{theorembootstrapfordataadaptiveestimands}
		\label{theorem:bootstrap-for-data-adaptive-estimands}
		Let $\Theta \subset \mathbb{R}^p$ be open and define $g \mapsto \boldsymbol{\theta}^g_0$ such that $\Phi^{g}(\boldsymbol{\theta}^g_0) = 0$ for all $g \in \mathcal{G} \subset \mathbb{D}$. Further assume that $g_N \in \mathcal{G}$ with probability tending to one as $N \to \infty$, where $g_N$ depends on the observed data and $(\mathbb{D}, d)$ is a metric space with metric $d$. Further assume there exists $g^* \in \mathcal{G}$ such that $d(g_N, g^*) = o_P(1)$ and define $\rho\left\{(\boldsymbol{\theta}_1, g_1), (\boldsymbol{\theta}_2, g_2)\right\} := \lVert \boldsymbol{\theta}_1 - \boldsymbol{\theta}_2 \rVert_2 + d(g_1, g_2)$ as a metric on $\Theta \times \mathbb{D}$.

		Additionally assume the following conditions:
		\begin{enumerate}
			\item[(i)]  $\Phi^g(\boldsymbol{\theta}_N) \to 0$ implies that $\boldsymbol{\theta}_N \to \boldsymbol{\theta}^g_0$ for all $g \in \mathcal{G}$.
			\item[(ii)] $\left\{ \phi^g_{\boldsymbol{\theta}}: g \in \mathcal{G}, \boldsymbol{\theta} \in \Theta \right\}$ is strong $P_0$-Glivenko--Cantelli.
			\item[(iii)] $\exists \,  \delta_1 > 0$ such that $\left\{ \phi^g_{\boldsymbol{\theta}}: g \in \mathcal{G}, \lVert \boldsymbol{\theta} - \boldsymbol{\theta}^{g^*}_0\rVert_2 < \delta_1 \right\}$ is $P_0$-Donsker and $P_0 \lVert \phi^g_{\boldsymbol{\theta}} - \phi^{g^*}_{\boldsymbol{\theta}^{g^*}_0} \rVert_{2}^2 \to 0 $ if $\rho \left\{ (\boldsymbol{\theta}, g), (\boldsymbol{\theta}^{g^*}_0, g^*) \right\} \to 0$.
			\item[(iv)] $P_0 \lVert \phi^{g^*}_{\boldsymbol{\theta}^{g^*}_0} \rVert_2^2 < \infty$ and $\boldsymbol{\theta} \mapsto \Phi^g({\boldsymbol{\theta}})$ is differentiable in $\lVert \boldsymbol{\theta} - \boldsymbol{\theta}^{g^*}_0 \rVert_2 < \delta_2$ for some $\delta_2 > 0$ for all $g \in \mathcal{G}$ with the matrix of partial derivatives, denoted by $\dot{\Phi}^g(\boldsymbol{\theta})$, being invertible at $\boldsymbol{\theta}^{g^*}_0$ for $g = g^*$. Further assume that $(\boldsymbol{\theta}, g) \mapsto \dot{\Phi}^g(\boldsymbol{\theta})$ is continuous at $(\boldsymbol{\theta}^{g^*}_0, g^*)$ with respect to the metric $\rho$.
			\item[(v)] $\Hat{\boldsymbol{\theta}}^{g_N}_N$ and $\Tilde{\boldsymbol{\theta}}^{g_N}_N$ solve the empirical and bootstrapped estimating equations approximately: $\mathbb{P}_N \phi^{g_N}_{\Hat{\boldsymbol{\theta}}_N^{g_N}} = o_P(N^{-1/2})$ and $\Tilde{\mathbb{P}}_N \phi^{g_N}_{\Tilde{\boldsymbol{\theta}}^{g_N}_N} = o_P(N^{-1/2})$.
			\item[(vi)] $g \mapsto \Phi^g(\boldsymbol{\theta})$ is Lipschitz continuous uniformly over $\boldsymbol{\theta} \in \Theta$: $\lVert \Phi^{g_1}(\boldsymbol{\theta}) - \Phi^{g_2}(\boldsymbol{\theta}) \rVert_2 \le L \cdot d(g_1, g_2)$ for all $g_1, g_2 \in \mathcal{G}$ and $\boldsymbol{\theta} \in \Theta$ for some $L > 0$.
		\end{enumerate}

		Under conditions (i) to (vi), we have that, as $N \to \infty$,
		\begin{equation*}
			N^{1/2} \left( \Hat{\boldsymbol{\theta}}^{g_N}_N - \boldsymbol{\theta}^{g_N}_0 \right) \overset{d}{\to} Z \sim \mathcal{N}\! \left(0, \Omega^{g^*}_0 \right) \quad \text{and} \quad
			N^{1/2} \left( \Tilde{\boldsymbol{\theta}}^{g_N}_N - \Hat{\boldsymbol{\theta}}^{g_N}_N \right) \overset{P}{\underset{W}{\leadsto}} Z \sim \mathcal{N}\! \left(0, \Omega^{g^*}_0 \right),
		\end{equation*}
		where $\Omega^{g^*}_0 = \left( \dot{\Phi}^{g^*}(\boldsymbol{\theta}^{g^*}_0) \right)^{-1} P_0 \left( \phi^{g^*}_{\boldsymbol{\theta}^{g^*}_0} \right)^{\otimes 2} \left\{ \left( \dot{\Phi}^{g^*}(\boldsymbol{\theta}^{g^*}_0) \right)^{-1} \right\}^{\top}$.
	\end{restatable}

	Theorem \ref{theorem:bootstrap-for-data-adaptive-estimands} is based on \textcite[Theorem 10.16]{kosorok2008introduction}, with an extension to data-adaptive target parameters. 
	Theorem \ref{theorem:bootstrap-for-data-adaptive-estimands} shows that one can essentially ignore the data-adaptiveness of $g_N$ for first-order inference on data-adaptive target parameters. Although consistency of the sandwich estimator for $\Omega^{g^*}_0$ is not guaranteed under only these assumptions, the bootstrap still yields valid large-sample inference: one may either use it to estimate $\Omega^{g^*}_0$ consistently or use the bootstrap distribution directly for confidence intervals (e.g., percentile intervals). Also note that $g_N$ is not re-estimated in bootstrap samples.

	Conditions (i) to (v) correspond to conditions (A) to (E) in \textcite[Theorem 10.16]{kosorok2008introduction} and are standard for consistency and asymptotic normality of Z-estimators; condition (vi) is specific to the data-adaptive setting. It requires that small changes in $g$ induce only small changes in the estimating map $\boldsymbol{\theta} \mapsto \Phi^g(\boldsymbol{\theta})$.
	These conditions are discussed in more detail below.

	Condition (i) is a standard identifiability condition for Z-estimators.

	Condition (ii) requires that the class of estimating functions $\{\phi^g_{\boldsymbol{\theta}}: g \in \mathcal{G}, \boldsymbol{\theta} \in \Theta\}$ is not too complex, which is satisfied if the class $\mathcal{G}$ is not too large. For example, if $\mathcal{G}$ is a class of real-valued functions, this condition may be satisfied if $\mathcal{G}$ itself is Glivenko--Cantelli and $\phi^g_{\boldsymbol{\theta}}$ is smooth in $g$ and $\boldsymbol{\theta}$.
	Condition (iii) requires that the class of estimating functions $\{\phi^g_{\boldsymbol{\theta}}: g \in \mathcal{G}, \lVert \boldsymbol{\theta} - \boldsymbol{\theta}^{g^*}_0\rVert_2 < \delta\}$ is not too complex---as for condition (ii), this is satisfied if the class $\mathcal{G}$ is not too large. The second part of Condition (iii) requires that $(\boldsymbol{\theta}, g) \mapsto \phi^g_{\boldsymbol{\theta}}$ is sufficiently smooth.  
	
	The first part of Condition (iv) is a standard regularity condition for Z-estimators. The second part again imposes a smoothness condition on the estimating function. 
	Condition (v) asks that both the empirical and bootstrap estimators solve their respective estimating equations up to $o_P(N^{-1/2})$. This allows for some numerical error in solving the estimating equations and is a standard condition for Z-estimators.

	Condition (vi) is a new condition that is specific to the data-adaptive setting. It requires that the expectation of the estimating function, $\Phi^g(\boldsymbol{\theta})$, is sufficiently smooth in $g$. Essentially, small changes in $g$ should not cause large changes in $\Phi^g(\boldsymbol{\theta})$, regardless of the value of $\boldsymbol{\theta}$. This condition is required for showing that (i) $\boldsymbol{\theta}^g_0$ is a continuous function of $g$ and that (ii) $\Hat{\boldsymbol{\theta}}^{g_N}_N$ is consistent for $\boldsymbol{\theta}^{g^*}_0$. These two intermediate results are used in the proof of Theorem \ref{theorem:bootstrap-for-data-adaptive-estimands}. Intuitively, if condition (vi) is not satisfied, a small change in $g$ could lead to a large change in the estimating function, which would make it difficult to establish the continuity of $g \mapsto \boldsymbol{\theta}^g_0$ and the consistency of $\Hat{\boldsymbol{\theta}}^{g_N}_N$.
		
	\subsection{Plug-in Estimator for Data-Adaptive Target Parameter in MA Framework}\label{sec:plug-in-estimator-statistical-estimand}

	We next propose a set of conditions specialized to the MA setting that are sufficient for the conditions in Theorem \ref{theorem:bootstrap-for-data-adaptive-estimands}. The relation between the notation in Theorem \ref{theorem:bootstrap-for-data-adaptive-estimands} and the notation in Section \ref{sec:ma-surrogate-index} is discussed in Remark \ref{remark:notation-relation}.

	\begin{remark}\label{remark:notation-relation}
		The observed data are the within-trial sample sizes and the corresponding empirical distributions $O := (n, \mathbb{P}^T_n)$, which are random elements in $\Omega := \mathbb{N} \times \mathcal{M}^*_{\text{NP}}$.
		The target parameter is indexed by the real-valued function $g \in \mathcal{G}$ on $\mathcal{X} \times \mathcal{S}$ where we assume that $\mathbb{D} = \ell^{\infty}(\mathcal{X} \times \mathcal{S})$, which is the class of bounded real-valued functions on $\mathcal{X} \times \mathcal{S}$. We use the metric induced by the supremum norm $d(g_1, g_2) = \lVert g_1 - g_2 \rVert_{\infty}$.
	\end{remark}

	The first condition requires that the estimated surrogate index $g_N$ converges to a fixed function in the supremum norm.
	\begin{enumerate}
		\item[(a)] There exists a $g^* \in \mathcal{G}$ such that $\lVert g_N - g^* \rVert_{\infty} \overset{P}{\to} 0$ and $g_N \in \mathcal{G}$ with probability tending to one. 
	\end{enumerate}
	While this is a weak condition that should hold for many estimators of the surrogate index, it has important implications related to overfitting and positivity violations for finite samples as discussed in Appendix \ref{appendix:overfitting-of-the-surrogate-index-estimator}. 
	Also note that the supremum norm cannot easily be replaced by other norms because some of the next conditions are not plausible for other norms like the $L^2(P)$ norm.

	The second condition is a regularity condition on the set of parameters:
	\begin{enumerate}
		\item[(b)] $\boldsymbol{\theta}$ can be restricted to a bounded open set $\Theta$.
	\end{enumerate}

	The third and fourth conditions restrict the complexity of the trial-specific treatment effect estimators and the estimator for $g_N$: $\Hat{\alpha}^g$ should be smooth in $g$ and $\mathcal{G}$ should not be too complex.
	\begin{enumerate}
		\item[(c)] The following classes of real- or matrix-valued functions on $\mathbb{N} \times \mathcal{M}^*_{\text{NP}}$ are bounded $P_0$-Donsker classes: $\left\{ \Hat{\alpha}^g: g \in \mathcal{G} \right\}$, $\left\{ \Hat{\beta}: g \in \mathcal{G} \right\}$, and $\left\{ \Hat{\Sigma}^g: g \in \mathcal{G} \right\}$. This is understood to hold element wise for the latter class.
		\item[(d)] Consider $g \mapsto \Hat{\alpha}^g(n, \mathbb{P}^T_n)$ and $g \mapsto \Hat{\Sigma}^g(n, \mathbb{P}^T_n)$ as random real- and matrix-valued functions on $\mathcal{G}$ indexed by $(n, \mathbb{P}^T_n)$. These functions are $P_0$-a.e.~continuous at $g^*$ with respect to the metric induced by the supremum norm.
	\end{enumerate}
	Conditions (c--d) can be shown using many strategies. In the following remark, we outline one such strategy that (practically) does not restrict the class of estimators for the trial-specific treatment effects. 
	\begin{remark}\label{remark:strong-lipschitz-condition}
		By \textcite[Theorem 19.5]{van2000asymptotic}, $\Hat{\alpha}^{\mathcal{G}} := \left\{ \Hat{\alpha}^g: g \in \mathcal{G} \right\}$ is a $P_0$-Donsker class if 
		\begin{equation*}
			\int_0^1 \sqrt{\log N_{[]}(\varepsilon, \Hat{\alpha}^{\mathcal{G}}, L^2(P_0))} \, d\varepsilon < \infty.
		\end{equation*}
		where $N_{[]}$ is the bracketing number.
		If we additionally assume that\footnote{This holds for many commonly used smooth estimators. Let $d := \lVert g_1 - g_2 \rVert_{\infty}$. Then the maximum difference between the transformed surrogate, using $g_1$ or $g_2$, for data points with non-zero probability under $\mathbb{P}^T_n$ is $d$. Many estimators satisfy the following: shifting the observed outcomes by at most $d$ will not change the estimated treatment effect by more than $d$. Note, however, that this may fail for estimators of ratios (e.g., estimator for the relative risk) if the denominator is close to zero.} for all $g_1, g_2 \in \mathcal{G}$,
		\begin{equation}\label{eq:lipschitz-alpha-hat-remark}
			|\Hat{\alpha}^{g_1}(n, \mathbb{P}^T_n) - \Hat{\alpha}^{g_2}(n, \mathbb{P}^T_n)| \le F(n, \mathbb{P}^T_n) \cdot \lVert g_1 - g_2 \rVert_{\infty}
		\end{equation}
		where $P_0 F^2 < \infty$, then \textcite[Theorem 9.23]{kosorok2008introduction} shows that
		\begin{equation*}
			N_{[]}(2 \varepsilon \lVert F \rVert_{2, P_0}, \Hat{\alpha}^{\mathcal{G}}, L^2(P_0)) \le N(\varepsilon, \mathcal{G}, \lVert \cdot \rVert_{\infty})
		\end{equation*}
		where  $N$ is the covering number and $\lVert F \rVert_{2, P_0} := (P_0 F^2)^{1/2}$.
		This implies that $\Hat{\alpha}^{\mathcal{G}}$ is a $P_0$-Donsker class if
		\begin{equation*}
			\int_0^1 \sqrt{\log N(\varepsilon, \mathcal{G}, \lVert \cdot \rVert_{\infty})} \, d\varepsilon < \infty,
		\end{equation*}
		which holds, for instance, for a parametric class of functions indexed by a bounded subset of $\mathbb{R}^d$ (by \textcite[Example 19.7]{van2000asymptotic} combined with \textcite[Lemma 9.18]{kosorok2008introduction}). Note that (\ref{eq:lipschitz-alpha-hat-remark}) also implies condition (d).
	\end{remark}
	Remark \ref{remark:strong-lipschitz-condition} shows that parametric estimators for $g_N$ may be used in combination with practically any within-trial estimator $\Hat{\alpha}^g$. 
	However, if we impose mild smoothness conditions on the estimator $(n, \mathbb{P}^T_n) \mapsto \Hat{\alpha}^g(n, \mathbb{P}^T_n)$, we have that $\left\{ \Hat{\alpha}^g: g \in \mathcal{G} \right\}$ is a $P_0$-Donsker class whenever $\mathcal{G}$ is a Donsker class for a certain set of measures, which allows one to use many machine learning methods to estimate $g_N$. This is made precise in the following lemma.

	\begin{restatable}{lemma}{lemmadonskerconditions}\label{lemma:donsker-conditions}
		Let $(n, \mathbb{P}^T_n) \mapsto a_i(n, \mathbb{P}^T_n) \in \mathcal{X} \times \mathcal{S}$ be a mapping that extracts the $(X, S)$-values of the $i$'th observation from the within-trial sample underlying $\mathbb{P}^T_n$, where we set $a_i(n, \mathbb{P}^T_n) = a_{\mathrm{pad}}$ for $i > n$ for some fixed $a_{\mathrm{pad}} \in \mathcal{X} \times \mathcal{S}$. Define $z_i(n, \mathbb{P}^T_n)$ and $x_i(n, \mathbb{P}^T_n)$ as mappings that extract the assigned treatment and the observed covariate, respectively, from the $i$'th observation, and set $z_i(n, \mathbb{P}^T_n) = 0$ and $x_i(n, \mathbb{P}^T_n) = x_{\mathrm{pad}}$ for $i > n$ for some fixed $x_{\mathrm{pad}} \in \mathcal{X}$.
		Suppose there exists a measurable function $\Hat{\alpha}$ of $3 n^* + 1$ arguments, where $n^*$ is as in condition (1) below, such that, for every $g \in \mathcal{G}$,
		\begin{equation*}
			\Hat{\alpha}^g(n, \mathbb{P}^T_n) = \Hat{\alpha}\left((g \circ a_1)(n, \mathbb{P}^T_n), \dots, (g \circ a_{n^*})(n, \mathbb{P}^T_n), z_1(n, \mathbb{P}^T_n), \dots, z_{n^*}(n, \mathbb{P}^T_n), x_1(n, \mathbb{P}^T_n), \dots, x_{n^*}(n, \mathbb{P}^T_n), n\right).
		\end{equation*}
		Assume the following conditions hold:
		\begin{enumerate}
			\item The within-trial sample sizes are bounded above by $n^* < \infty$.
			\item $\mathcal{X}$ is a bounded subset of $\mathbb{R}^d$ for some $d < \infty$.
			\item $P_0 |\Hat{\alpha}^g|^2 < \infty$ for some $g \in \mathcal{G}$.
			\item $\Hat{\alpha}$ is Lipschitz in its first $n^*$ arguments: there exists a constant $c < \infty$ such that, for any $u, v \in \mathbb{R}^{n^*}$ and any fixed $(z_1, \ldots, z_{n^*}, x_1, \ldots, x_{n^*}, n)$,
			\begin{equation*}
			\begin{aligned}
			& \left| \Hat{\alpha}\left(u_1, \ldots, u_{n^*}, z_1, \ldots, z_{n^*}, x_1, \ldots, x_{n^*}, n \right) - \Hat{\alpha}\left(v_1, \ldots, v_{n^*}, z_1, \ldots, z_{n^*}, x_1, \ldots, x_{n^*}, n \right) \right|^2 \\
			& \qquad\le c^2 \sum_{i = 1}^{n^*} (u_i - v_i)^2.
			\end{aligned}
			\end{equation*}
			\item $\mathcal{G}$, as a uniformly bounded set of real-valued functions on $\mathcal{X} \times \mathcal{S}$, is $P_{a_i, 0}$-Donsker for all $i = 1, \ldots, n^*$ for the pushforward measure $P_{a_i, 0} := P_0 \circ a_i^{-1}$.
		\end{enumerate} 
		We then have that $\left\{ \Hat{\alpha}^g: g \in \mathcal{G} \right\}$ is a $P_0$-Donsker class.
	\end{restatable}
	
	Conditions (e) and (f) are again smoothness conditions on $g \mapsto \Hat{\alpha}^g(n, \mathbb{P}^T_n)$ and $g \mapsto \Hat{\Sigma}^g(n, \mathbb{P}^T_n)$ as functions of $g$ indexed by $(n, \mathbb{P}^T_n)$.
	\begin{enumerate}
		\item[(e)] Consider $\Hat{\alpha}^g$ as a random variable indexed by $g$, we require that the functions $P_0 \Hat{\alpha}^g, P_0 (\Hat{\alpha}^g)^2, P_0 \Hat{\alpha}^g \Hat{\beta}: \mathcal{G} \to \mathbb{R}$ are Lipschitz continuous: $\left| P_0 \Hat{\alpha}^{g} - P_0 \Hat{\alpha}^{g'} \right| \le M \cdot \lVert g - g' \rVert_{\infty}$ for some $M < \infty$ and for all $g, g' \in \mathcal{G}$ (and similarly for $P_0 (\Hat{\alpha}^g)^2$ and $P_0 \Hat{\alpha}^g \Hat{\beta}$).
		\item[(f)] Consider $\Hat{\Sigma}^g$ as a random matrix indexed by $g$, we require that the function $P_0\left( n^{-1} \Hat{\Sigma}^g \right): \mathcal{G} \to \mathbb{R}^{2 \times 2}$ is element-wise Lipschitz continuous in the same sense as in condition (e).
	\end{enumerate}
	These are minor regularity conditions that follow from the same reasoning as in the footnote of Remark \ref{remark:strong-lipschitz-condition}.

	The following proposition summarizes the above discussion. 
	\begin{restatable}{proposition}{propositionsufficiencyconditionstheorem}\label{proposition:sufficiency-conditions-theorem}
		For the setup described in Sections \ref{sec:ma-framework} and \ref{sec:ma-surrogate-index}, Conditions (a--f) listed in Appendix \ref{sec:plug-in-estimator-statistical-estimand} are sufficient for the conditions of Theorem \ref{theorem:bootstrap-for-data-adaptive-estimands}.
	\end{restatable}

	\subsection{Overfitting of the Surrogate Index Estimator}\label{appendix:overfitting-of-the-surrogate-index-estimator}

	Condition (a) requires that $\lVert g_N - g^* \rVert_{\infty} \overset{P}{\to} 0$ as $N \to \infty$. 
	Although this is a weak requirement asymptotically, it has important finite-sample implications. 
	The next paragraphs illustrate how this condition can fail, how this is relevant in finite samples, and how this can be resolved by using SuperLearner with a modified cross-validation procedure \parencite{van2007super}.

	If a surrogate index estimator includes a trial indicator in $X$, the estimated surrogate index, evaluated in $(x, s)$-values occurring in any observed trial $t$, will be close to $E(Y \mid X = x, S = s, T = t)$ if the estimator is sufficiently flexible and the within-trial sample sizes are sufficiently large. In this case, $g_N$ will not converge to some fixed $g^*$ if the comparability assumption is violated. 
	If we now assume surrogacy without comparability, the mean of $g_N(X, S)$ in trial $t$ under treatment $z$ is approximately equal to the mean of $Y$ in trial $t$ under treatment $z$:
	\begin{align*}
		E(g_N(X, S) \mid Z = z, T = t) & \approx E \left\{\mathrm{E}(Y \mid X, S, T = t)\,\bigm|\, Z = z, T = t \right\} \\
		& = E \left\{\mathrm{E}(Y \mid X, Z = z, S, T = t)\,\bigm|\, Z = z, T = t \right\} &&~\mbox{(surrogacy)}\\
		& = E(Y \mid Z = z, T = t) &&~\mbox{(iterated expectations)}\\	
	\end{align*}
	Hence, the treatment effects on $g_N(X, S)$ and $Y$ will be close in all observed trials, even if the comparability assumption is severely violated.
	The surrogate index will then appear to be a perfect trial-level surrogate, yet it will fail for unobserved trials. 
	This shows that adjusting for trial indicators is problematic. This seems like a trivial point; however, the same type of overfitting can occur when there is limited overlap of $(X, S)$ across trials. (This corresponds to a violation of the so-called positivity or overlap assumption in the transportability and data-fusion literature \parencite{li2023efficient, degtiar2023review}.)
	Especially worrisome is the fact that, in such cases, even parametric models can overfit, especially if the number of trials is small. For instance, if $S$ is a categorical variable where one category is almost exclusively observed in one trial, a parametric model may be overfitting to that trial. 

	The trial-level type of overfitting can be avoided by using SuperLearner in combination with a modified cross-validation scheme \parencite{van2007super}. 
	Specifically, the SuperLearner should use a leave-one-trial-out cross-validation scheme and should include at least one very simple learner that cannot overfit at a trial level. 
	If certain $(X, S)$ regions occur only in one trial, flexible learners may extrapolate poorly, while simpler learners typically extrapolate more robustly. 
	Thus, the trained SuperLearner will place more weight on simple learners if there is little overlap in $(X, S)$ across trials.
	If there is sufficient overlap, the SuperLearner will place more weight on flexible learners if the underlying regression function is more complex.

	\appendixclearpage
	\section{Alternative Asymptotic Framework}\label{appendix:alternative-asymptotic-framework}
 
	\subsection{Setting and Notation}

	In the data-generating mechanism introduced in Section \ref{sec:ma-framework}, we assumed that the within-trial sample size $n$ is a random variable whose distribution does not depend on the number of independent trials $N$. Consequently, the full- and observed-data distributions, $P_0^F$ and $P_0$, were unknown but fixed distributions---fixed in the sense that they do not change with $N$.
	We now deviate from this setting and let $P_{0, n} \in \mathcal{M}$ be the observed-data distribution where all within-trial sample sizes are set to $n$. For a given trial, the empirical within-trial distribution satisfies $\mathbb{P}^T_n \in \mathcal{M}^*_{\text{NP}}$, and $\mathcal{M}$ is a model for the observed-data distribution.
	As in Section \ref{sec:ma-framework}, the observed-data distribution is determined by the full-data distribution as represented by the mapping $h_n: \mathcal{M}^F \to \mathcal{M}$ where this mapping is now indexed by $n$.
	Note that $P^F_0$ and $\mathcal{M}^F$ remain unchanged as the distribution of trials does not depend on $n$, but the way the observed data for a given trial are generated does. This is visualized in Figure \ref{fig:sampling-mechanism-diagram-new}, which replaces Figure \ref{fig:sampling-mechanism-diagram} from the main text.

	For a fixed $n$, the new data-generating mechanism is a special case of the previous setting where the distribution of $n$ is degenerate.
	Hence, under identifying Assumptions \ref{assumption:unbiased-trt-effect} and \ref{assumption:unbiased-sampling-variance}, all methods presented in the main text are valid for this new setting (i.e., with $P_0$ replaced with $P_{0, n}$ and $n$ fixed). 
	However, in reality, the identifying assumptions \ref{assumption:unbiased-trt-effect} and \ref{assumption:unbiased-sampling-variance} will not be satisfied exactly for a fixed $n$. Nonetheless, we expect these assumptions to be satisfied approximately, and more specifically, we expect this ``approximation error'' to decrease with increasing $n$. We, therefore, now consider the setting where the within-trial sample size $n$ is a function of the number of independent trials $N$ (i.e., $n = n(N)$) such that $n(N) \to \infty$ as $N \to \infty$, as also visualized in Figure \ref{fig:sampling-mechanism-diagram-new}.
	In most of the notations, we leave the dependence of the within-trial sample size $n$ on the number of independent trials $N$ implicit. So, when we're using $n$, we mean $n(N)$.

	In the following arguments, we first consider a fixed $g$ and denote the trial-level treatment-effect estimates by a tilde to emphasize that they do not satisfy the identifying assumptions: $(\Tilde{\alpha}^g_n, \Tilde{\beta}_n)'$ and $\Tilde{\Sigma}^g_n$. We then consider finite-sample unbiased counterparts, denoted by $(\Hat{\alpha}^g_n, \Hat{\beta}_n)'$ and $\Hat{\Sigma}^g_n$, which satisfy the identifying assumptions; the latter will be obtained by debiasing the former. We will consider the ``augmented'' data as $(\Tilde{\alpha}^g_n, \Tilde{\beta}_n, \Tilde{\Sigma}^g_n, \Hat{\alpha}^g_n, \Hat{\beta}_n, \Hat{\Sigma}^g_n)$ whose distribution is implied by $P^F_{0}$ and $n$, but only the first part is observable.
	We further let the empirical distribution of $(\Tilde{\alpha}^g_n, \Tilde{\beta}_n, \Tilde{\Sigma}^g_n)$ be denoted by $\Tilde{\mathbb{P}}_{N, n}$ and the empirical distribution of $(\Hat{\alpha}^g_n, \Hat{\beta}_n, \Hat{\Sigma}^g_n)$ be denoted by $\Hat{\mathbb{P}}_{N, n}$. Note that the latter distribution is actually not observed.
	Our goal in this section is to study when $\Hat{\Psi}^g(\Tilde{\mathbb{P}}_{N, n})$ is asymptotically equivalent to $\Hat{\Psi}^g(\Hat{\mathbb{P}}_{N, n})$, where $\Hat{\Psi}^g$ is the plug-in estimator \eqref{eq:plug-in-estimator}. If these two estimators are asymptotically equivalent, then we can ignore the fact that the trial-specific estimates only satisfy the identifying assumptions approximately. In the remainder of this section, we formalize this goal further and derive sufficient conditions for this asymptotic equivalence. We also discuss when these conditions are theoretically non-trivial and how they may be practically relevant.
	
	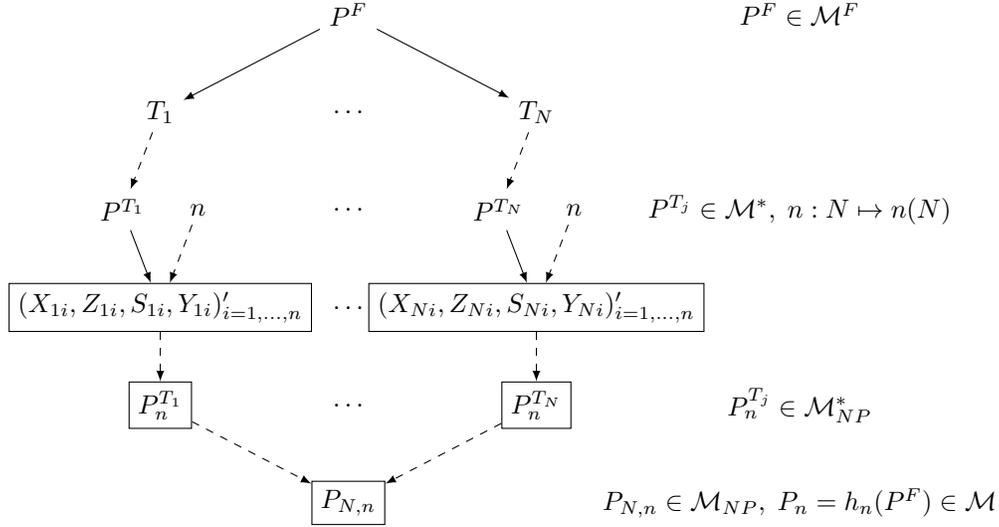
\begin{figure}
		\centering
		\begin{tikzpicture}[edge from parent/.style={draw,-latex}]
			\tikzstyle{level 1}=[sibling distance=25mm, level distance = 13mm] 
			\tikzstyle{level 2}=[sibling distance=10mm, level distance = 13mm] 
			\node (full) {$P^F$}
			child {node {$T_1$}
				child {node (a1) {$P^{T_1}$} edge from parent[dashed]}
				child {node (a2) {$n$} edge from parent[draw=none]}
			}
			child {node {$\ldots$} edge from parent[draw=none]
				child {node (b1) {$\ldots$} edge from parent[draw=none]}
			}
			child {node {$T_N$}
				child {node (c1) {$P^{T_N}$} edge from parent[dashed]}
				child {node (c2) {$n$} edge from parent[draw=none]}
			};
			
			\coordinate (CENTERa) at ($(a1)!0.5!(a2)$);
			\coordinate (CENTERc) at ($(c1)!0.5!(c2)$);
			
			\node (a3) at (CENTERa) [yshift=-13mm, draw] {$(X_{1i}, Z_{1i}, S_{1i}, Y_{1i})_{i = 1}^n$};
			\node (b3) at (b1) [yshift=-13mm] {$\ldots$};
			\node (c3) at (CENTERc) [yshift=-13mm, draw] {$(X_{Ni}, Z_{Ni}, S_{Ni}, Y_{Ni})_{i = 1}^n$};
			
			\draw [->, -latex] (a1) -- (a3);
			\draw [->, -latex, dashed] (a2) -- (a3);
			\draw [->, -latex] (c1) -- (c3);
			\draw [->, -latex, dashed] (c2) -- (c3);
			
			\node (a4) at (a3) [yshift=-13mm, draw] {$\mathbb{P}^{T_1}_{n}$};
			\node (b4) at (b3) [yshift=-13mm] {$\ldots$};
			\node (c4) at (c3) [yshift=-13mm, draw] {$\mathbb{P}^{T_N}_{n}$};
			
			\node (empirical) at (b4) [yshift=-13mm, draw] {$\mathbb{P}_{N, n}$};
			
			\node [right of =full, xshift=50mm] (fullmodel) {$P^F \in \mathcal{M}^F$};
			\node (distribution) at (fullmodel|-c2) {$P^{T_j} \in \mathcal{M}^*, \; n: N \mapsto n(N)$};
			\node (trialempdistribution) at (fullmodel|-c4) {$\mathbb{P}^{T_j}_{n} \in \mathcal{M}^*_{\text{NP}}$};
			\node (observedempdistribution) at (fullmodel|-empirical) {$\mathbb{P}_{N, n} \in \mathcal{M}_{\text{NP}}, \; \mathbb{P}_n = h_n(P^F) \in \mathcal{M}$};
			
			\draw [->, -latex, dashed] (a3) -- (a4);
			\draw [->, -latex, dashed] (c3) -- (c4);
			
			\draw [->, -latex, dashed] (a4) -- (empirical);
			\draw [->, -latex, dashed] (c4) -- (empirical);
		\end{tikzpicture}
		\caption{Data-generating mechanism underlying the meta-analytic framework where the within-trial sample size, $n$, increases with the total number of independent trials, $N$. 
			Full arrows represent sampling. Dashed arrows represent mappings.  
			Objects surrounded by a rectangle represent the observed data. $\mathcal{M}^F$: full-data model; $\mathcal{M}^*$: within-trial model; $\mathcal{M}^*_{\text{NP}}$: non-parametric within-trial model; $\mathcal{M}_{\text{NP}}$: non-parametric observed-data model; $\mathcal{M}$: observed-data model.}\label{fig:sampling-mechanism-diagram-new}
	\end{figure}

	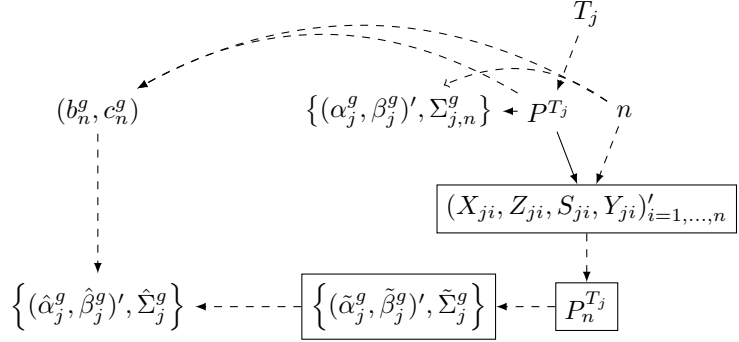
\begin{figure}
		\centering
		\begin{tikzpicture}[edge from parent/.style={draw,-latex}]
			\tikzstyle{level 1}=[sibling distance=10mm, level distance = 13mm] 
			\node (Tj) {$T_j$}
			child {
				node (PTj) {$P^{T_j}$} edge from parent[dashed]
			}
			child{
				node (nj) {$n$} edge from parent[draw=none]
			};
			
			\coordinate (center) at ($(PTj)!0.5!(nj)$);
			
			\node (data) at (center) [yshift=-13mm, draw] {$(X_{ji}, Z_{ji}, S_{ji}, Y_{ji})'_{i = 1, \dots, n}$};
			\node (empirical) at (data) [yshift=-13mm, draw] {$\mathbb{P}^{T_j}_{n}$};
			
			\draw [->, -latex] (PTj) -- (data);
			\draw [->, -latex, dashed] (nj) -- (data);

			\node (alphabeta) [left of =PTj, xshift=-10mm] {$\left\{ (\alpha_{j}^g, \beta_{j}^g)', \Sigma^g_{j, n} \right\}$};
			\node (alphabetatilde) [left of =empirical, xshift=-15mm, draw] {$\left\{ (\Tilde{\alpha}^g_{jn}, \Tilde{\beta}^g_{jn})', \Tilde{\Sigma}^g_{jn} \right\}$};
			\node (alphabetahat) [left of =alphabetatilde, xshift=-30mm] {$\left\{ (\Hat{\alpha}^g_{jn}, \Hat{\beta}^g_{jn})', \Hat{\Sigma}^g_{jn} \right\}$};
			\node (biases) [left of =alphabeta, xshift=-30mm] {$\left( b^g_n, c^g_n \right)$};
			\draw [->, dashed] (nj) to [bend right=30] (alphabeta);

			\draw [->, -latex, dashed] (data) -- (empirical);
			
			\draw [->, -latex, dashed] (PTj) -- (alphabeta);
			\draw [->, -latex, dashed] (empirical) -- (alphabetatilde);
			\draw [->, -latex, dashed] (alphabetatilde) -- (alphabetahat);
			\draw [->, -latex, dashed] (PTj) to [bend right=30] (biases);
			\draw [->, -latex, dashed] (nj) to [bend right=30] (biases);
			\draw [->, -latex, dashed] (biases) -- (alphabetahat);

		\end{tikzpicture}
		\caption{True, estimated, and debiased treatment effects as functionals of the true and empirical within-trial distributions. Some objects are indexed by $n$ while others are not, this is on purpose as further explained in Section \ref{sec:alternative-assumptions}: The true treatment effects are not indexed by $n$ because they are fixed functionals of the true within-trial distribution, the true within-trial covariance matrix $\Sigma^g_{j, n}$ is indexed by $n$ because the root-$n$ standardized variance of the estimators may change with $n$. 
		All other derived quantities are indexed by $n$ to emphasize that the corresponding functionals may depend on $n$.
		Dashed arrows represent deterministic mappings. Full arrows represent sampling. 
		Objects surrounded by a rectangle represent the observed data.}\label{fig:trial-level-sampling-mechanism-diagram-new}
	\end{figure}
	
	\subsection{Alternative Assumptions}\label{sec:alternative-assumptions}
	
	We relax Assumptions \ref{assumption:unbiased-trt-effect} and \ref{assumption:unbiased-sampling-variance} in Assumptions \ref{assumption:biased-trt-effect} and \ref{assumption:biased-sampling-variance}, respectively. The latter assumptions allow for a finite-sample bias in the trial-level treatment effect estimators that converges to zero as $n \to \infty$. 
	Throughout this subsection, $E_{P^T, n}$ denotes expectation over empirical distributions $\mathbb{P}^T_n$ generated by sampling $n$ i.i.d.~observations from $P^T$, and $\operatorname{Var}_{P^T, n}$ denotes the corresponding variance.
	We further let $\mathcal{G}$ be a set of real-valued functions on $\mathcal{S} \times \mathcal{X}$ and define the trial-level bias $b_n^g$ as follows for $P^T \in \mathcal{M}^*$:
	\begin{align*}
		b^g_n & := n^{1/2} \left[ E_{P^T, n} \! \left\{ (\Tilde{\alpha}^g_n, \Tilde{\beta}_n)' \right\} - \left\{\alpha^g(P^T), \beta(P^T) \right\}' \right]
	\end{align*}
	\begin{assumption}[vanishing bias for trial-level treatment effect estimators]\label{assumption:biased-trt-effect}
		Assume that there exist constants $\delta_1 > 0$, $n_0 < \infty$, and $K < \infty$ such that,
		\[
			P^F_{0} \lVert b^g_n \rVert_2^2 = O(n^{-\delta_1}) \quad \text{and} \quad P^F_{0} \lVert b^g_n \rVert_2^4 \le K
		\]
		for all $n \ge n_0$ and $g \in \mathcal{G}$.
	\end{assumption}
	If we could observe $b^g_n$, we could simply subtract it from the finite-sample biased estimators $(\Tilde{\alpha}^g_n, \Tilde{\beta}_n)'$ to obtain finite-sample unbiased estimators $(\Hat{\alpha}^g_n, \Hat{\beta}_n)'$:
	\begin{equation}\label{eq:debiased-trt-estimators}
		(\Hat{\alpha}^g_n, \Hat{\beta}_n)' := (\Tilde{\alpha}^g_n, \Tilde{\beta}_n)' - n^{-1/2} b^g_n \quad \text{where} \quad \Sigma^g_n(P^T) := n \cdot \operatorname{Var}_{P^T, n} \left\{ (\Hat{\alpha}^g_n, \Hat{\beta}_n)' \right\}.
	\end{equation}
	These debiased estimators are asymptotically equivalent to the original estimators because
	\begin{equation*}
		n^{1/2} \left\{ (\Tilde{\alpha}^g_n, \Tilde{\beta}_n)' - (\Hat{\alpha}^g_n, \Hat{\beta}_n)' \right\} = b^g_n = o_{P^F_{0}}(1)
	\end{equation*}
	where the second equality follows from $P^F_0 \lVert b_{n}^g \rVert_2^2 = O(n^{-\delta_1})$. 
	Further define the trial-level bias matrix $c_n^g$ as follows for $P^T \in \mathcal{M}^*$:
	\begin{equation*}
		c^g_n := E_{P^T, n} \! \left\{ \Tilde{\Sigma}^g_n \right\} - \Sigma_n^g(P^T).
	\end{equation*}
	\begin{assumption}[vanishing bias for within-trial variance estimators]\label{assumption:biased-sampling-variance}
		Let $\lVert \cdot \rVert_F$ denote the Frobenius norm. Assume that there exist constants $\delta_2 > 0$, $n_0 < \infty$, and $K < \infty$ such that,
		\[
		P_0^F \lVert c_n^g \rVert_F = O(n^{-\delta_2}) 
		\qquad \text{and} \qquad
		P_0^F \lVert c_n^g \rVert_F^2 \le K.
		\]
		for all $n \ge n_0$ and $g \in \mathcal{G}$.
	\end{assumption}
	Like we debiased the treatment effect estimators in (\ref{eq:debiased-trt-estimators}), we can also debias the within-trial variance estimators by subtracting the bias $c_{n}^g$ from the finite-sample biased $\Tilde{\Sigma}^g_n$:
	\begin{equation}\label{eq:debiased-variance-estimator}
		\Hat{\Sigma}^g_{n} := \Tilde{\Sigma}^g_n - c_{n}^g.
	\end{equation}
	By construction, $\Hat{\Sigma}^g_{n}$ is an unbiased estimator for the root-$n$ standardized variance of $(\Hat{\alpha}^g_{n}, \Hat{\beta}_{n})'$ in a given trial. Hence, Assumptions \ref{assumption:unbiased-trt-effect} and \ref{assumption:unbiased-sampling-variance} are satisfied for the debiased estimators $(\Hat{\alpha}^g_n, \Hat{\beta}_n, \Hat{\Sigma}_{n}^g)$. 
	Consequently, for fixed $n$, Sections \ref{sec:ma-surrogate-index} and \ref{sec:inference-data-adaptive-estimand} apply to the estimator $\Hat{\Psi}^g(\Hat{\mathbb{P}}_{N, n})$, defined in (\ref{eq:plug-in-estimator}), that uses the debiased estimators. Those sections derive the asymptotic distribution of
	\begin{equation*}
		N^{1/2} \left( \Hat{\Psi}^g(\Hat{\mathbb{P}}_{N, n}) - \Psi^{F, g}(P^F_0)\right) \text{ and } N^{1/2} \left( \Hat{\Psi}^{g_{N}}(\Hat{\mathbb{P}}_{N, n}) - \Psi^{F, g_{N}}(P^F_0)\right)
	\end{equation*} 
	for a fixed and a data-adaptive parameter, respectively, where $n$ is fixed and $N \to \infty$.
	
	\begin{remark}[within-trial covariance]
		The bias terms $b_n^g$ and $c_n^g$ have different scalings by design. The term $b_n^g$ is defined as $n^{1/2}[\cdots]$ because treatment effect estimators $(\Hat{\alpha}^g_n, \Hat{\beta}_n)'$ typically converge at the root-$n$ rate. By contrast, $c_n^g$ is unscaled because variance estimators are typically only required to be consistent, not to satisfy a specific convergence rate. 
	\end{remark}

	\subsection{Rate Conditions}

	In this subsection, we derive conditions under which the following holds for $N \to \infty$ and $n = n(N)$:
	\begin{equation}\label{eq:convergence-biased-estimator}
		\lVert \Hat{\Psi}^{g}(\Tilde{\mathbb{P}}_{N, n}) - \Hat{\Psi}^{g}(\Hat{\mathbb{P}}_{N, n}) \rVert_2 = o_{P^F_0}(N^{-1/2}) \text{ and } \lVert \Hat{\Psi}^{g_{N}}(\Tilde{\mathbb{P}}_{N, n}) - \Hat{\Psi}^{g_{N}}(\Hat{\mathbb{P}}_{N, n}) \rVert_2 = o_{P^F_0}(N^{-1/2}),
	\end{equation}
	where $\Hat \Psi^g$ is the plug-in estimator defined in (\ref{eq:plug-in-estimator}).
	Under these conditions, the limiting distributions based on $\Tilde{\mathbb{P}}_{N, n}$ and $\Hat{\mathbb{P}}_{N, n}$ coincide if $\Hat \Psi^g(\Hat{\mathbb{P}}_{N, n})$ and $\Tilde{\Psi}^{g_N}(\Hat{\mathbb{P}}_{N, n})$ are root-$N$ consistent, because (\ref{eq:convergence-biased-estimator}) implies that
	\begin{equation*}
		N^{1/2} \left( \Hat{\Psi}^{g}(\Tilde{\mathbb{P}}_{N, n}) - \Psi^{F, g}(P^F_0) \right) = N^{1/2} \left( \Hat{\Psi}^{g}(\Hat{\mathbb{P}}_{N, n}) - \Psi^{F, g}(P^F_0) \right) + o_{P^F_0}(1)
	\end{equation*}
	and similarly for $\Hat{\Psi}^{g_{N}}$. Hence, if $N^{1/2} \left( \Hat{\Psi}^{g}(\Hat{\mathbb{P}}_{N, n}) - \Psi^{F, g}(P^F_0) \right)$ converges in distribution, then $N^{1/2} \left( \Hat{\Psi}^{g}(\Tilde{\mathbb{P}}_{N, n}) - \Psi^{F, g}(P^F_0) \right)$ converges to the same limiting distribution, and similarly for $\Hat{\Psi}^{g_{N}}$.

	\begin{remark}\label{remark:limiting-distribution}
		Sections \ref{sec:ma-surrogate-index} and \ref{sec:inference-data-adaptive-estimand} of the main text only apply to the asymptotic regime where $n$ is fixed and $N \to \infty$. In the present setting, $n = n(N)$ with $n(N) \to \infty$ as $N \to \infty$; consequently, the limiting distributions of $\Hat{\Psi}^{g}(\Hat{\mathbb{P}}_{N, n})$ and $\Hat{\Psi}^{g_N}(\Hat{\mathbb{P}}_{N, n})$ may differ from the fixed-$n$ limits derived in these sections, even though Assumptions \ref{assumption:unbiased-trt-effect} and \ref{assumption:unbiased-sampling-variance} hold.
		This point is discussed further in Appendix \ref{appendix:practical-relevance}.
	\end{remark}

	In the following lemma, we present restrictions on $\delta_1$ and $\delta_2$ (appearing in Assumptions \ref{assumption:biased-trt-effect} and \ref{assumption:biased-sampling-variance}) such that (\ref{eq:convergence-biased-estimator}) holds.
	
	\begin{restatable}{lemma}{lemmanewasymptoticsmomentbased}
		\label{lemma:new-asymptotics-plug-in}
		Let $n(N) = N^{\gamma}$ with $\gamma > 0$ and assume that Assumptions \ref{assumption:biased-trt-effect} and \ref{assumption:biased-sampling-variance} hold, and that $P_0^F \left(\tilde{\alpha}^g_n\right)^2 < \infty$ and $P_0^F \left(\tilde{\beta}_n\right)^2 < \infty$. If $\gamma > \max\left\{ \frac{1}{1 + \delta_1}, \frac{1}{2(1 + \delta_2)} \right\}$, then 
		\begin{align}
			\lVert \Hat{\Psi}^{g}(\Tilde{\mathbb{P}}_{N, n}) - \Hat{\Psi}^{g}(\Hat{\mathbb{P}}_{N, n}) \rVert_2 & = o_{P^F_0}(N^{-1/2}),
		\end{align}
		for $\Hat{\Psi}^{g}$ as defined in (\ref{eq:plug-in-estimator}).
	\end{restatable}

	\begin{restatable}{corollary}{corollarynewasymptoticsmomentbased}
		\label{corollary:new-asymptotics-plug-in}
		Assume that the conditions of Lemma \ref{lemma:new-asymptotics-plug-in} hold for $g^*$, that $g_N \overset{P}{\to} g^*$, and that $P^F_0 \left( \tilde{\alpha}_n^{g_N}\right)^2 = O(1)$. If the following conditions hold:
		\begin{enumerate}
			\item[(i)] $P_0^F \lVert b_n^{g_N} - b_n^{g^*} \rVert_2^2 = O_{P^F_0}(N^{-\gamma \delta_1})$ and $P_0^F \lVert c_n^{g_N} - c_n^{g^*} \rVert_F = O_{P^F_0}(N^{-\gamma \delta_2})$
			\item[(ii)] $P_0^F \lVert b_n^{g_N} - b_n^{g^*} \rVert_2^4 = O_{P^F_0}(1)$ and $P_0^F \lVert c_n^{g_N} - c_n^{g^*} \rVert_F^2 = O_{P^F_0}(1)$,
		\end{enumerate}
		then
		\begin{align}
			\lVert \Hat{\Psi}^{g_{N}}(\Tilde{\mathbb{P}}_{N, n}) - \Hat{\Psi}^{g_{N}}(\Hat{\mathbb{P}}_{N, n}) \rVert_2 & = o_{P^F_0}(N^{-1/2}).
		\end{align}
	\end{restatable}
	
	\begin{remark}[rates of convergence]
		We discuss the consequences of Lemma \ref{lemma:new-asymptotics-plug-in} for special values of $\delta_1$ and $\delta_2$:
		\begin{itemize}
			\item $\delta_1 = 1$ and $\delta_2 = 0$. Then $P_0^F\lVert b_n^g \rVert_2^2 = O(n^{-1})$, so the root-$n$ standardized treatment-effect bias decays at rate\footnote{$P^F_0 \lVert b^g_n \rVert_2 \le (P^F_0 \lVert b^g_n \rVert_2^2)^{1/2} = O(n^{-\delta_1 / 2})$.} $n^{-1/2}$, while the bias in the within-trial covariance estimator is only required to remain bounded. In this case, Lemma \ref{lemma:new-asymptotics-plug-in} requires $\gamma > 1/2$. As discussed in Remark \ref{remark:naive-estimator-rate}, this regime is not desirable because it makes the within-trial estimation error asymptotically negligible.
			\item $\delta_1 > 1$ and $\delta_2 > 0$. Then the lemma allows $\gamma < 1/2$; within-trial sample sizes may grow slower than $N^{1/2}$. In this regime, the within-trial estimation error is not asymptotically negligible, as discussed in Remark \ref{remark:naive-estimator-rate}.
		\end{itemize}
	\end{remark}

	\begin{remark}[non-trivial rates of convergence]\label{remark:naive-estimator-rate}
			To determine which values for $\gamma$ are reasonable, we can impose the condition on $\gamma$ that the corresponding asymptotic regime does not lead to valid ``unadjusted'' estimators that ignore sampling variability in the trial-level treatment effect estimators. Specifically, the estimators of the covariance between $\alpha^g$ and $\beta$, based on (i) the unobservable $(\alpha^g, \beta)'$ and (ii) the unbiased within-trial estimators $(\Hat{\alpha}^g_n, \Hat{\beta}_n)'$ should not be asymptotically equivalent in the sense that 
			\begin{equation*}
				N^{1/2} \left\{ \widehat{\text{Cov}}(\Hat{\alpha}^g_n, \Hat{\beta}_n)_N - \widehat{\text{Cov}}(\alpha^g, \beta)_N \right\}
			\end{equation*}
			may not be $o_{P^F_0}(1)$.
			
			We now show that $\widehat{\text{Cov}}(\Hat{\alpha}^g_n, \Hat{\beta}_n)_N$ and $\widehat{\text{Cov}}(\alpha^g, \beta)_N$ are asymptotically equivalent in the above sense if $\gamma > \frac{1}{2}$.
			Consider the setting where $\mu_{\alpha, 0}^g := E_{P^F_0}(\alpha^g)$ and $\mu_{\beta, 0} := E_{P^F_0}(\beta)$ are known, $(\Hat{\alpha}^g_n, \Hat{\beta}_n)'$ are finite-sample unbiased, and $P^F_0 \lVert \Sigma^g_n \rVert_F = O(1)$. For simplicity, we use the known $\mu_{\alpha, 0}^g$ and $\mu_{\beta, 0}$ instead of the corresponding sample means in $\widehat{\text{Cov}}(\Hat{\alpha}^g_n, \Hat{\beta}_n)_N$ and $\widehat{\text{Cov}}(\alpha^g, \beta)_N$. The root-$N$ standardized difference is then as follows\footnote{We abuse notation here and let $\mathbb{P}_N$ be the empirical distribution of the (unobserved) data $(\Hat{\alpha}^g_n, \Hat{\beta}_n, \Hat{\Sigma}^g_n, \alpha^g, \beta)$.}:
			\begin{align*}
				& N^{1/2} \left\{ \widehat{\text{Cov}}(\Hat{\alpha}^g_n, \Hat{\beta}_n)_N - \widehat{\text{Cov}}(\alpha^g, \beta)_N\right\} \\
				= & N^{1/2} \mathbb{P}_N \left\{ (\Hat{\alpha}_n^g - \mu_{\alpha, 0}^g) (\Hat{\beta}_n - \mu_{\beta, 0}) - (\alpha^g - \mu_{\alpha, 0}^g) (\beta - \mu_{\beta, 0}) \right\} \\
				= & N^{1/2} \left\{ \mathbb{P}_N (\Hat{\alpha}_{n}^g - \alpha^g)(\Hat{\beta}_{n} - \beta)  + \mathbb{P}_N (\Hat{\alpha}_{n}^g - \alpha^g) (\beta - \mu_{\beta, 0}) + \mathbb{P}_N (\alpha^g - \mu_{\alpha, 0}^g) (\Hat{\beta}_{n} - \beta) \right\}
			\end{align*}
			The second and third terms above are sample means of mean-zero variables, the variance of which decreases to zero as $n \to \infty$; hence, these terms are asymptotically negligible for any $\gamma > 0$. 
			The first term above, however, is only asymptotically negligible under further conditions on $\gamma$. Let $n^{-1} \sigma_{\Hat{\alpha}, \Hat{\beta}, n}^g$ be the within-trial covariance between $(\Hat{\alpha}^g_{n}, \Hat{\beta}_{n})$; we then have that
			\begin{align*}
				N^{1/2}\mathbb{P}_N (\Hat{\alpha}_{n}^g - \alpha^g)(\Hat{\beta}_{n} - \beta) & = N^{1/2} (\mathbb{P}_N - P^F_0)(\Hat{\alpha}_{n}^g - \alpha^g)(\Hat{\beta}_{n} - \beta) + N^{1/2} n^{-1} P_0^F \sigma_{\Hat{\alpha}, \Hat{\beta}, n}^g \\
				& = N^{1/2} V_{N, n} + N^{1/2 - \gamma} P_0^F \sigma_{\Hat{\alpha}, \Hat{\beta}, n}^g.
			\end{align*}
			The term $V_{N, n}$ is $o_{P_0^F}(N^{-1/2})$ as long as $n \to \infty$ because it is the sample mean of a mean-zero random variable with a variance that decreases to zero as $n \to \infty$. The second term is $o_{P^F_0}(1)$ if $\gamma > \frac{1}{2}$; hence, the naive estimator of the covariance is asympotically equivalent to the estimator based on the true treatment effects if $\gamma > \frac{1}{2}$.
		\end{remark}
	
	\subsection{Practical Relevance}\label{appendix:practical-relevance}
	
	By allowing the within-trial sample sizes to increase with the number of independent trials, our approach becomes more practically useful because we are allowed to use more within-trial estimators. However, this comes at the cost of requiring within-trial sample sizes to increase sufficiently fast as a function of $N$. In most meta-analyses, $n$ is very large relative to $N$; hence, the asymptotic regime considered in this section is reasonable. Two caveats should be raised here, however. 
	First, all these results are asymptotic in the sense that $N \to \infty$ regardless of the asymptotic regime. When $N$ is small---as is almost always the case in a meta-analysis---one should interpret the results carefully. 
	Second, when $n$ is not much larger than $N$, one should interpret the results with even more suspicion as the current asymptotic regime may not hold up very well. 

	To some degree, the arguments and results in this section are not satisfactory. In Sections \ref{sec:ma-surrogate-index} and \ref{sec:inference-data-adaptive-estimand}, we derived the asymptotic distribution of the plug-in estimator $\Hat{\Psi}^{g}(\Hat{\mathbb{P}}_{N, n})$ and $\Hat{\Psi}^{g_{N}}(\Hat{\mathbb{P}}_{N, n})$ when $n$ is fixed and $N \to \infty$. 
	As already mentioned in Remark \ref{remark:limiting-distribution}, when $n$ increases with $N$, the limiting distributions may differ from the ones derived in Sections \ref{sec:ma-surrogate-index} and \ref{sec:inference-data-adaptive-estimand}.
	This is indeed the case as the limiting distributions of
	\begin{equation*}
		N^{1/2} \left( \Hat{\Psi}^{g}(\Hat{\mathbb{P}}_{N, n}) - \Psi^F(P^F_0) \right) \quad \text{and} \quad N^{1/2} \left( \widehat{\text{Cov}}(\alpha^g, \beta) - \Psi^F(P^F_0) \right)
	\end{equation*}
	are the same if $\gamma > 0$ and provided $P_0^F \left(\Hat{\sigma}_{\Hat{\alpha}, \Hat{\beta}, n}^g\right)^2 = O(1)$. Indeed, following similar derivations as in Remark \ref{remark:naive-estimator-rate}, it follows that 
	\begin{align*}
		N^{1/2} \left( \Hat{\Psi}^{g}(\Hat{\mathbb{P}}_{N, n}) - \widehat{\text{Cov}}(\alpha^g, \beta) \right) & = N^{1/2} n^{-1} \left( \mathbb{P}_N - P^F_0 \right) \Hat{\sigma}_{\Hat{\alpha}, \Hat{\beta}, n}^g + o_{P^F_0}(1)
	\end{align*}
	whenever $\gamma > 0$. The term $N^{1/2} n^{-1} \left( \mathbb{P}_N - P^F_0 \right) \Hat{\sigma}_{\Hat{\alpha}, \Hat{\beta}, n}^g$ is also $o_{P^F_0}(1)$ because it is the root-$N$ standardized mean of a mean-zero random variable with variance converging to zero.
	Hence, in the current asymptotic regime, the measurement error in the trial-level treatment effects has no effect on the asymptotic distribution if we adjust for the measurement error\footnote{This only holds when we effectively adjust for measurement error because, as shown in Remark \ref{remark:naive-estimator-rate}, not adjusting for the measurement error when $\gamma \le \frac{1}{2}$ can lead to incorrect inferences.}. 
	While this result is not fully satisfactory, it is not surprising that the measurement error becomes negligible if $n, N \to \infty$ because the variance of the measurement error is $O(n^{-1})$. 
	At the same time, if one does not want to work under the regime where $n, N \to \infty$, there seems no way forward without making Assumptions \ref{assumption:unbiased-trt-effect} and \ref{assumption:unbiased-sampling-variance} (or other similar assumptions).

	
	\appendixclearpage
	\section{Proofs}\label{appendix:proofs}
	

	\subsection{Proofs of Results in Main Text}\label{appendix:proofs-main}
	
	\propositionidmeanvariance*
	
	\begin{proof}
		Let $P^F \in \mathcal{M}^F$ be an arbitrary full-data distribution where all elements of $\mathcal{M}^F$ satisfy Assumptions \ref{assumption:unbiased-trt-effect} and \ref{assumption:unbiased-sampling-variance}, and let $P := h(P^F)$ be the corresponding observed-data distribution. 
		Define $U^g := \Hat{\alpha}^g - \alpha^g$. By Assumptions \ref{assumption:unbiased-trt-effect} and \ref{assumption:unbiased-sampling-variance},
		\begin{equation*}
			E_{P^F}(U^g \mid T) = 0,
			\qquad
			\operatorname{Var}_{P^F}(U^g \mid T) = n^{-1}\Sigma_{11}^g,
			\qquad
			E_{P^F}(n^{-1}\Hat{\Sigma}_{11}^g \mid T) = n^{-1}\Sigma_{11}^g.
		\end{equation*}
		Here, each conditional quantity corresponds to the notation in the assumptions as follows: 
		$E_{P^F}(\Hat{\alpha}^g \mid T)=E_{P^T,n}(\Hat{\alpha}^g)$, 
		$\operatorname{Var}_{P^F}(\Hat{\alpha}^g \mid T)=\operatorname{Var}_{P^T,n}(\Hat{\alpha}^g)$, and 
		$E_{P^F}(\Hat{\Sigma}_{11}^g \mid T)=E_{P^T,n}(\Hat{\Sigma}_{11}^g)$ with $n=n(T)$.

		We first show that $\Psi^{F, g}_1(P^F) = \Psi^g_1(P)$. The derivations for $\Psi^{F, g}_2(P^F) = \Psi^g_2(P)$ are analogous.
		\begin{align*}
			\Psi^{F, g}_1(P^F) & = E_{P^F} \left( \alpha^g \right) \\
			& = E_{P^F} \left\{ E_{P^F} \left( \Hat{\alpha}^g \mid T \right) \right\} &&~\mbox{(assumption \ref{assumption:unbiased-trt-effect})}\\
			& = E_{P^F} \left( \Hat{\alpha}^g \right) &&~\mbox{(iterated expectation)}\\
			& = E_{P} \left( \Hat{\alpha}^g \right) &&~\mbox{($P = h(P^F)$)}\\
			& = \Psi^g_1(P)
		\end{align*}

		We now show that $\Psi^{F, g}_3(P^F) = \Psi^g_3(P)$. The derivations for $\Psi^{F, g}_4(P^F) = \Psi^g_4(P)$ and $\Psi^{F, g}_5(P^F) = \Psi^g_5(P)$ are analogous.
		\begin{align*}
			\Psi^{F, g}_3(P^F) & = \operatorname{Var}_{P^F} \left(  \alpha^g \right) \\
			& = \operatorname{Var}_{P^F} \left( \Hat{\alpha}^g - U^g \right) \\
			& = \operatorname{Var}_{P^F} \left( \Hat{\alpha}^g \right) + \operatorname{Var}_{P^F} \left( U^g \right) - 2\operatorname{Cov}_{P^F} \left( \Hat{\alpha}^g, U^g \right) \\
			& = \operatorname{Var}_{P^F} \left( \Hat{\alpha}^g \right) + \operatorname{Var}_{P^F} \left( U^g \right) - 2\operatorname{Cov}_{P^F} \left( \alpha^g + U^g, U^g \right) \\
			& = \operatorname{Var}_{P^F} \left( \Hat{\alpha}^g \right) - \operatorname{Var}_{P^F} \left( U^g \right) - 2\operatorname{Cov}_{P^F} \left( \alpha^g, U^g \right) \\
			& = \operatorname{Var}_{P^F} \left( \Hat{\alpha}^g \right) - \operatorname{Var}_{P^F} \left( U^g \right),
		\end{align*}
		where
		\begin{equation*}
			\operatorname{Cov}_{P^F}(\alpha^g, U^g)
			= E_{P^F}(\alpha^g U^g)
			= E_{P^F}\!\left\{\alpha^g E_{P^F}(U^g \mid T)\right\}
			= 0,
		\end{equation*}
		since $\alpha^g$ is a function of $T$ and $E_{P^F}(U^g \mid T) = 0$. Hence, $\Psi^{F, g}_3(P^F) = \Psi^g_3(P)$ follows if we can show that $\operatorname{Var}_{P^F} \left( U^g \right) = E_{P} \left( n^{-1} \cdot \Hat{\Sigma}_{11}^g \right)$, which we now do:
		\begin{align*}
			\operatorname{Var}_{P^F} \left( U^g \right) & = E_{P^F} \left\{ \operatorname{Var}_{P^F}(U^g \mid T) \right\} &&~\mbox{(law of total variance and }E_{P^F}(U^g\mid T)=0\mbox{)}\\
			& = E_{P^F} \left( n^{-1} \Sigma_{11}^g \right) \\
			& = E_{P^F} \left( E_{P^F} \left( n^{-1} \Hat{\Sigma}_{11}^g \mid T \right) \right) &&~\mbox{(assumption \ref{assumption:unbiased-sampling-variance})}\\
			& = E_{P^F} \left( n^{-1} \Hat{\Sigma}_{11}^g \right) &&~\mbox{(iterated expectation)}\\
			& = E_{P} \left( n^{-1} \Hat{\Sigma}_{11}^g \right) &&~\mbox{($P = h(P^F)$)}.
		\end{align*}
	\end{proof}

	\propositionnonparametricmamodel*

	\begin{proof}
		This follows immediately from corollary \ref{corollary:non-parametric-ma-model}.
	\end{proof}

	\subsection{Proofs of Results in the Appendices}\label{appendix:proofs-additional}

	\subsubsection{Appendix \ref{appendix:plug-in-estimator}}

	\propositionnonparametricmamodelb*

	\begin{proof}
		Let $\boldsymbol{X} := (\Hat{X}, \Hat{\Sigma})$ and $\boldsymbol{Z} := (X, \Sigma)$, and let $Q_0 := Q_{P^F_0}$ be the distribution of $\boldsymbol{Z}$ under $P^F_0$. The model for $Q$ implied by $\mathcal{M}^F$ is denoted by $\mathcal{Q} := \left\{ Q_{P^F}: P^F \in \mathcal{M}^F \right\}$.
		By Lemma \ref{lemma:local-non-parametric-q}, $\mathcal{Q}$ is locally non-parametric at $Q_0$ and hence $\overline{\mathcal{T}}(Q_0) = L^2_0(Q_0)$.

		The density of $\boldsymbol{X} \mid \boldsymbol{Z}$ with respect to the dominating measure $\lambda_{\mathbb{R}^k} \otimes \delta_{\Sigma}$, where $\lambda_{\mathbb{R}^k}$ is the Lebesgue measure on $\mathbb{R}^k$ and $\delta_{\Sigma}$ is the Dirac measure at $\Sigma$, is given by
		\begin{equation*}
			p_{\boldsymbol{X} \mid \boldsymbol{Z}}(\Hat{x}, \Hat{\Sigma} \mid \boldsymbol{Z} = (x, \Sigma)) = \frac{1}{(2\pi)^{k/2} |\Sigma|^{1/2}} \exp\left\{-\frac{1}{2} (\Hat{x} - x)^\top \Sigma^{-1} (\Hat{x} - x)\right\}.
		\end{equation*}

		Note that the observed-data distribution only depends on $P^F$ through $Q_{P^F}$; hence, the observed-data distributions can also be indexed by $Q$. Let $b(X, \Sigma)$ be a score for $Q_0$ in $\mathcal{Q}$. The tangent set of $\mathcal{M}$ at $P_0$ is given by \parencite[Example 25.35]{van2000asymptotic}:
		\begin{equation*}
			\mathcal{T}(P_0) = \left\{ A_{Q_0} b: b \in L^2_0(Q_0) \right\},
		\end{equation*}
		where $A_{Q_0}$ is the score operator that maps a score $b$ for the unobserved-data model $\mathcal{Q}$ to the corresponding score for the observed-data model $\mathcal{M}$. This operator is given by
		\begin{equation*}
			(A_{Q_0} b) (\boldsymbol{x}) = E_{P^F_0} \left\{ b(\boldsymbol{Z}) \mid \boldsymbol{X} = \boldsymbol{x} \right\}.
		\end{equation*}

		The proof is complete if we can show that $\overline{\mathcal{T}}(P_0) = L^2_0(P_0)$. This holds if the range of the score operator $A_{Q_0}: L^2_0(Q_0) \to L^2_0(P_0)$, denoted by $R(A_{Q_0})$, is dense in $L^2_0(P_0)$, or equivalently, its closure $\overline{R}(A_{Q_0})$ is equal to $L^2_0(P_0)$.
		Because $A_{Q_0}$ is a bounded linear operator from $L^2_0(Q_0)$ to $L^2_0(P_0)$, we have that $\overline{R}(A_{Q_0}) = N(A^*)^{\perp}$, where $N(A^*)$ is the null space of the adjoint operator $A^*: L^2_0(P_0) \to L^2_0(Q_0)$, which is defined as follows\footnote{The adjoint is not indexed by $Q_0$ because the conditional distribution $\boldsymbol{X} \mid \boldsymbol{Z}$ is fixed.}:
		\begin{equation*}
			(A^* g)(\boldsymbol{z}) = E \left\{ g(\boldsymbol{X}) \mid \boldsymbol{Z} = \boldsymbol{z} \right\}.
		\end{equation*}
		Hence, we only need to show that $N(A^*) = \{0\}$.
		To this end, let $g \in L^2_0(P_0)$ satisfy $(A^* g)(\boldsymbol{Z}) = 0$ $Q_0$-a.s., i.e.,
		\begin{equation*}
			E \left\{ g(\boldsymbol{X}) \mid \boldsymbol{Z} = (x, \Sigma) \right\} = 0 \quad Q_0\text{-a.s.}
		\end{equation*}
		We must now show that $g(\boldsymbol{X}) = 0$ $P_0$-a.s. We can rewrite the above conditional expectation as follows, also using the assumption that $\Hat{\Sigma} = \Sigma$:
		\begin{equation*}
			E \left\{ g(\Hat{X}, \Sigma) \mid \boldsymbol{Z} = (x, \Sigma) \right\} = 0 \quad Q_0\text{-a.s.}.
		\end{equation*}
		The above expectation is with respect to a multivariate normal distribution with mean $x$ and variance $\Sigma$.
		Define $m_{\Sigma}(x) := E \left\{ g(\Hat{X}, \Sigma) \mid \boldsymbol{Z} = (x, \Sigma) \right\}$. 
		Because the above conditional expectation is zero $Q_0$-a.s., it follows that for $Q_0$-a.e.~$\Sigma$, we have $m_{\Sigma}(x) = 0$ for $Q_0(\cdot \mid \Sigma)$-a.e. $x$.
		Now define 
		\begin{align*}
			\mathcal{S}_1 & := \left\{ \Sigma: m_{\Sigma}(x) = 0 \text{ for } Q_0(\cdot \mid \Sigma)\text{-a.e. } x \right\} \\
			\mathcal{S}_2 & := \left\{ \Sigma: \text{the support of } Q_0(\cdot \mid \Sigma) \text{ contains a nonempty open set in } \mathbb{R}^k \right\} \\
			\mathcal{S}_3 & := \left\{ \Sigma: \Sigma \text{ is positive definite} \right\} \\
			\mathcal{S}_4 & := \left\{ \Sigma: E\left\{ g(\Hat{X}, \Sigma)^2 \mid \boldsymbol{Z} = (x, \Sigma) \right\} < \infty \text{ for } Q_0(\cdot \mid \Sigma)\text{-a.e. } x \right\},
		\end{align*}
		where $Q_0(\mathcal{S}_1 \cap \mathcal{S}_2 \cap \mathcal{S}_3 \cap \mathcal{S}_4) = 1$.
		Indeed, since $g(\boldsymbol{X}) \in L^2(P_0)$, we have that $E\left\{ g(\Hat{X}, \Sigma)^2 \mid \boldsymbol{Z} \right\} < \infty$ $Q_0$-a.s. and thus $Q_0(\mathcal{S}_4) = 1$.

		Fix $\Sigma \in \mathcal{S}_1 \cap \mathcal{S}_2 \cap \mathcal{S}_3 \cap \mathcal{S}_4$. The set $\{x: m_{\Sigma}(x) = 0\}$ has full $Q_0(\cdot \mid \Sigma)$-measure and is therefore dense in a nonempty open subset $\mathcal{X}$ of the support of $Q_0(\cdot \mid \Sigma)$\footnote{Otherwise, there would exist an open subset of the support whose intersection with $\{x: m_{\Sigma}(x) = 0\}$ is empty. By definition of the support, this open subset has non-zero probability, which contradicts $\{x: m_{\Sigma}(x) = 0\}$ having measure one.}; denote this set by $\mathcal{A}_{\Sigma}$. 
		By definition of $\mathcal{S}_4$, there exists $x_0$ such that $E\left\{ g(\Hat{X}, \Sigma)^2 \mid \boldsymbol{Z} = (x_0, \Sigma) \right\} < \infty$. Lemma \ref{lemma:gaussian-smoothing-continuity} therefore implies that $m_{\Sigma}$ is continuous, which implies that $m_{\Sigma}(x) = 0$ for all $ \in \mathcal{A}_{\Sigma}$.
		The completeness of the multivariate normal location family with fixed covariance $\Sigma$ implies $g(\Hat{X}, \Sigma) = 0$ $\lambda$-a.e.~\parencite[Theorem 4.3.1]{lehmann_testing_2008}, where $\lambda$ is the Lebesgue measure.
		Since $P_0(\cdot \mid \Sigma)$ has a density with respect to $\lambda$, this implies that $g(\Hat{X}, \Sigma) = 0$ $P_0(\cdot \mid \Sigma)$-a.s. This holds for $P_0$-a.e.~$\Sigma$, so $g(\Hat{X}, \Sigma) = 0$ $P_0$-a.s. 
		
		Hence, $N(A^*) = \{0\}$ and therefore $\overline{\mathcal{T}}(P_0) = L^2_0(P_0)$.
	\end{proof}

	\corollarynonparametricmamodel*

	\begin{proof}
		The observed-data model under the conditions of this corollary is larger than the observed-data model under the conditions of Proposition \ref{proposition:non-parametric-ma-modelb} because fewer restrictions are imposed on $\mathcal{M}^F$. The latter observed-data model was shown to be locally non-parametric. Hence, the former observed-data model is also non-parametric.
	\end{proof}

	\subsubsection{Appendix \ref{appendix:discussion-conditions-theorem-bootstrap-for-data-adaptive-estimands}}

	\theorembootstrapfordataadaptiveestimands*

	\begin{proof}
		We split the proof into three parts. First, we show that $\Hat{\boldsymbol{\theta}}^{g_N}_N$ and $\Tilde{\boldsymbol{\theta}}^{g_N}_N$ are consistent estimators of both $\boldsymbol{\theta}^{g^*}_0$ and $\boldsymbol{\theta}^{g_N}_0$. Second, we prove the convergence in distribution of $N^{1/2} (\Hat{\boldsymbol{\theta}}^{g_N}_N - \boldsymbol{\theta}^{g_N}_0)$. Third, we prove the bootstrap result. 

		Throughout this proof, $o_P(1)$ denotes convergence in $P$-probability under the joint law of the observed data and bootstrap weights. By Lemma \ref{lemma:conditional-convergence-in-p}, this is equivalent to convergence in conditional probability given the observed data.

		\paragraph{Part 1: Consistency}

		We first show consistency of $\Hat{\boldsymbol{\theta}}^{g_N}_N$ for the non-data-adaptive estimand $\boldsymbol{\theta}^{g^*}_0$; afterwards, we show that this also implies that $\lVert \Hat{\boldsymbol{\theta}}^{g_N}_N - \boldsymbol{\theta}^{g_N}_0 \rVert_2 = o_P(1)$. We then repeat similar steps for the bootstrapped estimator $\Tilde{\boldsymbol{\theta}}^{g_N}_N$. 
		\begin{align*}
			\lVert \Phi^{g_N}(\Hat{\boldsymbol{\theta}}^{g_N}_N) \rVert_2 & = \lVert \Phi^{g_N}(\Hat{\boldsymbol{\theta}}^{g_N}_N) - \mathbb{P}_N \phi^{g_N}_{\Hat{\boldsymbol{\theta}}^{g_N}_N} + \mathbb{P}_N \phi^{g_N}_{\Hat{\boldsymbol{\theta}}^{g_N}_N} \rVert_2 \\
			& = \lVert P_0 \phi^{g_N}_{\Hat{\boldsymbol{\theta}}^{g_N}_N} - \mathbb{P}_N \phi^{g_N}_{\Hat{\boldsymbol{\theta}}^{g_N}_N} + \mathbb{P}_N \phi^{g_N}_{\Hat{\boldsymbol{\theta}}^{g_N}_N} \rVert_2 \\
			& \le \sup_{\boldsymbol{\theta} \in \Theta, g \in \mathcal{G}} \lVert P_0 \phi^{g}_{\boldsymbol{\theta}} - \mathbb{P}_N \phi^{g}_{\boldsymbol{\theta}} \rVert_2 +  \lVert \mathbb{P}_N \phi^{g_N}_{\Hat{\boldsymbol{\theta}}^{g_N}_N}\rVert_2 \\
			& \le \sup_{\boldsymbol{\theta} \in \Theta, g \in \mathcal{G}} \lVert P_0 \phi^{g}_{\boldsymbol{\theta}} - \mathbb{P}_N \phi^{g}_{\boldsymbol{\theta}} \rVert_2 + o_P(N^{-1/2}) &&\mbox{(condition (v))} \\
			& \le o_P(1) + o_P(N^{-1/2})  = o_P(1) &&\mbox{(condition (ii))} 
		\end{align*}
		We use the above result in the following derivations, note that we used $g_N$ above and are using $g^*$ below.
		\begin{align*}
			\lVert \Phi^{g^*}(\Hat{\boldsymbol{\theta}}^{g_N}_N) \rVert_2 & = \lVert P_0 \phi^{g^*}_{\Hat{\boldsymbol{\theta}}^{g_N}_N} - P_0 \phi^{g_N}_{\Hat{\boldsymbol{\theta}}^{g_N}_N} + P_0 \phi^{g_N}_{\Hat{\boldsymbol{\theta}}^{g_N}_N} \rVert_2 \\
			& \le \lVert P_0 \phi^{g^*}_{\Hat{\boldsymbol{\theta}}^{g_N}_N} - P_0 \phi^{g_N}_{\Hat{\boldsymbol{\theta}}^{g_N}_N} \rVert_2 +  o_P(1) &&\mbox{(previous display)}\\
			& \le L \cdot d(g^*, g_N) + o_P(1) && \mbox{(condition (vi))}\\
			& = o_P(1)
		\end{align*}
		By condition (i), $\lVert \Phi^{g^*}(\Hat{\boldsymbol{\theta}}^{g_N}_N) \rVert_2 = o_P(1)$ implies that $\lVert \Hat{\boldsymbol{\theta}}^{g_N}_N - \boldsymbol{\theta}^{g^*}_0 \rVert_2 = o_P(1)$. Lemma \ref{lemma:consistency-general-plugin-estimator} then implies that $\lVert \Hat{\boldsymbol{\theta}}^{g_N}_N - \boldsymbol{\theta}^{g_N}_0 \rVert_2 = o_P(1)$.

		We now repeat the same steps for the bootstrapped estimator $\Tilde{\boldsymbol{\theta}}^{g_N}_N$. 
		\begin{align*}
			\lVert \Phi^{g_N}(\Tilde{\boldsymbol{\theta}}^{g_N}_N) \rVert_2 & = \lVert P_0 \phi^{g_N}_{\Tilde{\boldsymbol{\theta}}^{g_N}_N} - \Tilde{\mathbb{P}}_N \phi^{g_N}_{\Tilde{\boldsymbol{\theta}}^{g_N}_N} + \Tilde{\mathbb{P}}_N \phi^{g_N}_{\Tilde{\boldsymbol{\theta}}^{g_N}_N} \rVert_2 \\
			& \le \sup_{\boldsymbol{\theta} \in \Theta, g \in \mathcal{G}} \lVert P_0 \phi^{g}_{\boldsymbol{\theta}} - \Tilde{\mathbb{P}}_N \phi^{g}_{\boldsymbol{\theta}} \rVert_2 +  \lVert \Tilde{\mathbb{P}}_N \phi^{g_N}_{\Tilde{\boldsymbol{\theta}}^{g_N}_N}\rVert_2 \\
			& = \sup_{\boldsymbol{\theta} \in \Theta, g \in \mathcal{G}} \lVert P_0 \phi^{g}_{\boldsymbol{\theta}} - \Tilde{\mathbb{P}}_N \phi^{g}_{\boldsymbol{\theta}} \rVert_2 + o_P(N^{-1/2}) &&\mbox{(condition (v))} \\
			& \le \sup_{\boldsymbol{\theta} \in \Theta, g \in \mathcal{G}} \lVert P_0 \phi^{g}_{\boldsymbol{\theta}} - \mathbb{P}_N \phi^{g}_{\boldsymbol{\theta}} \rVert_2 + \sup_{\boldsymbol{\theta} \in \Theta, g \in \mathcal{G}} \lVert \mathbb{P}_N \phi^{g}_{\boldsymbol{\theta}} - \Tilde{\mathbb{P}}_N \phi^{g}_{\boldsymbol{\theta}} \rVert_2 + o_P(N^{-1/2}) \\
			& = o_P(1) + o_P(1) + o_P(N^{-1/2})
		\end{align*}
		The last equality follows from condition (ii), \textcite[corollary 10.14]{kosorok2008introduction}, and Lemma \ref{lemma:conditional-convergence-in-p}. The remaining arguments are the same as for $\Hat{\boldsymbol{\theta}}^{g_N}_N$ and are omitted to avoid repetition. We conclude that $\lVert \Tilde{\boldsymbol{\theta}}^{g_N}_N - \boldsymbol{\theta}^{g^*}_0 \rVert_2 = o_P(1)$ and $\lVert \Tilde{\boldsymbol{\theta}}^{g_N}_N - \boldsymbol{\theta}^{g_N}_0 \rVert_2 = o_P(1)$.

		Because $\Hat{\boldsymbol{\theta}}^{g_N}_N$ and $\Tilde{\boldsymbol{\theta}}^{g_N}_N$ are consistent for $\boldsymbol{\theta}^{g^*}_0$, we can further assume without loss of generality that $\Hat{\boldsymbol{\theta}}^{g_N}_N$ and $\Tilde{\boldsymbol{\theta}}^{g_N}_N$ lie in a neighborhood of $\boldsymbol{\theta}^{g^*}_0$ where conditions (iii) and (iv) hold. This will be used in the following parts of the proof. 

		\paragraph{Part 2: Convergence in Distribution of $N^{1/2} (\Hat{\boldsymbol{\theta}}^{g_N}_N - \boldsymbol{\theta}^{g_N}_0)$.}

		From condition (iii) and \textcite[lemma 13.3]{kosorok2008introduction} follows that $N^{1/2} \left(\mathbb{P}_N - P_0 \right) \phi^{g_N}_{\Hat{\boldsymbol{\theta}}^{g_N}_N} - N^{1/2} \left(\mathbb{P}_N - P_0 \right) \phi^{g^*}_{\boldsymbol{\theta}^{g^*}_0} = o_P(1)$. We start from this result in the following derivations:
		\begin{align*}
			N^{1/2} \left(\mathbb{P}_N - P_0 \right) \phi^{g^*}_{\boldsymbol{\theta}^{g^*}_0} & = N^{1/2} \mathbb{P}_N \phi^{g_N}_{\Hat{\boldsymbol{\theta}}^{g_N}_N} - N^{1/2} P_0 \phi^{g_N}_{\Hat{\boldsymbol{\theta}}^{g_N}_N}  + o_P(1) \\
			& = - N^{1/2} P_0 \phi^{g_N}_{\Hat{\boldsymbol{\theta}}^{g_N}_N}  + o_P(1) &&\mbox{(condition (v))} \\
			& = - N^{1/2} P_0 \phi^{g_N}_{\boldsymbol{\theta}^{g_N}_0} - N^{1/2} \dot{\Phi}^{g_N}(\bar{\boldsymbol{\theta}}_N) \cdot \left( \Hat{\boldsymbol{\theta}}^{g_N}_N - \boldsymbol{\theta}^{g_N}_0  \right) + o_P(1) &&\mbox{(mean value theorem + condition (iv))} \\
			& = - N^{1/2} \dot{\Phi}^{g_N}(\bar{\boldsymbol{\theta}}_N) \cdot \left( \Hat{\boldsymbol{\theta}}^{g_N}_N - \boldsymbol{\theta}^{g_N}_0  \right) + o_P(1)
		\end{align*}
		where $\bar{\boldsymbol{\theta}}_N$ is a point between $\Hat{\boldsymbol{\theta}}^{g_N}_N$ and $\boldsymbol{\theta}^{g_N}_0$. Note that we're abusing notation here if $p > 1$ because the mean-value theorem only exists for real-valued functions. In that case, the mean-value theorem is applied to each element of the vector-valued function $\Phi^{g_N}(\boldsymbol{\theta})$ where $\bar{\boldsymbol{\theta}}_N$ may differ from element to element. Anyhow, any such $\bar{\boldsymbol{\theta}}_N$ converges in probability to $\boldsymbol{\theta}^{g^*}_0$ because it lies between $\Hat{\boldsymbol{\theta}}^{g_N}_N$ and $\boldsymbol{\theta}^{g_N}_0$, and both converge in probability to $\boldsymbol{\theta}^{g^*}_0$. 
		By condition (iv) and the fact that $\lVert \bar{\boldsymbol{\theta}}_N - \boldsymbol{\theta}^{g^*}_0 \rVert_2 + d(g_N, g^*) = o_P(1)$, we have that $\dot{\Phi}^{g_N}(\bar{\boldsymbol{\theta}}_N) = \dot{\Phi}^{g^*}(\boldsymbol{\theta}^{g^*}_0) + o_P(1)$ by the continuous mapping theorem, and thus
		\begin{align*}
			N^{1/2} \left(\mathbb{P}_N - P_0 \right) \phi^{g^*}_{\boldsymbol{\theta}^{g^*}_0} = - N^{1/2} \dot{\Phi}^{g^*}(\boldsymbol{\theta}^{g^*}_0) \cdot \left( \Hat{\boldsymbol{\theta}}^{g_N}_N - \boldsymbol{\theta}^{g_N}_0  \right) + o_P(1)
		\end{align*}
		It further follows by condition (iv) that 
		\begin{equation*}
			N^{1/2} \left( \Hat{\boldsymbol{\theta}}^{g_N}_N - \boldsymbol{\theta}^{g_N}_0 \right) = - \left( \dot{\Phi}^{g^*}(\boldsymbol{\theta}^{g^*}_0) \right)^{-1} \cdot N^{1/2} \left(\mathbb{P}_N - P_0 \right) \phi^{g^*}_{\boldsymbol{\theta}^{g^*}_0} + o_P(1).
		\end{equation*}
		By the central limit theorem and Slutsky's theorem, we thus have that $N^{1/2} \left( \Hat{\boldsymbol{\theta}}^{g_N}_N - \boldsymbol{\theta}^{g_N}_0 \right) \overset{d}{\to} Z$ where $Z \sim \mathcal{N}(0, \Omega^{g^*}_0)$ with
		\begin{equation*}
			\Omega^{g^*}_0 = \left( \dot{\Phi}^{g^*}(\boldsymbol{\theta}^{g^*}_0) \right)^{-1} \cdot P_0 \left( \phi^{g^*}_{\boldsymbol{\theta}^{g^*}_0} \right)^{\otimes2} \cdot \left\{ \left( \dot{\Phi}^{g^*}(\boldsymbol{\theta}^{g^*}_0) \right)^{-1} \right\}^{\top}.
		\end{equation*}

		\paragraph{Part 3: Bootstrap Result}

		By \textcite[p.~198]{kosorok2008introduction} and condition (iii), we have that
		\begin{equation*}
			N^{1/2} \left( \Tilde{\mathbb{P}}_N - \mathbb{P}_N \right) \leadsto \mathbb{G} \quad \text{in} \quad \ell^{\infty}(\mathcal{F})
		\end{equation*}
		where $\mathcal{F} = \left\{ \phi^g_{\boldsymbol{\theta}} : g \in \mathcal{G}, \lVert \boldsymbol{\theta} - \boldsymbol{\theta}^{g^*}_0 \rVert_2 \le \delta \right\}$ and $\mathbb{G}$ is a tight Gaussian process. 
		Since $\Tilde{\boldsymbol{\theta}}^{g_N}_N$ is consistent by part 1, we can apply the same asymptotic equicontinuity argument as in \textcite[Lemma 13.3]{kosorok2008introduction} to the above bootstrapped process:
		\begin{equation*}
			N^{1/2} \left( \Tilde{\mathbb{P}}_N - \mathbb{P}_N \right) \phi^{g_N}_{\Tilde{\boldsymbol{\theta}}^{g_N}_N} = N^{1/2} \left( \Tilde{\mathbb{P}}_N - \mathbb{P}_N \right) \phi^{g^*}_{\boldsymbol{\theta}^{g^*}_0} + o_P(1),
		\end{equation*}
		and we similarly have that
		\begin{equation*}
			N^{1/2} \left(\mathbb{P}_N - P_0 \right) \phi^{g_N}_{\Tilde{\boldsymbol{\theta}}^{g_N}_N} = N^{1/2} \left(\mathbb{P}_N - P_0 \right) \phi^{g^*}_{\boldsymbol{\theta}^{g^*}_0} + o_P(1).
		\end{equation*}
		Combining the previous two displays, we have that 
		\begin{align*}
			N^{1/2} \left(\mathbb{P}_N - P_0 \right) \left( \phi^{g_N}_{\Tilde{\boldsymbol{\theta}}^{g_N}_N} - \phi^{g^*}_{\boldsymbol{\theta}^{g^*}_0} \right)  + N^{1/2} \left( \Tilde{\mathbb{P}}_N - \mathbb{P}_N \right) \left( \phi^{g_N}_{\Tilde{\boldsymbol{\theta}}^{g_N}_N} - \phi^{g^*}_{\boldsymbol{\theta}^{g^*}_0}  \right) = o_P(1).
		\end{align*}
		In the above equation, several terms cancel out. We can rewrite the above equation as follows:
		\begin{align*}
			- N^{1/2} P_0 \phi^{g_N}_{\Tilde{\boldsymbol{\theta}}^{g_N}_N} + N^{1/2}P_0 \phi^{g^*}_{\boldsymbol{\theta}^{g^*}_0} + N^{1/2} \Tilde{\mathbb{P}}_N \phi^{g_N}_{\Tilde{\boldsymbol{\theta}}^{g_N}_N} - N^{1/2} \Tilde{\mathbb{P}}_N \phi^{g^*}_{\boldsymbol{\theta}^{g^*}_0} = o_P(1),
		\end{align*}
		where, additionally, $N^{1/2}P_0 \phi^{g^*}_{\boldsymbol{\theta}^{g^*}_0} = 0$ by definition and $N^{1/2} \Tilde{\mathbb{P}}_N \phi^{g_N}_{\Tilde{\boldsymbol{\theta}}^{g_N}_N} = o_P(1)$ by condition (v).
		Hence, we have that, following similar steps as in part 2,
		\begin{align*}
			N^{1/2} \Tilde{\mathbb{P}}_N \phi^{g^*}_{\boldsymbol{\theta}^{g^*}_0} & = - N^{1/2} P_0 \phi^{g_N}_{\Tilde{\boldsymbol{\theta}}^{g_N}_N} + o_P(1) \\
			& = - N^{1/2} P_0 \phi^{g_N}_{\boldsymbol{\theta}^{g_N}_0} - N^{1/2} \dot{\Phi}^{g_N}(\bar{\boldsymbol{\theta}}_N) \cdot \left( \Tilde{\boldsymbol{\theta}}^{g_N}_N - \boldsymbol{\theta}^{g_N}_0  \right) + o_P(1) &&\mbox{(mean value theorem + condition (iv))} \\
			& = - N^{1/2} \dot{\Phi}^{g_N}(\bar{\boldsymbol{\theta}}_N) \cdot \left( \Tilde{\boldsymbol{\theta}}^{g_N}_N - \boldsymbol{\theta}^{g_N}_0  \right) + o_P(1).
		\end{align*}
		By the same arguments as in part 2, we have that $\dot{\Phi}^{g_N}(\bar{\boldsymbol{\theta}}_N) = \dot{\Phi}^{g^*}(\boldsymbol{\theta}^{g^*}_0) + o_P(1)$. We, therefore, have that 
		\begin{equation*}
			N^{1/2} \left( \Tilde{\boldsymbol{\theta}}^{g_N}_N - \boldsymbol{\theta}^{g_N}_0 \right) = - \left( \dot{\Phi}^{g^*}(\boldsymbol{\theta}^{g^*}_0) \right)^{-1} \cdot N^{1/2} \Tilde{\mathbb{P}}_N \phi^{g^*}_{\boldsymbol{\theta}^{g^*}_0} + o_P(1).
		\end{equation*}
		Finally, we have that 
		\begin{align*}
			N^{1/2} \left( \Tilde{\boldsymbol{\theta}}^{g_N}_N - \Hat{\boldsymbol{\theta}}^{g_N}_N \right) & = N^{1/2} \left( \Tilde{\boldsymbol{\theta}}^{g_N}_N - \boldsymbol{\theta}^{g_N}_0 \right) - N^{1/2} \left( \Hat{\boldsymbol{\theta}}^{g_N}_N - \boldsymbol{\theta}^{g_N}_0 \right) \\
			& = - \left( \dot{\Phi}^{g^*}(\boldsymbol{\theta}^{g^*}_0) \right)^{-1} \cdot N^{1/2} \Tilde{\mathbb{P}}_N \phi^{g^*}_{\boldsymbol{\theta}^{g^*}_0} + o_P(1) + \left( \dot{\Phi}^{g^*}(\boldsymbol{\theta}^{g^*}_0) \right)^{-1} \cdot N^{1/2} \mathbb{P}_N \phi^{g^*}_{\boldsymbol{\theta}^{g^*}_0} + o_P(1) \\
			& = - \left( \dot{\Phi}^{g^*}(\boldsymbol{\theta}^{g^*}_0) \right)^{-1} \cdot N^{1/2} \left( \Tilde{\mathbb{P}}_N - \mathbb{P}_N \right) \phi^{g^*}_{\boldsymbol{\theta}^{g^*}_0} + o_P(1) 
		\end{align*}
		From \textcite[Theorem 2.6]{kosorok2008introduction} and Lemma \ref{lemma:slutsky-conditional} now follows that 
		\begin{equation*}
			N^{1/2} \left( \Tilde{\boldsymbol{\theta}}^{g_N}_N - \Hat{\boldsymbol{\theta}}^{g_N}_N \right) \overset{P}{\underset{W}{\leadsto}} Z,
		\end{equation*}
		where $Z$ is as defined in part 2.
	\end{proof}

	\lemmadonskerconditions*

	\begin{proof}
		We first treat the case where $\mathcal{X} \subset \mathbb{R}$. If instead $\mathcal{X} \subset \mathbb{R}^d$ for some $d < \infty$, one may apply the same argument below after replacing each covariate extraction map $x_i$ by the $d$ real-valued coordinate maps $\pi_k \circ x_i$, $k = 1, \ldots, d$, where $\pi_k$ is the $k$'th coordinate projection.
		Formally, the extraction maps $a_i$, $z_i$, and $x_i$ are understood as measurable projection maps on the underlying (ordered, padded) within-trial sample that generates $\mathbb{P}^T_n$; we write them as functions of $(n, \mathbb{P}^T_n)$ for notational convenience.

		As assumed, $\Hat{\alpha}^g$ only depends on $(n, \mathbb{P}^T_n)$ through $\left\{(x_i, z_i, g \circ a_i)(n, \mathbb{P}^T_n): i = 1, \dots, n\right\}$. Since $n$ is random, this dependence is not expressed in terms of a fixed number of arguments. By condition (1.), we have that $n \le n^* < \infty$ almost surely, and by our padding convention we have $a_i(n, \mathbb{P}^T_n) = a_{\mathrm{pad}}$ (and similarly $x_i(n, \mathbb{P}^T_n) = x_{\mathrm{pad}}$ and $z_i(n, \mathbb{P}^T_n) = 0$) for $i > n$. Hence, there exists a measurable function $\Hat{\alpha}$ of $3 n^* + 1$ arguments such that
		\begin{equation*}
			\Hat{\alpha}^g(n, \mathbb{P}^T_n) = \Hat{\alpha}\left(g \circ a_1(n, \mathbb{P}^T_n), \dots, g \circ a_{n^*}(n, \mathbb{P}^T_n), z_1(n, \mathbb{P}^T_n), \dots, z_{n^*}(n, \mathbb{P}^T_n), x_1(n, \mathbb{P}^T_n), \dots, x_{n^*}(n, \mathbb{P}^T_n), n\right).
		\end{equation*}

		We now define $\mathcal{F}_j$ as follows:
		\begin{equation*}
			\mathcal{F}_j = \begin{cases}
				\mathcal{G} \circ a_j & \text{for } 1 \le j \le n^* \\
				\{z_{j - n^*}\} & \text{for } n^* < j \le 2 n^* \\
				\{x_{j - 2 n^*}\} & \text{for } 2 n^* < j \le 3 n^* \\
				\{(n, \mathbb{P}^T_n) \mapsto n\} & \text{for } j = 3 n^* + 1
			\end{cases}
		\end{equation*}
		The class $\left\{ \Hat{\alpha}^g: g \in \mathcal{G}\right\}$ is smaller than $\Hat{\alpha} \circ (\mathcal{F}_1 \times \dots \times \mathcal{F}_{3 n^* + 1}) := \left\{ \Hat{\alpha} \circ f: f \in \mathcal{F}_1 \times \dots \times \mathcal{F}_{3 n^* + 1} \right\}$ because the former restricts the first $n^*$ elements of $f$ to be $\left( g \circ a_1, \dots, g \circ a_{n^*} \right)$ for $g \in \mathcal{G}$, whereas the latter allows for a different $g \in \mathcal{G}$ in each of the first $n^*$ elements.
		Hence, it is sufficient to show that the latter (larger) class is $P_0$-Donsker.

		We can apply \textcite[Theorem 9.31]{kosorok2008introduction} to conclude that $\Hat{\alpha} \circ (\mathcal{F}_1 \times \dots \times \mathcal{F}_{3 n^* + 1})$ is a $P_0$-Donsker class if the following conditions hold:
		\begin{enumerate}
			\item[(i)] $\lVert P_0 \rVert_{\mathcal{F}_j} < \infty$ for all $j = 1, \dots, 3 n^* + 1$.
			\item[(ii)] $\Hat{\alpha} \circ f$ is square-integrable for at least one $f \in \mathcal{F}_1 \times \dots \times \mathcal{F}_{3 n^* + 1}$.
			\item[(iii)] $\left| \Hat{\alpha} \circ f(n, \mathbb{P}^T_n) - \Hat{\alpha} \circ h(n, \mathbb{P}^T_n) \right|^2 \le c^2 \sum_{j = 1}^{3 n^* + 1} \left( f_j(n, \mathbb{P}^T_n) - h_j(n, \mathbb{P}^T_n) \right)^2$ for every $f, h \in \mathcal{F}_1 \times \dots \times \mathcal{F}_{3 n^* + 1}$ and some constant $c < \infty$.
			\item[(iv)] $\mathcal{F}_j$ is a $P_0$-Donsker class for all $j = 1, \dots, 3 n^* + 1$.
		\end{enumerate}
		Condition (i) holds because $\mathcal{G}$ is uniformly bounded by condition (5.), $\mathcal{Z} = \left\{0, 1\right\}$ is bounded, $\mathcal{X}$ is bounded by condition (2.), and $n \le n^*$ by condition (1.). 
		Condition (ii) holds by condition (3.): take $g_0 \in \mathcal{G}$ such that $P_0 |\Hat{\alpha}^{g_0}|^2 < \infty$, and define $f \in \mathcal{F}_1 \times \dots \times \mathcal{F}_{3 n^* + 1}$ by $f_j := g_0 \circ a_j$ for $j \le n^*$ and by taking $f_j$ equal to the unique element of $\mathcal{F}_j$ for $j > n^*$. Then $\Hat{\alpha} \circ f = \Hat{\alpha}^{g_0}$ is square-integrable.
		Condition (iii) follows from condition (4.) because, for $f, h \in \mathcal{F}_1 \times \dots \times \mathcal{F}_{3 n^* + 1}$, we have $f_j = h_j$ for $j > n^*$ (the corresponding classes are singletons). Hence, the Lipschitz bound in condition (4.) yields
		\begin{equation*}
			\left| \Hat{\alpha} \circ f(n, \mathbb{P}^T_n) - \Hat{\alpha} \circ h(n, \mathbb{P}^T_n) \right|^2 \le c^2 \sum_{j = 1}^{n^*} \left( f_j(n, \mathbb{P}^T_n) - h_j(n, \mathbb{P}^T_n) \right)^2,
		\end{equation*}
		which is stronger than the stated condition (iii).
		Condition (iv) holds for all $j$ because:
		\begin{itemize}
			\item For $j \le n^*$: $\mathcal{G}$ is a $P_{a_j, 0} = P_0 \circ a_j^{-1}$-Donsker class by condition (5.), which implies that $\mathcal{G} \circ a_j$ is a $P_0$-Donsker class.
			\item For $n^* < j \le 2 n^*$: $\{z_{j - n^*}\}$ contains a single function with bounded image and is thus a $P_0$-Donsker class.
			\item For $2 n^* < j \le 3 n^*$: $\{x_{j - 2 n^*}\}$ contains a single function with bounded image and is thus a $P_0$-Donsker class.
			\item For $j = 3 n^* + 1$: $\{(n, \mathbb{P}^T_n) \mapsto n\}$ contains a single function with bounded image (bounded by $n^*$) and is thus a $P_0$-Donsker class.
		\end{itemize}
		All conditions (i--iv) are satisfied, so \textcite[Theorem 9.31]{kosorok2008introduction} implies that $\Hat{\alpha} \circ (\mathcal{F}_1 \times \dots \times \mathcal{F}_{3 n^* + 1})$ is $P_0$-Donsker. Since $\{ \Hat{\alpha}^g: g \in \mathcal{G} \}$ is a subclass of this class, it is also $P_0$-Donsker.
	\end{proof}

	\propositionsufficiencyconditionstheorem*

	\begin{proof}
	We first derive some intermediate results. Next, we verify the conditions of Theorem \ref{theorem:bootstrap-for-data-adaptive-estimands}.

	\paragraph{Intermediate Results}

	From the form of the estimating function, as given in Appendix \ref{appendix:plug-in-estimator}, we see that $\Phi^g(\boldsymbol{\theta})$ only depends on $g$ through $P_0 \Hat{\alpha}^g$, $P_0 (\Hat{\alpha}^g)^2$, $P_0 \Hat{\alpha}^g \Hat{\beta}$, and $P_0 \frac{1}{n} \Hat{\Sigma}^g$. Furthermore, each element of $\Phi^g(\boldsymbol{\theta})$ is a linear function of these arguments where the coefficients depend on $\boldsymbol{\theta}$:
	\begin{equation*}
		\Phi^g(\boldsymbol{\theta}) = \begin{pmatrix}
			P_0 \Hat{\alpha}^g - \mu_{\alpha} \\
			P_0 \Hat{\beta} - \mu_{\beta} \\
			P_0 (\Hat{\alpha}^g)^2 - 2 \mu_{\alpha} \cdot P_0 \Hat{\alpha}^g - P_0 \frac{1}{n} \Hat{\sigma}_{\Hat{\alpha}^g} - d_{\alpha} + \mu_{\alpha}^2 \\
			P_0 \Hat{\beta}^2 - 2 \mu_{\beta} \cdot P_0 \Hat{\beta} - P_0 \frac{1}{n} \Hat{\sigma}_{\Hat{\beta}} - d_{\beta} + \mu_{\beta}^2 \\
			P_0 \Hat{\alpha}^g \Hat{\beta} - \mu_{\beta} \cdot P_0 \Hat{\alpha}^g  - \mu_{\alpha} \cdot P_0 \Hat{\beta} + P_0 \frac{1}{n} \Hat{\sigma}_{\Hat{\alpha}^g, \Hat{\beta}} + \mu_{\alpha} \mu_{\beta} - d_{\alpha, \beta}
		\end{pmatrix}
	\end{equation*}
	Hence, $\Phi^g(\boldsymbol{\theta})$ is a Lipschitz continuous function of $\left( P_0 \Hat{\alpha}^g, P_0 (\Hat{\alpha}^g)^2, P_0 \Hat{\alpha}^g \Hat{\beta}, P_0 \frac{1}{n} \Hat{\Sigma}^g \right)$ indexed by $\boldsymbol{\theta}$ where the Lipschitz constants are upper bounded by $C < \infty$ if $\boldsymbol{\theta}$ is restricted to a bounded set $\Theta$ (as required in condition (b)).
	Hence, for any $\boldsymbol{\theta} \in \Theta$ and any $g_1, g_2 \in \mathcal{G}$, we have that (where $\lVert \cdot \rVert_F$ is the Frobenius norm):
	\begin{align*}
		\lVert \Phi^{g_1}(\boldsymbol{\theta}) - \Phi^{g_2}(\boldsymbol{\theta}) \rVert_2 & \le C \cdot \left( | P_0 \Hat{\alpha}^{g_1} - P_0 \Hat{\alpha}^{g_2}| + | P_0 (\Hat{\alpha}^{g_1})^2 - P_0 (\Hat{\alpha}^{g_2})^2| + | P_0 \Hat{\alpha}^{g_1} \Hat{\beta} - P_0 \Hat{\alpha}^{g_2} \Hat{\beta}| + \left\lVert P_0 \frac{1}{n} \Hat{\Sigma}^{g_1} - P_0 \frac{1}{n} \Hat{\Sigma}^{g_2} \right\rVert_F \right) \\
		& \le C \cdot \left( M_1 \cdot \lVert g_1 - g_2 \rVert_{\infty} + M_2 \cdot \lVert g_1 - g_2 \rVert_{\infty} + M_3 \cdot \lVert g_1 - g_2 \rVert_{\infty} + M_4 \cdot \lVert g_1 - g_2 \rVert_{\infty} \right) \\
		&~\qquad \mbox{(Lipschitz continuity in conditions (e--f))} \\
		& \le C^* \cdot  \lVert g_1 - g_2 \rVert_{\infty} \\
		&~\qquad \mbox{($C^* := C \cdot (M_1 + M_2 + M_3 + M_4)$)}
	\end{align*}

	Similar to above, note that the estimating function $\phi^g_{\boldsymbol{\theta}}$ only depends on $g$ through $\Hat{\alpha}^g$ and $\frac{1}{n}\Hat{\Sigma}^g$, and depends on the observed data through $\Hat{\beta}$. Equivalently, we may write
	\begin{equation*}
		\phi^g_{\boldsymbol{\theta}} = \phi \circ \left( \Hat{\alpha}^g, \Hat{\beta}, n^{-1} \Hat{\Sigma}^g, \boldsymbol{\theta} \right)
	\end{equation*}
	for a measurable map $\phi$ that is a second-degree polynomial function of its arguments (see Appendix \ref{appendix:plug-in-estimator}). By conditions (b) and (c), these arguments are uniformly bounded (uniformly over $g \in \mathcal{G}$ and $\boldsymbol{\theta} \in \Theta$), so $\phi$ is Lipschitz on the relevant domain.
	It now follows by applying \textcite[Theorem 9.31]{kosorok2008introduction} componentwise (after stacking the entries of $n^{-1} \Hat{\Sigma}^g$ and $\boldsymbol{\theta}$) that $\left\{ \phi_{\boldsymbol{\theta}}^g: g \in \mathcal{G}, \boldsymbol{\theta} \in \Theta \right\}$ is $P_0$-Donsker if 
	\begin{enumerate}
		\item $\mathcal{F}_1 = \left\{ \Hat{\alpha}^g: g \in \mathcal{G} \right\}$, $\mathcal{F}_2 = \left\{ \Hat{\beta} \right\}$, $\mathcal{F}_3 = \left\{ \frac{1}{n} \Hat{\Sigma}^g: g \in \mathcal{G} \right\}$, and $\mathcal{F}_4 = \left\{ \boldsymbol{\theta}: \boldsymbol{\theta} \in \Theta \right\}$ are $P_0$-Donsker classes,
		\item $\lVert P_0 \rVert_{\mathcal{F}_j} < \infty$ for $j = 1, \ldots, 4$, and
		\item $\phi^g_{\boldsymbol{\theta}}$ is square integrable for some $(\boldsymbol{\theta}, g) \in \Theta \times \mathcal{G}$.
	\end{enumerate}
	Condition (1.) follows from condition (c) for $j = 1$; for $j = 2$, $\mathcal{F}_2$ is a singleton with $\Hat{\beta}$ bounded and thus $P_0$-Donsker; for $j = 3$, note that $\{\Hat{\Sigma}^g: g \in \mathcal{G}\}$ is bounded $P_0$-Donsker by condition (c), and multiplying by the fixed bounded function $(n, \mathbb{P}^T_n) \mapsto n^{-1}$ preserves the $P_0$-Donsker property; and for $j = 4$, $\mathcal{F}_4$ is a bounded class of constant functions.
	Condition (2.) follows from the involved classes of functions being uniformly bounded by conditions (b) and (c).
	Condition (3.) follows from $\phi$ being a polynomial and all involved functions being uniformly bounded by conditions (b) and (c).
	
	\paragraph{Verification of the Conditions}

	Condition (i) of Theorem \ref{theorem:bootstrap-for-data-adaptive-estimands} follows from $\boldsymbol{\theta}_0^g$ being the unique root of $\Phi^g(\boldsymbol{\theta})$ on any compact $K \supset \Theta$ for $g \in \mathcal{G}$ and $\boldsymbol{\theta} \mapsto \Phi^g(\boldsymbol{\theta})$ being a continuous function.

	Condition (ii) of Theorem \ref{theorem:bootstrap-for-data-adaptive-estimands} follows from $\left\{ \phi_{\boldsymbol{\theta}}^g: g \in \mathcal{G}, \boldsymbol{\theta} \in \Theta \right\}$ being a $P_0$-Donsker class (intermediate result) because a $P_0$-Donsker class is a $P_0$-Glivenko-Cantelli class.

	The first part of Condition (iii) of Theorem \ref{theorem:bootstrap-for-data-adaptive-estimands} follows from the intermediate Donsker result and the fact that $\Theta$ is open. Indeed, for $\boldsymbol{\theta}^{g^*}_0 \in \Theta$, there exists a $\delta > 0$ such that $\left\{ \boldsymbol{\theta} \in \mathbb{R}^5: \lVert \boldsymbol{\theta} - \boldsymbol{\theta}^{g^*}_0 \rVert_2 < \delta \right\} \subset \Theta$ (and a subset of a Donsker class is a Donsker class). 
	For the second part, we write the estimating function $\phi^g_{\boldsymbol{\theta}}: \mathbb{N} \times \mathcal{M}^*_{\text{NP}} \times \Theta \to \mathbb{R}^5$ as the following composite function:
	\begin{equation*}
		 \mathbb{N} \times \mathcal{M}^*_{\text{NP}} \times \Theta \mapsto \phi^g_{\boldsymbol{\theta}} = \phi \circ \left( \Hat{\alpha}^g, \Hat{\beta}, n^{-1} \Hat{\Sigma}^g, \boldsymbol{\theta} \right)
	\end{equation*}
	where $\Hat{\alpha}^g$, $\Hat{\beta}: \mathbb{N} \times \mathcal{M}^*_{\text{NP}} \to \mathbb{R}$ and $n^{-1} \Hat{\Sigma}^g: \mathbb{N} \times \mathcal{M}^*_{\text{NP}} \to \mathbb{R}^{2 \times 2}$. The function $\phi$ above is a second-degree polynomial function and, hence, continuous. Hence, if for every $(n, \mathbb{P}^T_n) \in \mathbb{N} \times \mathcal{M}^*_{\text{NP}}$, the arguments converge to the appropriate limits indexed by $g^*$ and $\boldsymbol{\theta}^{g^*}_0$ as $\lVert \boldsymbol{\theta} - \boldsymbol{\theta}^{g^*}_0\rVert_2 + \lVert g - g^* \rVert_{\infty} \to 0$, then $\phi^g_{\boldsymbol{\theta}}(n, \mathbb{P}^T_n) \to \phi^{g^*}_{\boldsymbol{\theta}^{g^*}_0}(n, \mathbb{P}^T_n)$ pointwise.
	By condition (d), $\Hat{\alpha}^{g}$ and $n^{-1} \Hat{\Sigma}^{g}$ converge pointwise to $\Hat{\alpha}^{g^*}$ and $n^{-1} \Hat{\Sigma}^{g^*}$, respectively, if $\lVert g - g^* \rVert_{\infty} \to 0$. The other arguments, $\Hat{\beta}$ and $\boldsymbol{\theta}$ trivially converge to $\Hat{\beta}$ and $\boldsymbol{\theta}^{g^*}_0$ as $\lVert \boldsymbol{\theta} - \boldsymbol{\theta}^{g^*}_0\rVert_2 \to 0$.
	The pointwise convergence implies, by the dominated convergence theorem, that $P_0 \lVert \phi_{\boldsymbol{\theta}}^{g} - \phi_{\boldsymbol{\theta}^{g^*}_0}^{g^*} \rVert_2^2 \to 0$ as $\lVert \boldsymbol{\theta} - \boldsymbol{\theta}^{g^*}_0\rVert_2 + \lVert g - g^* \rVert_{\infty} \to 0$.

	The first part of condition (iv) of Theorem \ref{theorem:bootstrap-for-data-adaptive-estimands} follows from $\Hat{\alpha}^g$, $\Hat{\beta}$, $\Hat{\Sigma}^g$, and $\boldsymbol{\theta}$ being bounded (uniformly in $\mathcal{G}$) by conditions (b) and (c). The second part follows from the expression for $\dot{\Phi}^g(\boldsymbol{\theta})$ given in (\ref{eq:phi-dot}).

	Condition (v) of Theorem \ref{theorem:bootstrap-for-data-adaptive-estimands} follows from the definition of the estimators as the exact solutions of the corresponding estimating equations.

	Condition (vi) of Theorem \ref{theorem:bootstrap-for-data-adaptive-estimands} follows from the intermediate results.
	\end{proof}
		
	\subsubsection{Appendix \ref{appendix:alternative-asymptotic-framework}}
	
	\lemmanewasymptoticsmomentbased*

	\begin{proof}
		We prove the lemma first for the mean parameters and then for the covariance parameters. For simplicity of notation, we let $\mathbb{P}_N$ be the empirical distribution of the data, including the unobserved bias terms $b_n^g$ and $c_n^g$ where $n = N^{\gamma}$. 

		\item \paragraph{Mean Parameters}
		
		The estimators for $E_{P^F_0}(\alpha^g)$ are given by
		\begin{equation*}
			\Hat{\Psi}^g(\Tilde{\mathbb{P}}_{N, n}) = \Tilde{\mathbb{P}}_{N} \Tilde{\alpha}^g_n \quad \text{and} \quad \Hat{\Psi}^g(\Hat{\mathbb{P}}_{N, n}) = \Hat{\mathbb{P}}_{N} \Hat{\alpha}^g_n
		\end{equation*}
		for the observed trial-level estimators and the debiased estimators, respectively. We now show that they are asymptotically equivalent under conditions.
		\begin{align*}
			N^{1/2} \left( \Hat{\Psi}^{g}(\Tilde{\mathbb{P}}_{N, n}) - \Hat{\Psi}^g(\Hat{\mathbb{P}}_{N, n}) \right) & = N^{1/2} \mathbb{P}_N \left( \Tilde{\alpha}^g_n - \Hat{\alpha}^g_n \right) \\
			& \le N^{1/2} \mathbb{P}_N n^{-1/2} \lVert b^g_n \rVert_2 &&~\mbox{(Assumption \ref{assumption:biased-trt-effect})} \\
			& = N^{1/2} n^{-1/2} \cdot (\mathbb{P}_N - P_0^F) \lVert b_{n}^g \rVert_2 \\
			& \quad + N^{1/2} n^{-1/2} P_0^F \lVert b_{n}^g \rVert_2 &&~\mbox{(add and subtract $N^{1/2} n^{-1/2} P_0^F \lVert b_{n}^g \rVert_2$)}
		\end{align*}
		By Assumption \ref{assumption:biased-trt-effect}, $\lVert b_{n}^g \rVert_2$ has finite variance for $n \ge n_0 \in \mathbb{R}$; hence, the first term above is the root-$N$ standardized empirical mean of mean-zero i.i.d.~random variables with variance going to zero as $n \to \infty$. The first term is thus $o_{P^F_0}(1)$ whenever $\gamma > 0$. 
		The second term is $o_{P^F_0}(1)$ if
		\begin{equation*}
			N^{1/2} n^{-1/2} P_0^F \lVert b_{n}^g \rVert_2 \le N^{1/2} n^{-1/2} (P_0^F \lVert b_{n}^g \rVert_2^2)^{1/2} =  N^{\frac{1}{2} - \frac{1}{2}\gamma} \cdot O(N^{- \frac{1}{2} \gamma \delta_1}) = N^{\frac{1}{2} - \frac{1}{2}\gamma - \frac{1}{2}\gamma \delta_1} \cdot O(1) = o(1).
		\end{equation*}
		This holds for $\gamma > \frac{1}{1 + \delta_1}$.
		The same holds analogously for the estimator of $E_{P^F_0}(\beta)$.

		\item \paragraph{Covariance Parameters}

		The estimators for $E_{P^F_0}(\alpha^g - E_{P^F_0}(\alpha^g))^2$ are given by
		\begin{equation*}
			\Hat{\Psi}^g(\Tilde{\mathbb{P}}_{N, n}) = \mathbb{P}_{N} \left( \Tilde{\alpha}^g_n - \mathbb{P}_{N} \Tilde{\alpha}^g_n \right)^2 - \mathbb{P}_{N} \left( n^{-1} \Tilde{\sigma}^g_{\alpha} \right) \quad \text{and} \quad \Hat{\Psi}^g(\Hat{\mathbb{P}}_{N, n}) = \mathbb{P}_{N} \left( \Hat{\alpha}^g_n - \mathbb{P}_{N} \Hat{\alpha}^g_n \right)^2 - \mathbb{P}_{N} \left( n^{-1}\Hat{\sigma}^g_{\alpha, n} \right),
		\end{equation*}
		for the observed trial-level estimators and the debiased estimators, respectively. We now show that they are asymptotically equivalent under conditions.
		\begin{align*}
			N^{1/2} \left( \Hat{\Psi}^g(\Tilde{\mathbb{P}}_{N, n}) - \Hat{\Psi}^g(\Hat{\mathbb{P}}_{N, n}) \right) & = N^{1/2} n^{-1} \mathbb{P}_N \left( \Tilde{\sigma}^g_{\alpha} - \Hat{\sigma}^g_{\alpha, n} \right)  + N^{1/2} \mathbb{P}_N \left\{ \left( \Tilde{\alpha}^g_n - \mathbb{P}_N \Tilde{\alpha}^g_n \right)^2 - \left( \Hat{\alpha}^g_n - \mathbb{P}_N \Hat{\alpha}^g_n \right)^2 \right\}\\
			& = N^{1/2} n^{-1} \mathbb{P}_N \left( \Tilde{\sigma}^g_{\alpha} - \Hat{\sigma}^g_{\alpha, n} \right)  \\
			& \quad + N^{1/2} \mathbb{P}_N \left\{ \left( \Tilde{\alpha}^g_n - \mathbb{P}_N \Tilde{\alpha}^g_n \right)^2 - \left( \Tilde{\alpha}^g_n - n^{-1/2} b^g_n - \mathbb{P}_N \left( \Tilde{\alpha}^g_n - n^{-1/2} b^g_n \right) \right)^2 \right\} \\
			& = N^{1/2} n^{-1} \mathbb{P}_N \left( \Tilde{\sigma}^g_{\alpha} - \Hat{\sigma}^g_{\alpha, n} \right)  \\
			& \quad + N^{1/2} \mathbb{P}_N \left\{ \left( \Tilde{\alpha}^g_n - \mathbb{P}_N \Tilde{\alpha}^g_n \right)^2 - \left( (\Tilde{\alpha}^g_n - \mathbb{P}_N \Tilde{\alpha}^g_n) - n^{-1/2}( b^g_n - \mathbb{P}_N  b^g_n) \right)^2 \right\} \\
			& = N^{1/2} n^{-1} \mathbb{P}_N \left( \Tilde{\sigma}^g_{\alpha} - \Hat{\sigma}^g_{\alpha, n} \right)  \\
			& \quad + N^{1/2} \mathbb{P}_N \left\{ 2 n^{-1/2} ( b^g_n - \mathbb{P}_N  b^g_n) (\Tilde{\alpha}^g_n - \mathbb{P}_N \Tilde{\alpha}^g_n) - n^{-1} \left( b^g_n - \mathbb{P}_N b^g_n \right)^2 \right\} \\
			& \le N^{1/2} n^{-1} \mathbb{P}_N \lVert c^g_n \rVert_F \\
			& \quad + 2 N^{1/2} n^{-1/2} \mathbb{P}_N \left( b^g_n - \mathbb{P}_N  b^g_n \right) \left( \Tilde{\alpha}^g_n - \mathbb{P}_N \Tilde{\alpha}^g_n \right) \\
			& \quad + N^{1/2} n^{-1} \mathbb{P}_N \left( b^g_n - \mathbb{P}_N b^g_n \right)^2 \\
			& =: T_1 + T_2 + T_3
		\end{align*}
		We now bound each of the three terms above. 
		First, we have 
		\begin{equation*}
			T_1 = N^{1/2} n^{-1} (\mathbb{P}_N - P_0^F) \lVert c_n^g \rVert_F + N^{1/2} n^{-1} P_0^F \lVert c_n^g \rVert_F,
		\end{equation*}
		where $N^{1/2} n^{-1} (\mathbb{P}_N - P_0^F) \lVert c_n^g \rVert_F$ is the root-$N$ standardized empirical mean of mean-zero i.i.d.~random variables with variance going to zero as $n \to \infty$ because $\lVert c_n^g \rVert_F$ has finite variance by assumption \ref{assumption:biased-sampling-variance}; hence, $N^{1/2} n^{-1} (\mathbb{P}_N - P_0^F) \lVert c_n^g \rVert_F = o_{P^F_0}(1)$ whenever $\gamma > 0$.
		Furthermore, $N^{1/2} n^{-1} P_0^F \lVert c_n^g \rVert_F$ is $o_{P^F_0}(1)$ if 
		\begin{equation*}
			N^{1/2} n^{-1} P_0^F \lVert c_n^g \rVert_F = N^{1/2} N^{-\gamma} \cdot O(n^{-\delta_2}) = N^{1/2 - \gamma - \gamma \delta_2} \cdot O(1) = o_P(1).
		\end{equation*}
		This holds for $\gamma > \frac{1}{2} \cdot \frac{1}{1 + \delta_2}$.

		Second, we have by the Cauchy-Schwarz inequality:
		\begin{align*}
			T_2 & \le N^{1/2} n^{-1/2} \left\{ \mathbb{P}_N \left( b^g_n - \mathbb{P}_N  b^g_n \right)^2 \right\}^{1/2} \cdot \left\{ \mathbb{P}_N \left( \Tilde{\alpha}^g_n - \mathbb{P}_N \Tilde{\alpha}^g_n \right)^2 \right\}^{1/2}.
		\end{align*}
		The elements on the right-hand side correspond to sample variances. The sample variance corresponding to $\Tilde{\alpha}^g$ is $O_{P^F_0}(1)$ and the sample variance corresponding to $b^g_n$ is bounded by above by 
		\begin{align*}
			\mathbb{P}_N \left( b^g_n - \mathbb{P}_N b^g_n \right)^2 & = \mathbb{P}_N \left( b^g_n - P_0^F b^g_n + P_0^F b^g_n - \mathbb{P}_N b^g_n \right)^2 \\
			& \le 2 \mathbb{P}_N \left( b^g_n - P_0^F b^g_n \right)^2 + 2\left( P_0^F b^g_n - \mathbb{P}_N b^g_n \right)^2  \\
			& = 2 (\mathbb{P}_N - P^F_0) \left( b^g_n - P_0^F b^g_n \right)^2 + 2 P^F_0 \left( b^g_n - P_0^F b^g_n \right)^2  + 2 \left( (\mathbb{P}_N - P_0^F) b^g_n \right)^2 \\
			& \le 2 (\mathbb{P}_N - P^F_0) \left( b^g_n - P_0^F b^g_n \right)^2 + 2 P^F_0 \lVert b^g_n \rVert_2^2 + 2 \left( (\mathbb{P}_N - P_0^F) b^g_n \right)^2 \\ 
			& = O_{P^F_0}(N^{-1/2}) + O(n^{-\delta_1}) 
		\end{align*}
		In the last equality, we used the fact that $P^F_0 \lVert b^g_n \rVert_2^4 < K < \infty$ for the mean-zero sample averages and $P_0^F \lVert b^g_n \rVert_2^2 = O(n^{-\delta_1})$.
		Hence, $T_2 = o_{P^F_0}(1)$ if $N^{1/2} n^{-1/2} \cdot O(n^{-\delta_1}) = o_{P^F_0}(1)$, which holds for $\gamma > \frac{1}{2} \cdot \frac{1}{\frac{1}{2} + \delta_1}$.
		
		Third, $T_3 = o_{P^F_0}(1)$ by the above derivations. 
		
		The derivations for the other covariance parameters are analogous.
	\end{proof}

	\corollarynewasymptoticsmomentbased*

	\begin{proof}
		The derivations in the proof of Lemma \ref{lemma:new-asymptotics-plug-in} can be repeated for an estimated $g_N$ because
		\begin{equation*}
			\frac{1}{2} \lVert b^{g_N}_n \rVert_2^2 \le \lVert b^{g_N}_n - b^{g^*}_n \rVert_2^2 + \lVert b^{g^*}_n \rVert_2^2 \quad \text{and} \quad  \lVert c^{g_N}_n \rVert_F \le \lVert c^{g_N}_n - c^{g^*}_n \rVert_F + \lVert c^{g^*}_n \rVert_F.
		\end{equation*}
		By assumption, $P_0^F \lVert b^{g_N}_n - b^{g^*}_n \rVert_2^2$ and $P_0^F \lVert c^{g_N}_n - c^{g^*}_n \rVert_F$ are of the same order as $P_0^F \lVert b^{g^*}_n \rVert_2^2 = O(n^{-\delta_1}) = O(N^{-\gamma \delta_1})$ and $P_0^F \lVert c^{g^*}_n \rVert_F = O(n^{-\delta_2}) = O(N^{-\gamma \delta_2})$, respectively. 
		The same argument can be repeated for $\lVert b^{g_N}_n \rVert_2^4$ and $\lVert c^{g_N}_n \rVert_F^2$. 
	\end{proof}

	\subsection{Technical Lemmas}

	\begin{restatable}{lemma}{lemmagaussiansmoothingcontinuity}
		\label{lemma:gaussian-smoothing-continuity}
		Fix a positive definite matrix $\Sigma \in \mathbb{R}^{k \times k}$. For each $x \in \mathbb{R}^k$, let $U_x \sim \mathcal{N}(x, \Sigma)$ and let $f: \mathbb{R}^k \to \mathbb{R}$ be measurable.
		If $E\{ f(U_{x_0})^2 \} < \infty$ for some $x_0 \in \mathbb{R}^k$, then the map $m(x) := E\{ f(U_x) \}$ is well defined and continuous on $\mathbb{R}^k$.
	\end{restatable}

	\begin{proof}
		Let $U_{x} \sim \mathcal{N}(x, \Sigma)$, write $U := U_{x_0}$, and let $\varphi_{\Sigma}$ denote the density of $\mathcal{N}(0, \Sigma)$ with respect to Lebesgue measure.
		For each $x \in \mathbb{R}^k$, define the likelihood ratio
		\begin{equation*}
			r_x(U) := \frac{\varphi_{\Sigma}(U - x)}{\varphi_{\Sigma}(U - x_0)} = \exp\left\{ (x - x_0)^\top \Sigma^{-1} (U - x_0) - \frac{1}{2} (x - x_0)^\top \Sigma^{-1} (x - x_0) \right\}.
		\end{equation*}
		Then $m(x) := E \{ f(U_{x}) \} = E\{ f(U) r_x(U) \}$, and from the properties of the log-normal distribution it follows that $E\{ r_x(U)^2 \} = \exp\left\{ (x - x_0)^\top \Sigma^{-1} (x - x_0) \right\} < \infty$ for all $x$.
		Since $E\{ f(U)^2 \} < \infty$ by assumption, it follows by Cauchy--Schwarz that $|m(x)| < \infty$ for all $x$.
		
		Moreover, for any $x_1, x_2 \in \mathbb{R}^k$, it follows by the properties of the log-normal distribution that
		\begin{equation*}
			E\{ r_{x_1}(U) r_{x_2}(U) \} = \exp\left\{ (x_1 - x_0)^\top \Sigma^{-1} (x_2 - x_0) \right\},
		\end{equation*}
		and thus that
		\begin{align*}
			E\{ (r_{x_1}(U) - r_{x_2}(U))^2 \} & = E\{ r_{x_1}(U)^2 \} + E\{ r_{x_2}(U)^2 \} - 2 E\{ r_{x_1}(U) r_{x_2}(U) \} \\
			& = \exp\left\{ (x_1 - x_0)^\top \Sigma^{-1} (x_1 - x_0) \right\} \\
			& \quad + \exp\left\{ (x_2 - x_0)^\top \Sigma^{-1} (x_2 - x_0) \right\} \\
			& \quad - 2 \exp\left\{ (x_1 - x_0)^\top \Sigma^{-1} (x_2 - x_0) \right\},
		\end{align*}
		which converges to $0$ as $x_1 \to x_2$.

		Hence, it follows that for any $x_1, x_2 \in \mathbb{R}^k$,
		\begin{equation*}
			|m(x_1) - m(x_2)| = \left| E \left\{ f(U_{x_0}) \left( r_{x_1}(U_{x_0}) - r_{x_2}(U_{x_0}) \right) \right\} \right| \le \left\{ E f(U_{x_0})^2 \right\}^{1/2} \left\{ E (r_{x_1}(U) - r_{x_2}(U))^2 \right\}^{1/2},
		\end{equation*}
		which proves continuity of $m$ because the right-hand side converges to $0$ as $x_1 \to x_2$.
	\end{proof}

	\begin{restatable}{lemma}{lemmalocalnonparametricq}
		\label{lemma:local-non-parametric-q}
		Let $\mathcal{Q}$ be the collection of all distributions $Q$ of $\boldsymbol{Z} = (X, \Sigma)$ such that the support of $X \mid \Sigma$ contains a nonempty open set in $\mathbb{R}^k$ $Q$-a.s. and $\Sigma$ is positive definite $Q$-a.s.
		Then $\mathcal{Q}$ is locally non-parametric at any $Q_0 \in \mathcal{Q}$, i.e., $\overline{\mathcal{T}}(Q_0) = L^2_0(Q_0)$.
	\end{restatable}

	\begin{proof}
		Let $h \in L^2_0(Q_0) \cap L^{\infty}(Q_0)$.
		For $|t| < \lVert h \rVert_{\infty}^{-1}$, define a one-dimensional parametric submodel $(Q_t)_{t}$ through $Q_0$ by the linear tilt
		\begin{equation*}
			\frac{dQ_t}{dQ_0} := 1 + t h.
		\end{equation*}
		Since $1 + t h$ is bounded away from $0$ and $\infty$, $Q_t$ is equivalent to $Q_0$, meaning that $Q_t$ and $Q_0$ are mutually absolutely continuous (i.e., $Q_t \ll Q_0$ and $Q_0 \ll Q_t$, so they have the same null sets).
		In particular, $\Sigma$ is positive definite $Q_t$-a.s.
		Moreover, for $Q_0$-a.e.~$\Sigma$ (equivalently, for $Q_t$-a.e.~$\Sigma$), the conditional distributions satisfy
		\begin{equation*}
			\frac{dQ_t(\cdot \mid \Sigma)}{dQ_0(\cdot \mid \Sigma)}(x) = \frac{1 + t h(x, \Sigma)}{E_{Q_0}\left\{ 1 + t h(X, \Sigma) \mid \Sigma \right\}},
		\end{equation*}
		which is bounded away from $0$ and $\infty$.
		Hence, $Q_t(\cdot \mid \Sigma)$ and $Q_0(\cdot \mid \Sigma)$ have the same support for $Q_0$-a.e.~$\Sigma$ (equivalently, for $Q_t$-a.e.~$\Sigma$).
		Therefore, the support of $X \mid \Sigma$ contains a nonempty open set in $\mathbb{R}^k$ for $Q_t$-a.e.~$\Sigma$, and thus $Q_t \in \mathcal{Q}$ for all sufficiently small $t$.

		The score of this parametric submodel at $t = 0$ equals $h$.
		Thus, $L^2_0(Q_0) \cap L^{\infty}(Q_0) \subset \mathcal{T}(Q_0)$.
		Since $L^2_0(Q_0) \cap L^{\infty}(Q_0)$ is dense in $L^2_0(Q_0)$, it follows that $\overline{\mathcal{T}}(Q_0) = L^2_0(Q_0)$.
	\end{proof}

	\begin{restatable}{lemma}{lemmaconditionalconvergenceinp}
		\label{lemma:conditional-convergence-in-p}
		Let $(\mathbb{E}, \lVert \cdot \rVert)$ be a normed vector space and let $X_N$ be a sequence of random elements in $\mathbb{E}$ defined on the product space of the observed data and bootstrap weights. Let $o_P(1)$ denote convergence in $P$-probability under the joint law of the observed data and bootstrap weights. Then, $\lVert X_N \rVert = o_P(1)$ if and only if $X_N$ converges to zero in conditional probability given the observed data, i.e., for all $\epsilon > 0$ and all $\delta > 0$,
		\begin{equation*}
			P\left( P\left( \lVert X_N \rVert > \epsilon \mid \mathcal{O}_N \right) > \delta \right) \to 0,
		\end{equation*}
		where $\mathcal{O}_N$ is the $\sigma$-field generated by the observed data.
	\end{restatable}

	\begin{proof}
		Fix $\epsilon > 0$ and define the conditional tail probability $Q_{N, \epsilon} := P\left( \lVert X_N \rVert > \epsilon \mid \mathcal{O}_N \right)$.
		Then $0 \le Q_{N, \epsilon} \le 1$ and, by the law of total probability, $P\left( \lVert X_N \rVert > \epsilon \right) = E\left\{ Q_{N, \epsilon} \right\}$.

		$\Rightarrow$.~~~If $\lVert X_N \rVert = o_P(1)$, then $P(\lVert X_N \rVert > \epsilon) \to 0$ for all $\epsilon > 0$, which implies $E\{Q_{N,\epsilon}\} \to 0$. For any $\delta > 0$, Markov's inequality yields
		\begin{equation*}
			P\left( Q_{N, \epsilon} > \delta \right) \le \frac{E\left\{ Q_{N, \epsilon} \right\}}{\delta} = \frac{P\left( \lVert X_N \rVert > \epsilon \right)}{\delta} \to 0.
		\end{equation*}
		This shows that $Q_{N,\epsilon} \to 0$ in $P$-probability, i.e., $X_N$ converges to zero in conditional probability.

		$\Leftarrow$.~~~For all $\epsilon, \delta > 0$, we have $P(Q_{N,\epsilon} > \delta) \to 0$. 
		Fix $\epsilon > 0$. Since $0 \le Q_{N,\epsilon} \le 1$, for any $\delta > 0$,
		\begin{equation*}
			P(\lVert X_N \rVert > \epsilon) = E\left\{ Q_{N, \epsilon} \right\} = E\left[ Q_{N, \epsilon} \cdot \left\{ \mathbf{1}(Q_{N, \epsilon} \le \delta) +  \mathbf{1}(Q_{N, \epsilon} > \delta) \right\} \right]  \le \delta + P\left( Q_{N, \epsilon} > \delta \right).
		\end{equation*}
		Letting $N \to \infty$ gives $\limsup_{N \to \infty} P(\lVert X_N \rVert > \epsilon) \le \delta$; since this holds for every $\delta > 0$, it follows that $P(\lVert X_N \rVert > \epsilon) \to 0$.
		Since this holds for every $\epsilon > 0$, we conclude that $\lVert X_N \rVert = o_P(1)$.
	\end{proof}

	\begin{restatable}{lemma}{lemmaslutskyconditional}
		\label{lemma:slutsky-conditional}
		Let $(\mathbb{E}, \lVert \cdot \rVert)$ be a normed vector space. Let $X_N$ and $Y_N$ be sequences of random elements in $\mathbb{E}$ defined on the product space of the observed data and bootstrap weights. If $X_N \overset{P}{\underset{W}{\leadsto}} \mathbb{G}$ and $\lVert Y_N \rVert = o_P(1)$, then $X_N + Y_N \overset{P}{\underset{W}{\leadsto}} \mathbb{G}$.
	\end{restatable}

	\begin{proof}
		Let $e(x, y) := \lVert x - y \rVert$ be the metric induced by the norm on $\mathbb{E}$. Fix any $h \in \operatorname{BL}_1$. By adding and subtracting $E_W h(X_N)$ and using the triangle inequality,
		\begin{align*}
			\left| E_W h(X_N + Y_N) - E h(\mathbb{G}) \right| & \le \left| E_W h(X_N) - E h(\mathbb{G}) \right| + E_W \left| h(X_N + Y_N) - h(X_N) \right|.
		\end{align*}
		Since $h$ is $1$-Lipschitz with respect to $e$ and satisfies $\lVert h \rVert_{\infty} \le 1$, we have
		\begin{equation*}
			\left| h(X_N + Y_N) - h(X_N) \right| \le e(X_N + Y_N, X_N) \wedge 2 = \lVert Y_N \rVert \wedge 2.
		\end{equation*}
		Therefore,
		\begin{equation*}
			\sup_{h \in \operatorname{BL}_1} \left| E_W h(X_N + Y_N) - E h(\mathbb{G}) \right| \le \sup_{h \in \operatorname{BL}_1} \left| E_W h(X_N) - E h(\mathbb{G}) \right| + E_W\left( \lVert Y_N \rVert \wedge 2 \right).
		\end{equation*}
		The first term on the right-hand side is $o_P(1)$ by the assumption that $X_N \overset{P}{\underset{W}{\leadsto}} \mathbb{G}$.

		For the second term, we have for any $\epsilon > 0$,
		\begin{equation*}
			E\left\{ \lVert Y_N \rVert \wedge 2 \right\} \le E\left[ (\lVert Y_N \rVert \wedge 2) \mathbf{1}(\lVert Y_N \rVert \le \epsilon) \right] + E\left[ (\lVert Y_N \rVert \wedge 2) \mathbf{1}(\lVert Y_N \rVert > \epsilon) \right] \le \epsilon + 2 P(\lVert Y_N \rVert > \epsilon),
		\end{equation*}
		so $E\{ \lVert Y_N \rVert \wedge 2 \} \to 0$.
		By the law of iterated expectations and Markov's inequality, for any $\eta > 0$,
		\begin{align*}
			P\left( E_W\left\{ \lVert Y_N \rVert \wedge 2 \right\} > \eta \right) \le \frac{E\left[ E_W\left\{ \lVert Y_N \rVert \wedge 2 \right\} \right]}{\eta} = \frac{E\left\{ \lVert Y_N \rVert \wedge 2 \right\}}{\eta} \to 0,
		\end{align*}
		so $E_W(\lVert Y_N \rVert \wedge 2) = o_P(1)$. Combining the previous displays yields
		\begin{equation*}
			\sup_{h \in \operatorname{BL}_1} \left| E_W h(X_N + Y_N) - E h(\mathbb{G}) \right| \overset{P}{\to} 0.
		\end{equation*}
		Finally, since $X_N$ and $Y_N$ are random elements on the product space, $h(X_N + Y_N)$ is measurable for each $h \in \operatorname{BL}_1$, so $h(X_N + Y_N)^* = h(X_N + Y_N)_*$ almost surely and the majorant/minorant condition in the definition of $\overset{P}{\underset{W}{\leadsto}}$ holds trivially. This proves $X_N + Y_N \overset{P}{\underset{W}{\leadsto}} \mathbb{G}$.
	\end{proof}

	\begin{restatable}{lemma}{lemmaconsistencygeneralpluginestimator}
		\label{lemma:consistency-general-plugin-estimator}
		Under the conditions (i) and (vi) of Theorem \ref{theorem:bootstrap-for-data-adaptive-estimands}, we have that $g \mapsto \Psi^{g}(P_0)$ is continuous at $g^*$. This in turn implies that $\Hat{\Psi}^{g_N}(P_N) \overset{P}{\to} \Psi^{g_N}(P_0)$ if $\Hat{\Psi}^{g_N}(P_N) \overset{P}{\to} \Psi^{g^*}(P_0)$ and $g_N \overset{P}{\to} g^*$.
	\end{restatable}

	\begin{proof}
		If $\Psi^{g}(P_0)$ is continuous in $g$ at $g^*$, then it automatically follows that $\Psi^{g_N}(P_0) \overset{P}{\to} \Psi^{g^*}(P_0)$ whenever $g_N \overset{P}{\to} g^*$ by the continuous mapping theorem, and hence also $\Hat{\Psi}^{g_N}(P_N) \overset{P}{\to} \Psi^{g_N}(P_0)$ if $\Hat{\Psi}^{g_N}(P_N) \overset{P}{\to} \Psi^{g^*}(P_0)$.
		We thus only need to show that this continuity holds under conditions (i) and (vi) of Theorem \ref{theorem:bootstrap-for-data-adaptive-estimands}.

		Let $\boldsymbol{\theta}_0^{g^*} := \Psi^{g^*}(P_0)$ where, by definition of the estimand and condition (i) of Theorem \ref{theorem:bootstrap-for-data-adaptive-estimands}, $P_0 \phi_{\boldsymbol{\theta}_0^{g^*}}^{g^*} = 0$ and $P_0 \phi_{\boldsymbol{\theta}}^{g^*} \neq 0$ for $\boldsymbol{\theta} \ne \boldsymbol{\theta}_0^{g^*}$. 
		Let $g_N$ be a sequence of functions in $\mathcal{G}$ such that $d(g_N, g^*) \to 0$ and let $\boldsymbol{\theta}_N := \Psi^{g_N}(P_0)$ such that $P_0 \phi_{\boldsymbol{\theta}_N}^{g_N} = 0$. 
		We now have by condition (vi) of Theorem \ref{theorem:bootstrap-for-data-adaptive-estimands} that
		\begin{align*}
			\left| P_0 \phi^{g_N}_{\boldsymbol{\theta}_N} - P_0 \phi^{g^*}_{\boldsymbol{\theta}_N} \right| & \le C \cdot d(g_N, g^*) \to 0.
		\end{align*}
		Because $P_0 \phi^{g_N}_{\boldsymbol{\theta}_N} = 0$ for all $N$, we have that $P_0 \phi^{g^*}_{\boldsymbol{\theta}_N} \to 0$. By condition (i) of Theorem \ref{theorem:bootstrap-for-data-adaptive-estimands}, this implies that $\boldsymbol{\theta}_N \to \boldsymbol{\theta}_0^{g^*}$. This shows that $\Psi^{g_N}(P_0) \to \Psi^{g^*}(P_0)$ whenever $d(g_N, g^*) \to 0$; hence, $g \mapsto \Psi^{g}(P_0)$ is continuous at $g^*$.
	\end{proof}

	\appendixclearpage
	\section{Simulations}\label{appendix:simulation-results}

	\subsection{Data-Generating Mechanism}

	We consider two data-generating mechanisms: a proof-of-concept scenario and a vaccine scenario. The proof-of-concept scenario is meant to illustrate the potential advantage of using an estimated surrogate index over the surrogate endpoint itself, but is not very realistic. The vaccine scenario is meant to reflect a set of vaccine trials and is more realistic. We will show that our methods work in principle in the proof-of-concept scenario and that they work reasonably well in practice in the realistic scenario.

	\subsubsection{Proof-of-Concept Scenario}
	
	We first describe the data-generating mechanism for a fixed trial. Next, we describe how between-trial heterogeneity is simulated and how this relates to the comparability and surrogacy assumptions.
	
	\paragraph{Within-Trial Data-Generating Mechanism}
	
	We observe in each trial one binary baseline covariate $X_1$, the treatment indicator $Z$, the univariate continuous surrogate $S$, and the univariate continuous clinical endpoint $Y$.
	The within-trial observed-data distribution is determined by the following set of equations:
	\begin{equation}\label{eq:proof-of-concept-dgm-within-trial}
		\begin{cases}
			X_1 \sim \text{Bernoulli}(p) \\
			S = \beta_{S, Z} \cdot Z + \epsilon_S \text{ where } \epsilon_S \sim \mathcal{N}(0, 1) \\
			Y = - 0.25 \cdot S \cdot X_1 + S \cdot (1 - X_1) + \beta_{Y, Z} \cdot Z + \beta_{Y, S^2} \cdot S^2 + \epsilon_Y \text{ where } \epsilon_Y \sim \mathcal{N}(0, 1)
		\end{cases}
	\end{equation}
	where the first $n /2$ subjects are assigned to treatment $Z = 1$ and the last $n /2$ subjects to treatment $Z = 0$. We let $n = 2000$.
	Further note that the surrogate is negatively associated with the clinical endpoint if $X_1 = 1$, but positively associated if $X_1 = 0$. This is unrealistic: we expect the direction of this association to be the same for all participant subgroups for a good surrogate. Nonetheless, our methods can be applied in this artificial scenario.
		
	\paragraph{Trial-Level Random Parameters}
	
	In (\ref{eq:proof-of-concept-dgm-within-trial}), several parameters vary across trials and are treated as random parameters. We consider two scenarios for the distributions of these trial-level parameters: one with moderate violations of the surrogacy and comparability assumptions, and another with only slight violations.

	First, the trial-level parameter $p$, which determines the distribution of the baseline covariate $X_1$, is sampled from a uniform distribution: $p \sim \text{Uniform}(0.30, 0.70)$. This reflects differences in participant populations across trials.
	Second, the trial-level parameter $\beta_{S, Z}$, which determines the treatment effect on the surrogate, is sampled as follows: $\beta_{S, Z} \sim \mathcal{N}(0.5, 0.3^2)$. The variance of this distribution captures the between-trial variability in treatment effects on the surrogate.
	Third, the trial-level parameter $\beta_{Y, Z}$ determines the treatment effect on the clinical endpoint, adjusted\footnote{Because we have adjusted here for post-randomization covariates, this parameter lacks a causal interpretation.} for the surrogate and the baseline covariate. It is sampled as follows:
	$\beta_{Y, Z} \sim \mathcal{N}(0, \sigma^2_{\beta_{Y, Z}})$. This parameter directly relates to the surrogacy assumption, which holds if $\beta_{Y, Z} = 0$. The magnitude of the surrogacy violation depends on the variance $\sigma^2_{\beta_{Y, Z}}$.
	Finally, the trial-level parameter $\beta_{Y, S^2}$ represents a quadratic effect of the surrogate on the clinical endpoint. It is sampled as follows:
	$\beta_{Y, S^2} \sim \mathcal{N}(0, \sigma^2_{\beta_{Y, S^2}})$. This parameter directly relates to the comparability assumption, which holds if $\beta_{Y, S^2}$ (and $\beta_{Y, Z}$) are constant across all trials. The magnitude of the comparability violation depends on the variance $\sigma^2_{\beta_{Y, S^2}}$.

	In the first scenario, we set $\sigma^2_{\beta_{Y, Z}} = \sigma^2_{\beta_{Y, S^2}} = 0.1^2$, representing moderate violations of surrogacy and comparability. In the second scenario, we set $\sigma^2_{\beta_{Y, Z}} = \sigma^2_{\beta_{Y, S^2}} = 0.05^2$, representing slight violations of these assumptions.

	We illustrate the data-generating mechanism---and the implied violations of the comparability assumption---for the proof-of-concept scenario in Figure \ref{fig:proof-of-concept-scenario}.
	This figure is based on simulated individual-participant data from six trials for (i) moderate and (ii) slight violations of the surrogacy and comparability assumptions.
	In Figure \ref{fig:proof-of-concept-scenario} (a), the surrogate is plotted against the clinical endpoint separately by trial and by the degree of violations of the surrogacy and comparability assumptions.
	Smooth regression curves, stratified by $X_1$, were added to further illustrate that the sign of the association between the surrogate and the clinical endpoint depends on $X_1$.
	We show the smooth regression curves in a single plot in Figure \ref{fig:proof-of-concept-scenario} (b) to demonstrate more clearly that the comparability assumption is more severely violated in the first scenario than in the second scenario. 

	\begin{figure}
		\centering
		\includegraphics[width=0.8\textwidth]{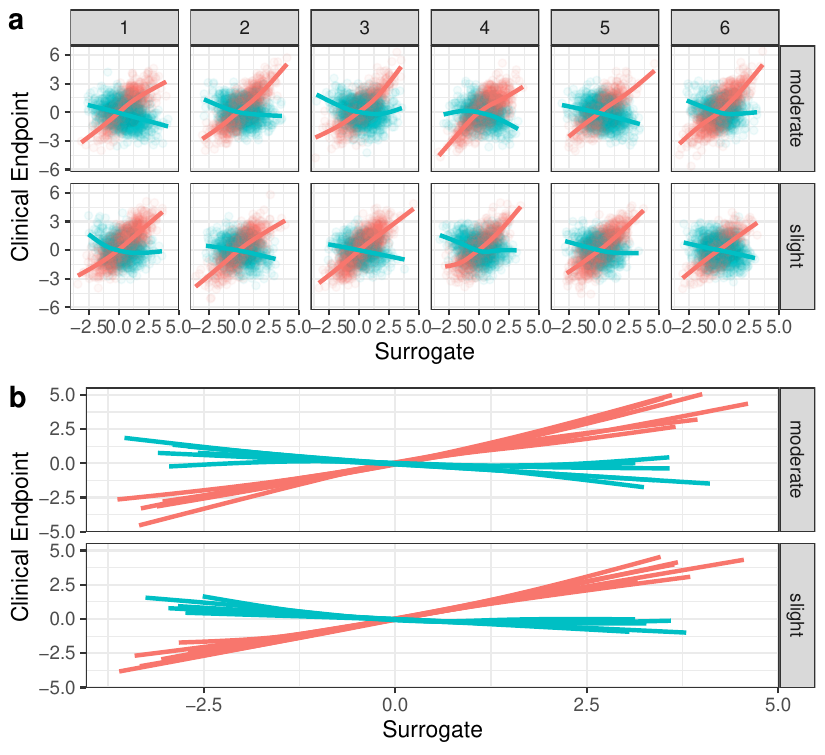}
		\caption{Proof-of-concept scenario: individual-participant data in six simulated trials under moderate and slight violations of the surrogacy and comparability assumptions. Dots and smooth regression curves are colored by the value for $X_1$.
		Panel a shows the simulated individual-participant data separately by trial (columns) and degree of violations of surrogacy and comparability (rows). Note that these plots pool observations across treatment groups.
		Smooth regression curves (full lines) are fitted to these data in each trial, stratified by $X_1$. Panel b shows the smooth regression curves from panel a in a single plot for each scenario.}
		\label{fig:proof-of-concept-scenario}
	\end{figure}

	\paragraph{Implementation}

	The surrogate index is estimated through a linear regression model with main effects for $S$ and $X_1$, the corresponding interaction, and a quadratic effect for $S$. The parameters are estimated by pooling data from the $N$ simulated trials in each replication of the simulation study. 
	The data-adaptive target parameter $\rho^{g_N}$ is approximated in each replication by sampling $2 \cdot 10^4$ trials with sample size $200$ and estimating $\rho^{g_N}_{\text{trial}}$ where $g_N$ is the previously estimated surrogate index and where the trial-level treatment effects are estimated by corresponding sample mean differences\footnote{Sample mean differences are unbiased estimators of the population mean differences; the covariance estimator based on the sample covariance is also unbiased. So, there is no need for large within-trial sample sizes when approximating $\rho^{g_N}_{\text{trial}}$ through Monte Carlo integration.}. 

	The trial-level treatment effects on the surrogate, the estimated surrogate index, and the clinical endpoint are estimated through fitting a linear regression model for $S$ and $Y$ with main effects for $X_1$ and $Z$.
	The estimated coefficients for the treatment indicator are consistent and asymptotically normal estimators of the treatment effects $\alpha^g$ and $\beta$ in terms of mean differences. The covariance matrices of $(\Hat{\alpha}^g, \Hat{\beta})'$ and $(\Hat{\alpha}^{g_N}, \Hat{\beta})'$ are estimated through the sandwich estimator using the \texttt{vcovHC()} function from the \textit{sandwich} R package with \texttt{type = "HC0"} \parencite{zeileis2004, sandwich}. This procedure is generally more efficient than simple mean differences but valid even if the linear model is misspecified \parencite{van2024covariate}. The sandwich estimator for the covariance matrix is consistent but not finite-sample unbiased; therefore, our justification relies on the alternative asymptotic regime described in Appendix \ref{appendix:alternative-asymptotic-framework}.

	\subsubsection{Vaccine Scenario}

	We first describe the data-generating mechanism for a fixed trial. Next, we describe how between-trial heterogeneity is simulated and how this relates to the comparability and surrogacy assumptions.
	
	\paragraph{Within-Trial Data-Generating Mechanism}
	
	We observe in each trial three baseline covariates: $X_1$ is an ``internal'' measure of exposure (e.g., the average number of contacts with persons with transmittable virus per week), $X_2$ is an external covariate associated with the amount of exposure (e.g., calendar time), and $X_3$ indicates whether the subject is naive (i.e., whether they have been previously infected). The treatment indicator $Z$ is as before. Furthermore, $S$ is the log10 neutralizing antibody titer measured 20 days after vaccination, and $Y$ is the infectious disease outcome 80 days after measuring $S$ (e.g., confirmed symptomatic infection).
	The within-trial observed-data distribution is determined by the following set of equations:
	\begin{equation}\label{eq:realistic-dgm-within-trial}
		\begin{cases}
			X_1 \sim \mathcal{N}(\mu_1, 1) \\
			X_2 \sim \mathcal{N}(\mu_2, 1) \\
			X_3 \sim \text{Bernoulli}(p) \\
			S = \beta_{S, Z} \cdot Z + \epsilon_S \text{ where } \epsilon_S \sim \mathcal{N}(0, 1) \\
			P(Y = 1 \mid \cdot) = \text{expit} \left\{ \beta_{Y, Z} \cdot Z + (1 + \beta_{Y, S} + X_3) \cdot S + X_1 - 2 \cdot X_3 \right\} \cdot 0.10 \cdot \exp \left(\sin(X_2) + 1 \right)
		\end{cases}
	\end{equation}
	where the first $n /2$ subjects are assigned to treatment $Z = 1$ and the last $n /2$ subjects to treatment $Z = 0$. We let $n = 5000$.
	
	\paragraph{Trial-Level Random Parameters}
	
	As for the proof-of-concept scenario, we consider two scenarios for the distributions of the trial-level parameters: one with moderate violations of the surrogacy and comparability assumptions and another with only slight violations.

	First, the trial-level parameters $\mu_1$, $\mu_2$, and $p$ reflect population differences across trials. These parameters are sampled as follows:
	\begin{equation*}
		\mu_1 \sim \text{Uniform}(-1, 1), \quad \mu_2 \sim \text{Uniform}(-1, 1), \quad p \sim \text{Uniform}(0.25, 0.75).
	\end{equation*}
	Second, the trial-level parameter $\beta_{S, Z}$, which determines the treatment effect on the surrogate, is sampled as follows: $\beta_{S, Z} \sim \mathcal{N}(1, 0.5^2)$.
	Third, the trial-level parameter $\beta_{Y, Z}$ determines the treatment effect on the clinical endpoint, adjusted for the surrogate and baseline covariates, and is sampled as follows: $\beta_{Y, Z} \sim \mathcal{N}(0, \sigma^2_{\beta_{Y, Z}})$. The surrogacy assumption holds if $\beta_{Y, Z} = 0$.
	Fourth, the trial-level parameter $\beta_{Y, S}$ corresponds to the association between the surrogate and the clinical endpoint and is sampled as $\beta_{Y, S} \sim \mathcal{N}(0, \sigma^2_{\beta_{Y, S}})$. This parameter directly relates to the comparability assumption, which holds if $\beta_{Y, S}$ (and $\beta_{Y, Z}$) are constant across all trials.

	In the first scenario, we set $\sigma^2_{\beta_{Y, Z}} = \sigma^2_{\beta_{Y, S}} = 0.30^2$, representing moderate violations of surrogacy and comparability. In the second scenario, we set $\sigma^2_{\beta_{Y, Z}} = \sigma^2_{\beta_{Y, S}} = 0.15^2$, representing slight violations of these assumptions.

	\begin{remark}
		In these simulations, we assume that $S$ is measured for all subjects; however, in practice, this may not always be the case. Indeed, $S$ is often measured some time after randomization, and if a subject is infected beforehand, $S$ is generally no longer well defined. A similar issue arises in the COVID-19 trials we analyze in Section \ref{sec:data-application} of the main text.
	\end{remark}

	We illustrate the data-generating mechanism---and the implied violations of the comparability assumption---for the vaccine scenario in Figure \ref{fig:vaccine-scenario}. This figure is based on individual-participant data from six simulated trials for (i) moderate and (ii) slight violations of the surrogacy and comparability assumptions.
	In Figure \ref{fig:vaccine-scenario} (a), the estimated surrogate index is plotted against the clinical endpoint separately by trial and by the degree of violations of the surrogacy and comparability assumptions. Smooth regression curves were added: under the surrogacy and comparability assumptions, these curves should be close to the identity line for all trials; deviations from this indicate a violation of comparability. 
	We show the smooth regression curves in a single plot in Figure \ref{fig:vaccine-scenario} (b) to demonstrate more clearly that the comparability assumption is more severely violated in the first scenario than in the second scenario. Note that variability in the curves for the $[0, 0.2]$-range of the estimated surrogate index is most important because most values of the estimated surrogate index are in this range.

	\begin{figure}
		\centering
		\includegraphics[width=0.8\textwidth]{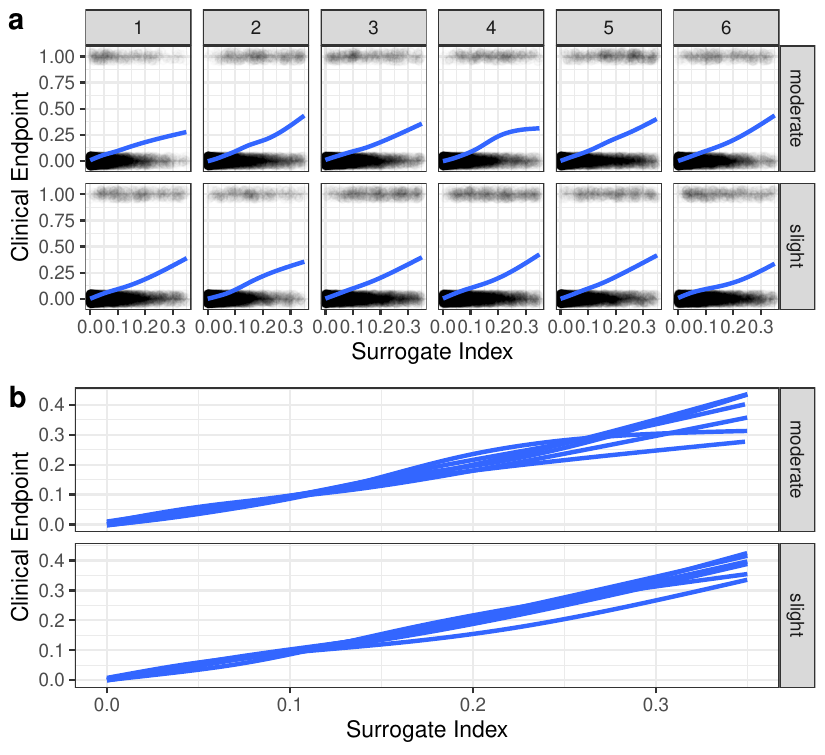}
		\caption{Vaccine scenario: individual-participant data in six simulated trials under moderate and slight violations of the surrogacy and comparability assumptions.
		Panel a shows scatter plots for the estimated surrogate index against the clinical endpoint, separately by trial (columns) and degree of violations of surrogacy and comparability (rows). Note that vertical jittering is used to avoid overplotting; in reality the clinical endpoint is either 0 or 1. 
		Smooth regression curves are fitted to these data in each trial. Panel b shows the smooth curves from panel a in the same plot.}
		\label{fig:vaccine-scenario}
	\end{figure}

	\paragraph{Implementation}

	The treatment effect on the clinical endpoint is defined---as in the data application---as the log relative risk and similarly for the treatment effect on the surrogate index. The treatment effect on the original surrogate is the mean difference. 
	These treatment effects are estimated using sample means and the covariance matrix is estimated through the delta method. 
	The surrogate index is estimated (pooling all data from the simulated $N$ trials in each replication) through the SuperLearner with several logistic regression models and a generalized additive logistic regression model as candidate learners \parencite{van2007super}. The meta-learner uses leave-one-trial-out cross-validation with the binomial log-likelihood as the loss function. The SuperLearner is implemented using the \textit{sl3} R package \parencite{coyle2021sl3-rpkg}. 

	The data-adaptive target parameter $\rho^{g_N}_{\text{trial}}$ is approximated by sampling $5 \cdot 10^3$ trials\footnote{We had to limit the number of trials here to ensure that the simulations remain computationally feasible.} with sample size $3000$ and estimating $\rho^{g_N}_{\text{trial}}$ where $g_N$ is the previously estimated surrogate index and where the trial-level treatment effects are estimated using sample means and the corresponding covariance matrices through the delta method.
	
	\subsection{Results}

	In this section, we present the results of the simulations. In each scenario, we simulated $500$ data sets. This number of replications is on the lower end of what is typically used in simulation studies, but it is sufficient to illustrate the main points of this paper. More replications would have been computationally too expensive because each replication requires us to generate individual-participant data from multiple trials and, in each replication, we have to numerically approximate the data-adaptive estimand $\rho^{g_N}_{\text{trial}}$ through Monte Carlo integration.

	We focus on estimation of and inference about the $\rho_{\text{trial}}$ ($\rho^{g_N}_{\text{trial}}$) parameter (the Pearson correlation between $\alpha$ ($\alpha^{g_N}$) and $\beta$). This parameter is a function of the covariance matrix parameters; hence, its estimate is computed as this function evaluated in the covariance parameter estimates. The standard errors for the estimated correlation parameters are obtained through the delta method using the sandwich estimator. Wald confidence intervals are obtained by using the delta method on Fisher's Z scale and transforming the limits back to the original scale (which ensures confidence limits in the unit interval). If the estimate is larger than 0.999, we use the Wald confidence interval on the $\rho$ scale (so without any transformation) because Fisher's Z-transformation is not defined for values larger than 1 and may be unstable for values very close to 1.
	Another set of confidence intervals is obtained through the BCa confidence intervals based on the multiplier bootstrap where the weights are sampled from a unit-exponential distribution. Throughout, we aim for a nominal coverage of 95\%.

	As mentioned in Section \ref{sec:simulations}, two finite-sample adjustments are always used in the simulations. First, in the expression for the plug-in estimator (\ref{eq:plug-in-estimator}), the variance of $(\Hat{\alpha}^g, \Hat{\beta})'$ is estimated by dividing by $N - 1$ instead of $N$. Second, the sandwich estimator in (\ref{eq:sandwich-estimator}) is adjusted by multiplying it with $(N - 1) / N$ and we use the $t$ distribution with $N - 1$ degrees of freedom for the Wald confidence intervals (instead of the standard normal distribution).

	The variance matrix estimator in (\ref{eq:plug-in-estimator}) may yield estimates that are not positive definite due to the measurement error correction. As a result, the corresponding estimated trial-level correlation might fall outside the valid range of $[-1, 1]$. To address this, we consider an additional finite-sample correction: if the estimated covariance matrix is not positive definite, we replace it with the nearest positive definite matrix in the Frobenius norm. This is achieved using the \texttt{nearPD()} function from the \textit{Matrix} R package. 

	The simulations also include coverage results for the Bayesian credible intervals based on the Bayesian model described in Section \ref{appendix:bayesian}. The Bayesian analyses were performed using the \textit{rstan} package with 1 chain and 10,000 iterations \parencite{RStan}. The first 5,000 iterations were discarded as burn-in.
	As already noted before, this Bayesian approach is asymptotically invalid (i.e., the posterior distribution may not concentrate around $\rho^{g_N}_{\text{trial}}$ or $\rho^{g^*}_{\text{trial}}$ as $N \to \infty$), but is included in the simulations because it is expected to be more stable for small $N$ and may outperform the asymptotically valid (but finite-sample unstable) non-parametric estimator and corresponding confidence intervals.

	\paragraph{Data-Adaptive Estimands}

	The distributions of the data-adaptive estimands are shown in Figure \ref{fig:distribution-estimands}. These figures show that there is little variability in the data-adaptive estimands in each setting. Furthermore, the data-adaptive estimands are close to 1 in the scenarios with only slight violations of the surrogacy and comparability assumptions. In all scenarios, the data-adaptive estimands are also higher than their non-data-adaptive counterparts. This illustrates the potential advantage of using an estimated surrogate index over the univariate surrogate $S$.

	\begin{figure}
		\centering
		\includegraphics[width=0.8\textwidth]{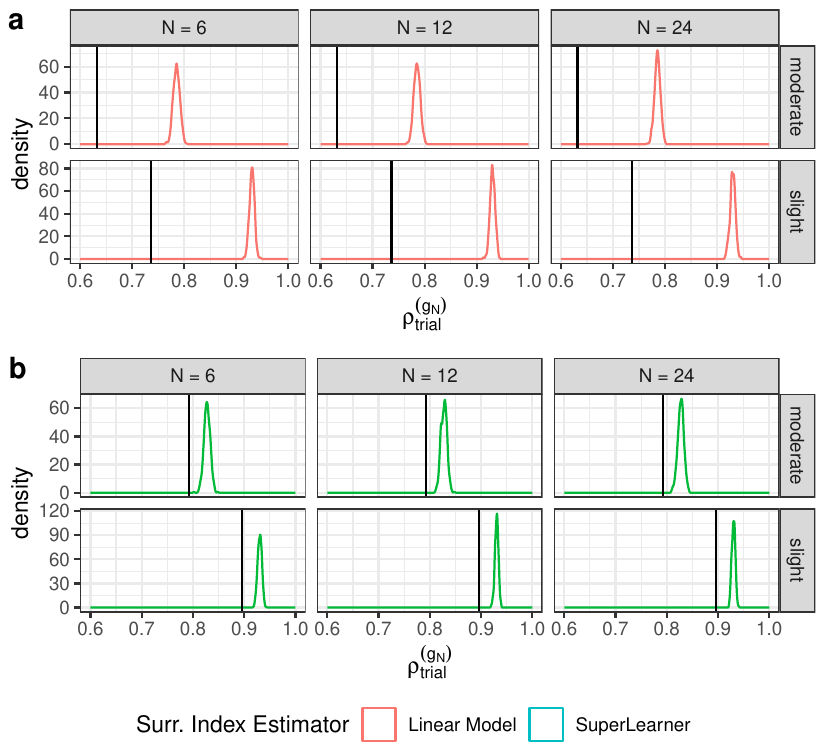}
		\caption{Distribution of the data-adaptive estimands $\rho_{\text{trial}}^{g_N}$ for (a) the proof-of-concept scenario and (b) the vaccine scenario. The vertical lines represent the values for the corresponding $\rho_{\text{trial}}$ for the original surrogate. The columns correspond to the different numbers of independent trials $N$ and the rows correspond to the degree of violation of the surrogacy and comparability assumptions, moderate or slight violations. Note that some of the apparent variability in $\rho_{\text{trial}}^{g_N}$ in these plots is due to numerical error in the Monte Carlo integration.}
		\label{fig:distribution-estimands}
	\end{figure}

	\paragraph{Bias and Estimation Error}

	Figure \ref{fig:mse-mean-bias} (a) shows the median bias of the estimated trial-level correlation coefficients for the corresponding estimands (which are not data-adaptive for the original surrogate and data-adaptive for the estimated surrogate indices). This figure shows that the estimator is median biased, but the bias decreases as $N$ increases. 

	\begin{figure}
		\centering
		\includegraphics[width=0.8\textwidth]{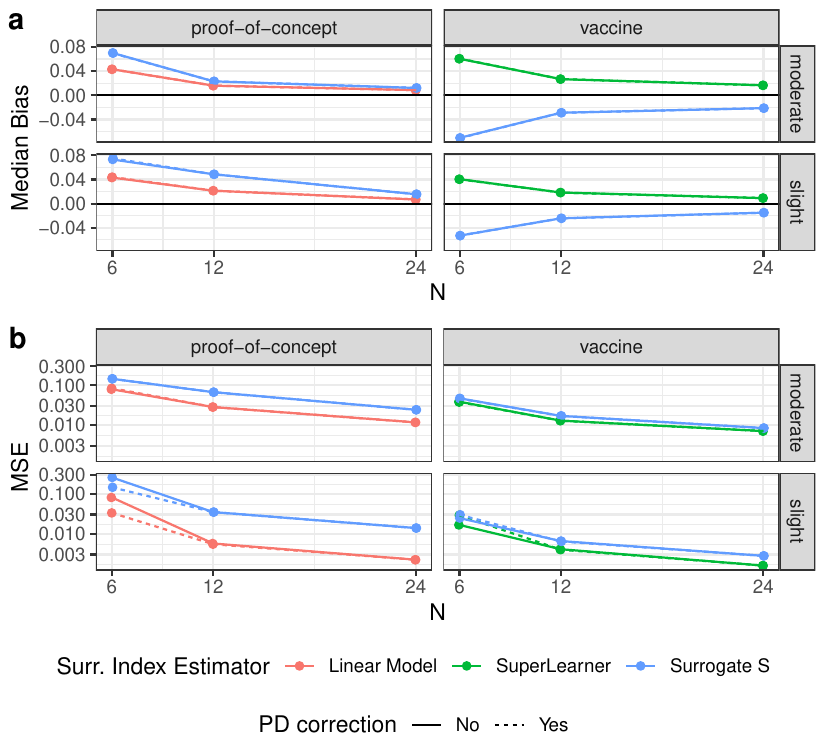}
		\caption{Median bias and mean squared error (MSE) of the estimated trial-level correlation coefficients for the corresponding data-adaptive or non-data-adaptive estimands as a function of the number of independent trials $N$. The columns correspond to the simulation scenario (proof of concept or vaccine) and the rows correspond to the degree of violation of the surrogacy and comparability assumptions, moderate or slight violations. PD: positive definite.}
		\label{fig:mse-mean-bias}
	\end{figure}

	Figure \ref{fig:mse-mean-bias} (b) shows the mean squared error (MSE) of the estimated trial-level correlation coefficients for the corresponding estimands. This figure shows that, as expected, the MSE decreases as $N$ increases. Even for $N = 6$ the MSE is small enough to consider the estimator as useful in practice. For a larger number of trials, the MSE is very small; hence, $\rho_{\text{trial}}^{g_N}$ can be estimated with high precision. Note, however, that this high precision is related to the true value being close to 1 in some scenarios. For instance, in the proof-of-concept scenario with moderate violations of the surrogacy and comparability assumptions, the MSE with the original surrogate is considerably larger than the MSE in other settings where the estimand is closer to one. 

	Figure \ref{fig:mse-mean-bias} (b) also shows that correcting the estimated covariance matrix to the closest positive definite matrix leads to a smaller MSE in some settings, but overall makes little difference. 
	
	\paragraph{Coverage and Inference}

	Figure \ref{fig:coverage-combined} (a) and (b) show the coverage of the confidence intervals for the estimated trial-level correlation coefficients for the corresponding estimands, respectively, for the Wald and BCa confidence intervals. 
	For $N = 6$, the empirical coverage is considerably lower than the nominal 95\% for both the Wald and BCa intervals. Which one of these two has closer-to-nominal coverage for $N = 6$ depends on the scenario.
	These deviations from nominal coverage are expected because (i) $N = 6$ is a very small sample size, (ii) the estimator is non-parametric, and (iii) the correlation parameter is a non-linear function of the covariance parameters.
	As $N$ increases to $24$, the empirical coverage approaches the nominal 95\% coverage for both the Wald intervals, but the BCa intervals still have some undercoverage.

	\begin{figure}
		\centering
		\includegraphics[width=0.8\textwidth]{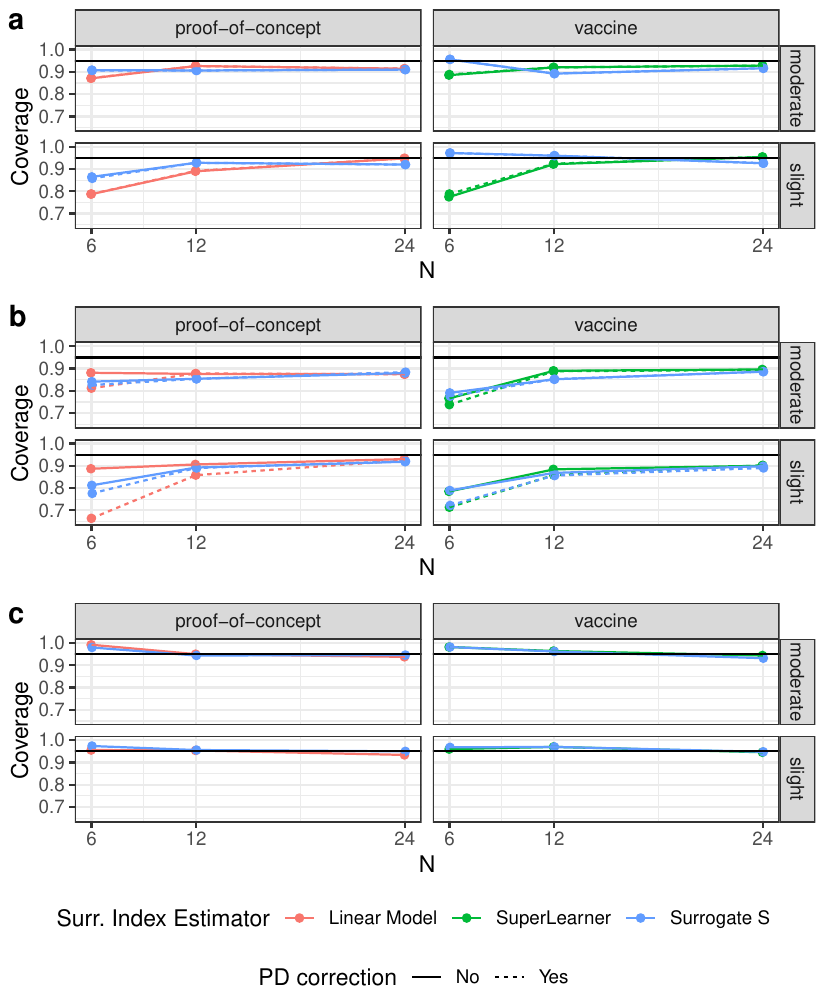}
		\caption{Empirical coverage of the confidence intervals for (a) the Wald, (b) BCa and (c) credible intervals, for the trial-level correlation coefficients as a function of the number of independent trials $N$. The horizontal line in each subfigure represents the nominal 95\% coverage. 
		The columns correspond to the simulation scenario (proof of concept or vaccine) and the rows correspond to the degree of violation of the surrogacy and comparability assumptions, moderate or slight violations. PD: positive definite.}
		\label{fig:coverage-combined}
	\end{figure}

	Figure \ref{fig:coverage-combined} (c) shows the empirical coverage of the Bayesian credible intervals. This figure shows that the Bayesian credible intervals have almost nominal coverage across all scenarios and sample sizes. As discussed before, the Bayesian approach is not asymptotically valid and the coverage will, therefore, converge to zero as $N \to \infty$. However, this does not preclude the Bayesian approach from being practically useful in small-sample size settings (e.g., for $N = 6$) where non-parametric methods may be unstable.



	\subsection{Verification of the Asymptotics and its Failure}\label{appendix:asymptotic-failure-simulations}

	To illustrate the alternative asymptotic framework of Appendix \ref{appendix:alternative-asymptotic-framework} and to verify the theoretical asymptotic results, we considered an additional simulation study where the number of independent trials $N$ and the within-trial sample sizes $n$ increase, at varying rates. We consider the same data-generating mechanism as in the proof-of-concept scenario (with moderate violations of surrogacy and comparability), but we set $N$ and $n$ to values in $\{100, 500, 1000, 2000, 5000\}$. We only consider the non-data-adaptive estimand here; so, the trial-level surrogacy for the original surrogate.
	We use $100$ MC replications for each setting.

	The trial-level treatment effects on the surrogate and the clinical endpoint are estimated, as before, through fitting a linear regression model for $S$ and $Y$ with main effects for $X_1$ and $Z$.
	However, we now add a flexible term for an extra random noise variable \texttt{X2} to the linear regression model, where the noise is generated independently. 
	More specifically, we add \texttt{bs(X2, df = 8)} to the linear regression model for the surrogate and clinical endpoints. This introduces many additional parameters in the linear regression model, which are not needed. We further use the same unadjusted sandwich covariance estimator as before. By adjusting for additional noise, we expect this covariance estimator to be downwardly biased in finite samples, but it will still be consistent. 

	\begin{remark}
		We have considered a bad covariance estimator to illustrate the failure of the fixed $n$, increasing $N$ asymptotic framework used in the main text. This covariance estimator (i.e., the unadjusted sandwich covariance estimator) should not be used in practice for small sample sizes because there are many finite-sample adjustments readily available \parencite{zeileis2004}. With those finite-sample adjustments, the covariance estimators are almost unbiased, even in small samples, and inference ``still works'' in practice when $N \gg n$. This just shows that the rates derived in Appendix \ref{appendix:alternative-asymptotic-framework} are probably not relevant in most realistic scenarios because most reasonable estimators that are (only) asymptotically valid, tend to also have a small bias in finite samples. 
	\end{remark}

	We first consider the z-statistic computed as $z = (\Hat{\rho}_{\text{trial}} - \rho_{\text{trial}}) / \widehat{\mathrm{SE}}$, where $\widehat{\mathrm{SE}}$ is the standard error of the estimated trial-level correlation coefficient. If the asymptotics work, then this statistic should follow a standard normal distribution. If the asymptotics fail, which we expect if the trial-level treatment effect estimators and their covariance estimators are biased, and the within-trial sample size $n$ does not increase sufficiently fast as a function of the number of trials $N$, then the z-statistic should not follow a standard normal distribution.
	The estimated distribution of the z-statistic is shown in Figure \ref{fig:failure-asymptotics} for all combinations of $n$ and $N$. This figure shows that the z-statistic is approximately standard normal in all scenarios except when $N \gg n$. 

	\begin{figure}
		\centering
		\includegraphics[width=0.8\textwidth]{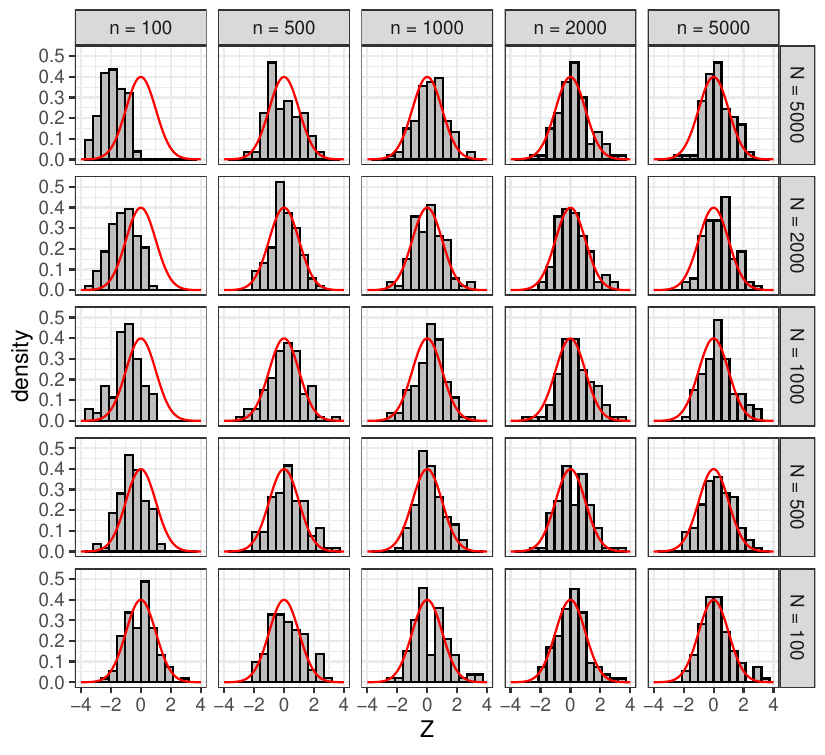}
		\caption{Distribution of the z-statistic for different combinations of the number of independent trials $N$ and the within-trial sample sizes $n$. This statistic is defined as $z = (\Hat{\rho}_{\text{trial}} - \rho_{\text{trial}}) / \widehat{\mathrm{SE}}$, where $\widehat{\mathrm{SE}}$ is the standard error of the estimated trial-level correlation coefficient. The red solid line represents the standard normal density.}
		\label{fig:failure-asymptotics}
	\end{figure}

	The failure of the z-statistic to follow a standard normal distribution when $N \gg n$ also leads to poor coverage of confidence intervals, whether they are indirectly based on the z-statistic (i.e., Wald confidence intervals) or constructed through the bootstrap. This is illustrated in Figure \ref{fig:failure-asymptotics-coverage} (a) and (b) for the Wald and BCa bootstrap confidence intervals, respectively. For instance, for $n = 100$, coverage is near nominal when $N = 100$ but decreases as $N$ increases to 5000. For $n = 500$ and larger, the coverage is nominal from $N = 100$ to $N = 5000$; however, if we would increase $N$ further, we expect the empirical coverage to decrease for any fixed $n$.

	\begin{figure}
		\centering
		\includegraphics[width=0.8\textwidth]{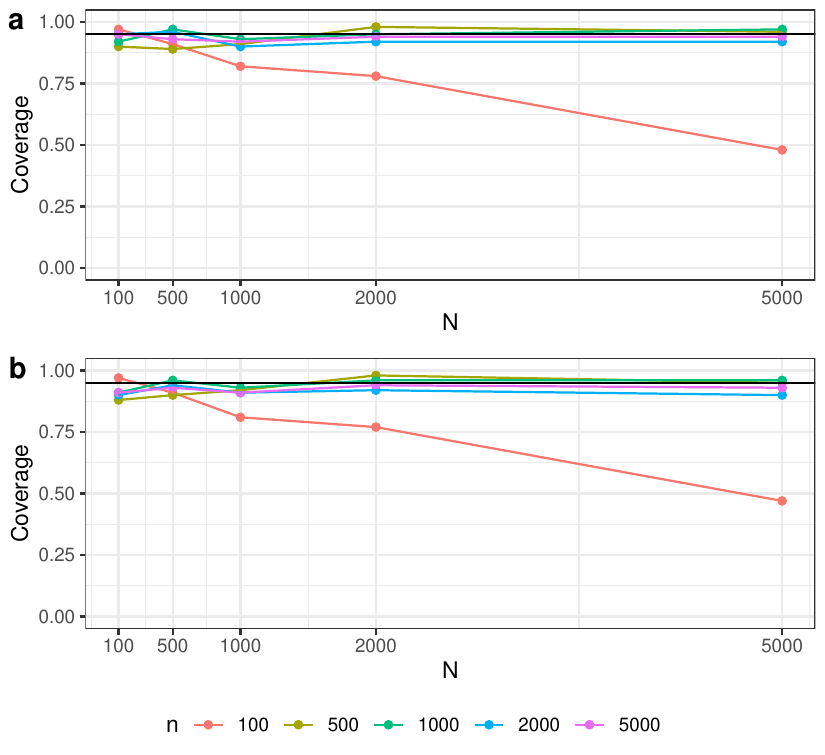}
		\caption{Empirical coverage of the confidence intervals for (a) the Wald and (b) BCa intervals, for the trial-level correlation coefficients as a function of the number of independent trials $N$ for varying within-trial sample sizes $n$. The horizontal line in each subfigure represents the nominal 95\% coverage. }
		\label{fig:failure-asymptotics-coverage}
	\end{figure}

	These results are of little practical relevance to the meta-analysis setting, however, as most meta-analyses have $n \gg N$ and $N$ is anyhow often very small. Still, these results may be useful beyond the meta-analysis setting where the observations are empirical distributions and the target parameters are related to the distribution of the ``true'' distributions (i.e., $P^T$) from which these empirical distributions (i.e., $\mathbb{P}^T_n$) are approximations. See, for instance, \textcite{lin2023causal} for a discussion of causal inference with distributional outcomes and \textcite{peng2026estimating} for explicit corrections for this approximation error.

	\appendixclearpage
	\section{Data Application}\label{appendix:data-application}

	This section contains additional results for the data application and provides additional details on the methodological choices. 
	We begin by exploring the data. Next, we describe how the neutralization titers are adjusted to the circulating strains.
	We then detail how the surrogate indices were estimated and evaluate their predictive performance.
	Finally, we provide additional results for the meta-analyses, including results from Bayesian meta-analyses.

	\subsection{Antibody Markers in Naive Subjects}\label{appendix:naive-subject-markers}

	SARS-CoV-2-naive subjects cannot have antibodies against the virus prior to exposure. Therefore, all naive placebo recipients are assumed to have the lowest possible antibody value, defined as the highest assay lower limit across the trials (10.84 BAU/ml for IgG Spike and 2.61 IU50/ml for nAb ID50) divided by two. Any vaccine recipient with an antibody measurement falling below these harmonized thresholds was also assigned this same standardized value.

	\subsection{Data Exploration}\label{appendix:data-description}

	Figure \ref{fig:distribution-risk-score} displays the distribution of the risk score (defined in the figure caption) by trial and treatment group. This figure shows that there are substantial differences in the distribution of the risk score across trials, indicating that the probability of COVID-19 varies considerably between trials. These between-trial differences underscore the importance of accounting for trial-specific baseline risks when estimating the surrogate index.
	We initially included the risk score in the surrogate index estimation but later decided against it because it obscured the interpretation: the risk score is estimated from the placebo-arm data (including the clinical endpoint) in each trial separately and hence would only be known in a new trial after observing the clinical endpoint.

	\begin{figure}
		\centering
		\includegraphics[width=0.8\textwidth]{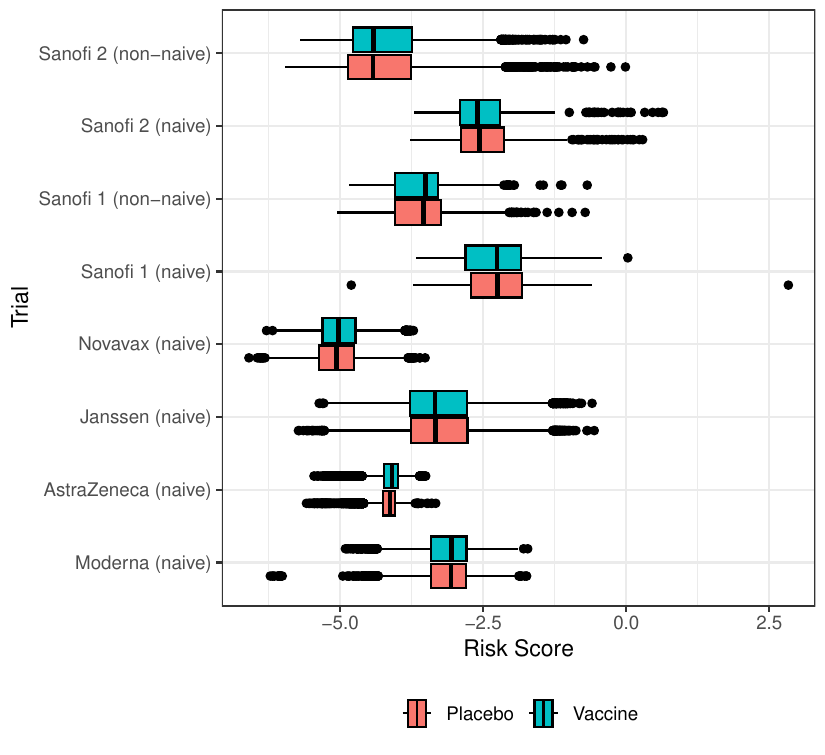}
		\caption{Distribution of the risk score by trial and treatment group. The risk score is defined as the logit of the predicted risk, based on an ensemble statistical learning algorithm trained on the placebo arm \parencite{Team2022}. The risk score was built for each trial separately except for Sanofi where two risk scores were developed, both pooling data over Stage 1 and Stage 2, one for naive and one for non-naive.}
		\label{fig:distribution-risk-score}
	\end{figure}

	Figure \ref{fig:distribution-titers} displays the distribution of IgG Spike, nAb ID50, and adjusted nAb ID50 by trial and treatment group, revealing substantial differences in the distribution of the antibody markers across trials. Furthermore, the distributions of these markers do not completely overlap between trials. This lack of overlap makes trial-level overfitting more likely, as further discussed in Appendix \ref{appendix:overfitting-of-the-surrogate-index-estimator}.


	\begin{figure}
		\centering
		\includegraphics[width=0.8\textwidth]{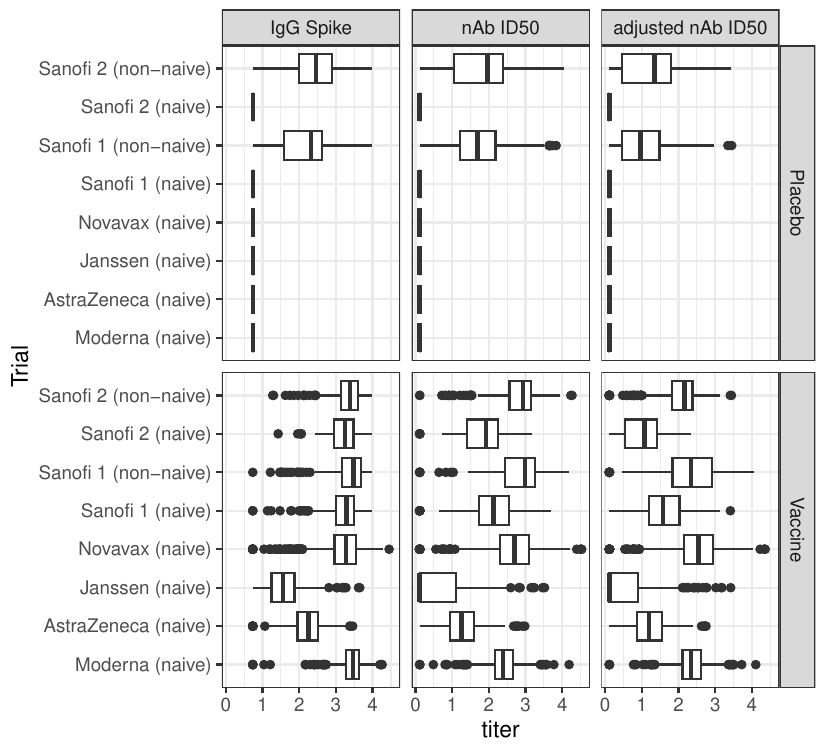}
		\caption{Distribution of IgG Spike, nAb ID50, and adjusted nAb ID50 by trial and treatment group. Antibody markers in the placebo groups of Moderna, AstraZeneca, Janssen, Novavax, and the naive Sanofi trials are set to the assay's lower limit of detection divided by 2, as these participants are SARS-CoV-2 naive and are not expected to have antibodies against SARS-CoV-2, as was verified by measurement of antibodies from random samples. The boxplots are weighted by the case-cohort sampling weights.}
		\label{fig:distribution-titers}
	\end{figure}

	Figure \ref{fig:cumulative-incidence} displays the estimated cumulative incidence functions of virologically confirmed symptomatic SARS-CoV-2 infection by trial and treatment group. This figure shows that there are substantial differences in cumulative incidence functions across trials and treatment groups. There is also substantial variability in the treatment effects (i.e., VE for symptomatic infection by 80 days) between trials. This variability is desirable, as between-trial differences in treatment effects on the clinical endpoint are required for a meaningful trial-level surrogacy analysis. Indeed, the trial-level correlation quantifies the amount of variability in the trial-level treatment effects on the clinical endpoint that is explained by a linear regression on the treatment effects on the surrogate (index). If this variability is minimal, there is little variability to explain. In the limiting case where there is no variability in the trial-level treatment effects on the clinical endpoint, the trial-level correlation is ill defined.

	\begin{figure}
		\centering
		\includegraphics[width=0.8\textwidth]{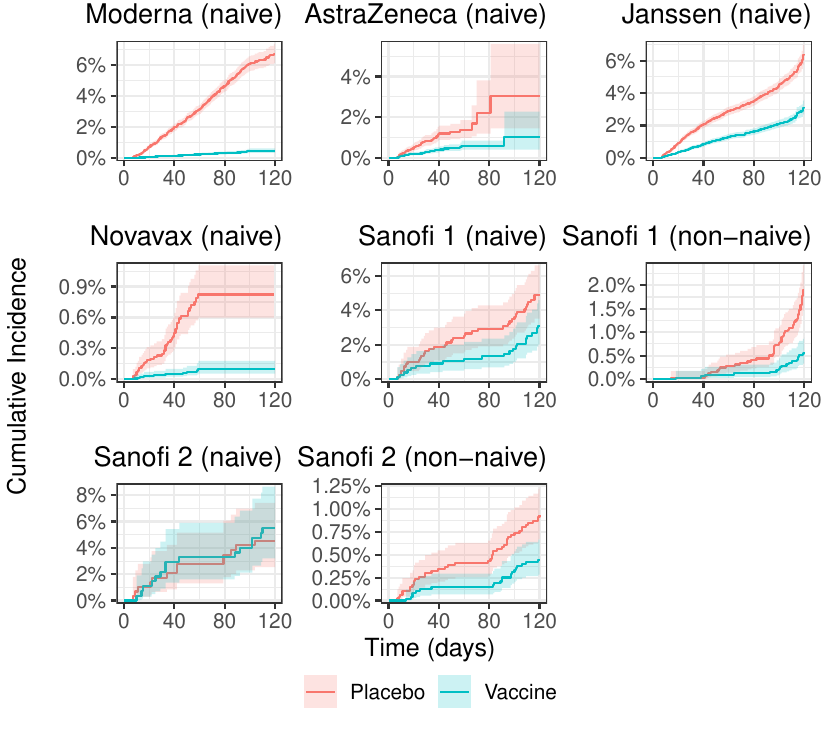}
		\caption{Estimated cumulative incidence functions (solid lines) of COVID-19 by trial and treatment group. The cumulative incidence functions are estimated using the Kaplan--Meier estimator. Shaded areas represent pointwise 95\% confidence intervals. Note that the cumulative incidences are shown up to 120 days post marker measurement visit, even though some trials have substantially longer follow-up.}
		\label{fig:cumulative-incidence}
	\end{figure}

	\subsection{Adjustment of Neutralization Titer to Circulating Strains}\label{appendix:adjustment-neutralization-titer}

	The neutralization titer measured in the vaccine trials quantifies neutralization against the D614G reference strain. However, as detailed in Section \ref{sec:data-description}, the D614G strain was not representative of the strains circulating in some trials. For example, during the Sanofi trials, the dominant strain was the Omicron strain.
	When there is such a mismatch between the measured neutralization titer and the circulating strains, the practical relevance of trial-level surrogacy evaluation becomes questionable. To address this, we adjust the neutralization titer measured against the reference strain to approximate the neutralization titer against the circulating strains. This adjustment allows us to evaluate whether the neutralization titer against the circulating strains serves as a meaningful trial-level surrogate for COVID-19 caused by symptomatic infections with those strains.

	In what follows, we describe the neutralization titer adjustment for data from a single trial. The same approach can be repeated for each trial separately.
	Let $s_{\text{ref}, i}$ be the neutralization titer (on the ID50 scale, so not log-transformed) against the reference strain for subject $i$ and $s_{v, i}$ the neutralization titer against strain $v$ for subject $i$. We assume a vaccine- and variant-specific reduction in titer against non-reference strains, constant across subjects:
	\begin{equation}\label{eq:adjustment-neutralization-titer}
		s_{v, i} = \frac{1}{c_v} \cdot s_{\text{ref}, i}
	\end{equation} 
	where $c_v$ is the \textit{abrogation coefficient}. Although $c_v$ is unknown, it can be estimated using observations of $s_{v, i}$ and $s_{\text{ref}, i'}$ as the geometric mean titer (GMT) ratio\footnote{Note that we are using indices $i$ and $i'$ here because the GMT ratio can be estimated even if $s_{v, i}$ and $s_{\text{ref}, i'}$ are not measured on the same subjects.}:
	\begin{equation*}
		\frac{\text{GMT}(s_{\text{ref}})}{\text{GMT}(s_v)}= \frac{e^{E\left( \log s_{\text{ref}, i'} \right)}}{e^{E \left( \log s_{v, i} \right)}} = \frac{e^{E\left( \log s_{\text{ref}, i'} \right)}}{e^{E \left( \log \left( s_{\text{ref}, i} \cdot \frac{1}{c_v} \right) \right)}} = e^{E \left( \log s_{\text{ref}, i'} \right) - E \left( \log s_{\text{ref}, i} \right) + \log c_v} = e^{\log c_v} = c_v
	\end{equation*}
	We use the estimated GMT ratios reported in the literature, which are summarized in Table \ref{tab:GMT-ratios} for each trial. For the Sanofi trials, we consider the GMT ratios for naive and non-naive participants separately.

	\begin{table}[htb]
		\centering
		\rotatebox{90}{
		\begin{minipage}{\textheight}
		\caption{Estimated GMT ratios for the neutralization titer against the reference strain and the circulating strains.}
		\label{tab:GMT-ratios}
		\begin{tabular}{ccccccccccccccc}
			\toprule
			Trial & Ref. & Epsilon & Gamma & Zeta & Alpha & Delta & Lambda & Iota & Beta & Mu & Omicron & BA.1 & BA.2 & BA.4.5 \\
			\midrule
			Moderna & 1 & 2.1 & 3.2 & 3.2 & - & - & - & - & - & - & - & - & - & - \\
			AstraZeneca & 1 & 1.55 & 3.4 & - & 0.9 & 1.55 & 1.55 & - & - & - & - & - & - & - \\
			Janssen & 1 & 1.55 & 3.4 & 2.2 & 0.9 & 1.55 & 1.55 & 0.9 & 3.6 & 3.6 & - & - & - & - \\
			Novavax & 1 & 2.1 & 3.2 & - & 1.2 & - & - & 2.3 & 7.4 & - & - & - & - & - \\
			Sanofi 1 Naive & 1 & - & - & - & - & 2 & - & - & - & - & 4.2 & 4.9 & 3.3 & 4 \\
			Sanofi 1 Non-Naive Vaccine & 1 & - & - & - & - & 0.59 & - & - & - & - & 4.7 & 7.3 & 4.3 & 6.5 \\
			Sanofi 1 Non-Naive Placebo & 1 & - & - & - & - & 0.43 & - & - & - & - & 5.2 & 11.8 & 9.1 & 11.4 \\
			Sanofi 2 Naive & 1 & - & - & - & - & 1.4 & - & - & - & - & 3.3 & 6.3 & 5.8 & 9.1 \\
			Sanofi 2 Non-Naive Vaccine & 1 & - & - & - & - & 1.1 & - & - & - & - & 4 & 7.7 & 5.2 & 6.6 \\
			Sanofi 2 Non-Naive Placebo & 1 & - & - & - & - & 1.1 & - & - & - & - & 4.4 & 4.6 & 3.7 & 4.3 \\
			\bottomrule
		\end{tabular}
		\end{minipage}
		}
	\end{table}

	Since multiple strains circulated in each trial, computing $s_v$ for a single strain is insufficient. Instead, we compute the weighted average neutralization titer against all circulating strains. The weights for each strain are subject-specific. Specifically, the weight for strain $v$ and subject $i$ is denoted by $\pi_{v, i}$ and equals the proportion of the circulating strains that is strain $v$ for individual $i$, averaged over the follow-up period for individual $i$\footnote{Suppose that there are only two strains $v_1$ and $v_2$, and participant $i$ is followed for 3 days where $30\%$, $40\%$, and $50\%$ of the COVID-19 cases were due to strain $v_1$ on day 1, 2, and 3, respectively. Then $\pi_{v_1, i} = \frac{1}{3} \left( 0.3 + 0.4 + 0.5 \right)$.}. These averaged proportions were obtained using circulating SARS-CoV-2 strain empirical frequencies from GISAID (using sequences from a subject's local geography) for all subjects in the vaccine trials \parencite{gisaid_variants,shu2017gisaid}.
	The adjusted neutralization titer $\tilde{s}_i$ can then be expressed as:
	\begin{equation*}
		\tilde{s}_i = \sum_{v} \pi_{v, i} \cdot s_{v, i} = \sum_{v} \pi_{v, i} \cdot \frac{1}{c_v} \cdot s_{\text{ref}, i} = s_{\text{ref}, i} \cdot \sum_{v} \frac{\pi_{v, i}}{c_v}.
	\end{equation*}
	Thus, for each individual $i$, the adjustment factor $\sum_{v} \frac{\pi_{v, i}}{c_v}$ is computed and multiplied by the original titer $s_{\text{ref}, i}$ to obtain $\tilde{s}_i$, the titer adjusted to the circulating strains.

	\begin{remark}[missing GMT ratios]
		The estimated GMT ratios in Table \ref{tab:GMT-ratios} cover the most prevalent strains in each trial. However, for some strains that circulated in some of the trials, we do not have estimated GMT ratios. This is only a minor practical issue because the missing strains had low prevalence. 
		To avoid any further problems caused by the missing estimated GMT ratios, we set the variant-specific weights $\pi_{v, i}$ to zero for the missing strains.
	\end{remark}

	The nAb ID50 titers are plotted against the adjusted nAb ID50 titers in Figure \ref{fig:scatterplot-titer-adjustment}. This figure shows that, as expected, the adjustment is larger for the Sanofi trials because the circulating strains in these trials were more distant from the reference strain than in the other trials.

	\begin{figure}
		\centering
		\includegraphics[width=0.8\textwidth]{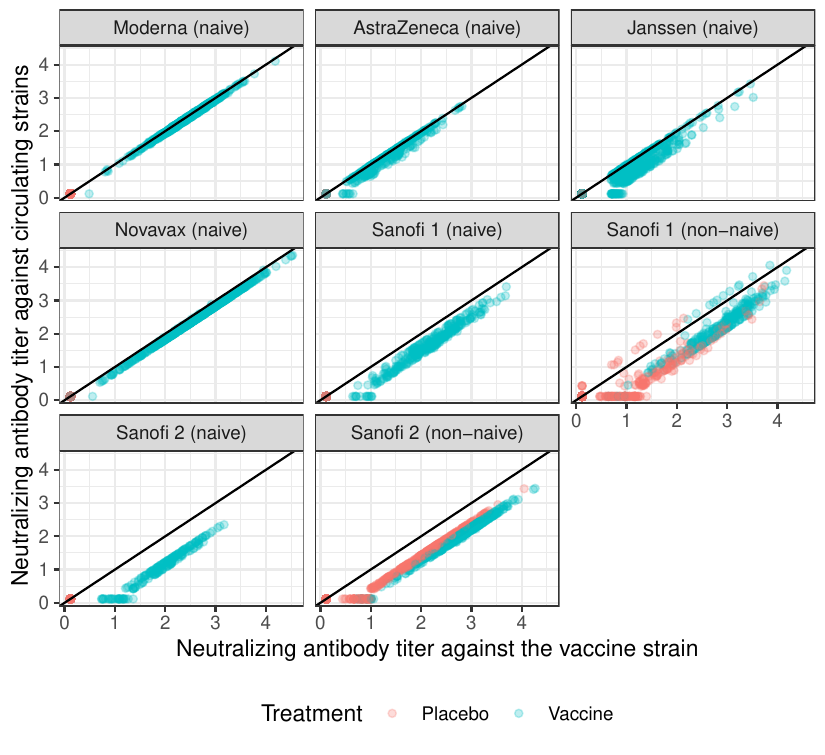}
		\caption{nAb ID50 titers against the adjusted nAb ID50 titers.}
		\label{fig:scatterplot-titer-adjustment}
	\end{figure}

	\paragraph{Limitations}

	The adjustment described above has important limitations. We highlight a few of them here.

	First, the adjustment is based on the assumption that there is a constant multiplicative reduction in neutralization titer against the non-reference strains across subjects (i.e., the constant $c_v$ is the same for all subjects from the same trial). This assumption may not hold in practice, as the reduction in neutralization titer may vary across individuals. Still, we expect this assumption to hold approximately in practice because neutralization titers against different strains are often strongly correlated with Spearman rank correlations $> 0.9$.

	Second, the adjustment is based on estimated GMT ratios. We do not take the sampling variability of these estimates into account, so if the GMT ratio estimates are imprecise, the adjustment may not be accurate.

	Third, we use a weighted average of the neutralization titers against the circulating strains. This summary loses information. 
	For example, a given subject may have a high neutralization titer against all circulating strains except one. In this case, the subject is still at risk of infection with the strain against which they have a low neutralization titer, even though their weighted average neutralization titer is high (and thus suggests that the subject is at low risk of COVID-19 caused by the circulating strains).
	However, we expect this to be a minor issue because neutralization titers against different strains are generally strongly correlated.

	\subsection{Surrogate Index Estimators}\label{appendix:surrogate-index-estimators}

	To estimate the surrogate index, we consider the following baseline covariates:
	\begin{itemize}
		\item \texttt{age}. Age in years measured at the time of enrollment in the trial.
		\item \texttt{BMI}. Body mass index (BMI) divided into four categories: underweight (BMI < 18.5), normal weight (18.5 $\le$ BMI < 25), overweight (25 $\le$ BMI < 30), and obese (BMI $\ge$ 30).
		\item \texttt{sex}. Biological sex.
		\item \texttt{high\_risk}. A binary variable indicating whether the subject is at high risk of severe COVID-19. 
		\item \texttt{logit\_prob\_infection}. The logit of the estimated cumulative incidence of COVID-19 at the 80 days post marker-measurement visit in the placebo arm of the participant's trial.
	\end{itemize}
	We further define \texttt{time\_to\_event} as the time from 6 days post marker-measurement visit until the infection event or censoring, and \texttt{event} as the event indicator. \texttt{infection\_80d} denotes whether the infection event was observed within 80 days post marker-measurement visit; subjects censored before 80 days have \texttt{infection\_80d = NA}. 
	We let \texttt{marker} refer to the base-10 logarithm of IgG Spike, nAb ID50, or adjusted nAb ID50.
	
	We consider the following two approaches to estimate the surrogate index: 
	\begin{itemize}
		\item \textbf{Cox proportional hazards model.} The surrogate index is estimated using a Cox proportional hazards model with the following form: \texttt{Surv(time\_to\_event, event) $\sim$ sex + high\_risk + BMI + bs(age) + strata(trial) + ns(marker, knots = c(1.5, 2.5), Boundary.knots = c(0, 4))}. We truncate follow-up at 120 days post marker-measurement visit; hence, all observations are censored at 120 days (even if longer follow-up is available for some subjects). Observations are weighted by the case-cohort weights. 
		Probability of COVID-19 by 80 days post marker-measurement visit is estimated by using the estimated regression coefficients and the estimated baseline hazard function from the Janssen trial. This estimated probability is a prediction of \texttt{infection\_80d}.
		\item \textbf{SuperLearner.} The surrogate index is estimated using the SuperLearner algorithm implemented in the \textit{sl3} R package \parencite{van2007super, coyle2021sl3-rpkg} using \texttt{infection\_80d} as the outcome. We use the following library of candidate algorithms: 
		\begin{enumerate}
			\item Logistic regression with main effects for all covariates and the surrogate.
			\item Logistic regression with main effects for only \texttt{logit\_prob\_infection} and the surrogate.
			\item Logistic regression with main effects for all covariates and the following natural cubic splines for the surrogate: \texttt{ns(marker, knots = c(1.5, 2.5), Boundary.knots = c(0, 4))}.
			\item Logistic regression with main effects for all covariates and the surrogate, and an interaction term between the \texttt{marker} and \texttt{sex}.
			\item Logistic regression combining all terms from the previous two models.
			\item Generalized additive model (GAM) with smooth terms for all continuous covariates and the surrogate. We use the defaults of \texttt{Lrnr\_gam}.
		\end{enumerate}	
		We use the meta-learner as implemented in \texttt{Lrnr\_solnp} with the binomial log-likelihood loss function, and we use leave-one-trial-out cross-validation instead of the default.		
		The observations with a missing \texttt{infection\_80d} or \texttt{marker} are excluded; the remaining observations are weighted by the product of the inverse probability of censoring weights (IPCWs) and the case-cohort weights. The IPCWs are estimated using the Kaplan--Meier estimator stratified by trial and treatment arm.
	\end{itemize}

	We have estimated the surrogate index by pooling the individual-participant data from all trials except the non-naive Sanofi trials. Receiver operating characteristic (ROC) curves for predicting \texttt{infection\_80d} using the estimated surrogate indices are shown in Figures \ref{fig:roc-pseudoneutid50-naive-only}--\ref{fig:roc-bindspike-naive-only}, stratified by trial. The ROC curves for trials whose data were not used for estimating the surrogate index (Sanofi 1 (non-naive) and Sanofi 2 (non-naive)) serve as an external validation of predictive performance.
	These ROC curves are based on subjects with observed \texttt{infection\_80d} and \texttt{marker}, weighted by the product of the IPCWs and the case-cohort weights.

	\begin{figure}
		\centering
		\includegraphics[width=0.8\textwidth]{"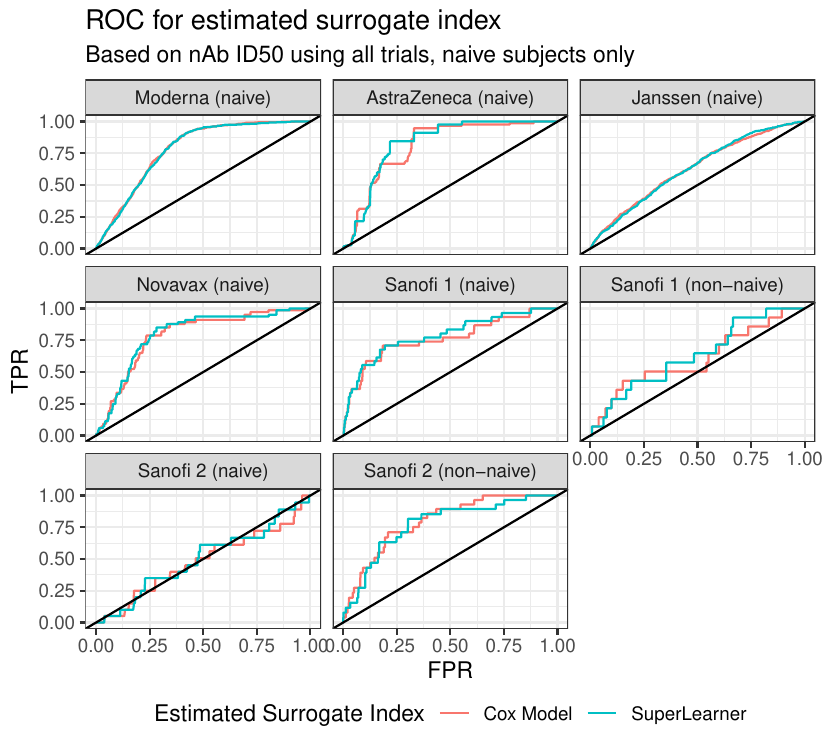"}
		\caption{Receiver operating characteristic (ROC) curves for predicting \texttt{infection\_80d} with the estimated surrogate index based on nAb ID50. The surrogate index was estimated using all trials with only naive subjects. }
		\label{fig:roc-pseudoneutid50-naive-only}
	\end{figure}

	\begin{figure}
		\centering
		\includegraphics[width=0.8\textwidth]{"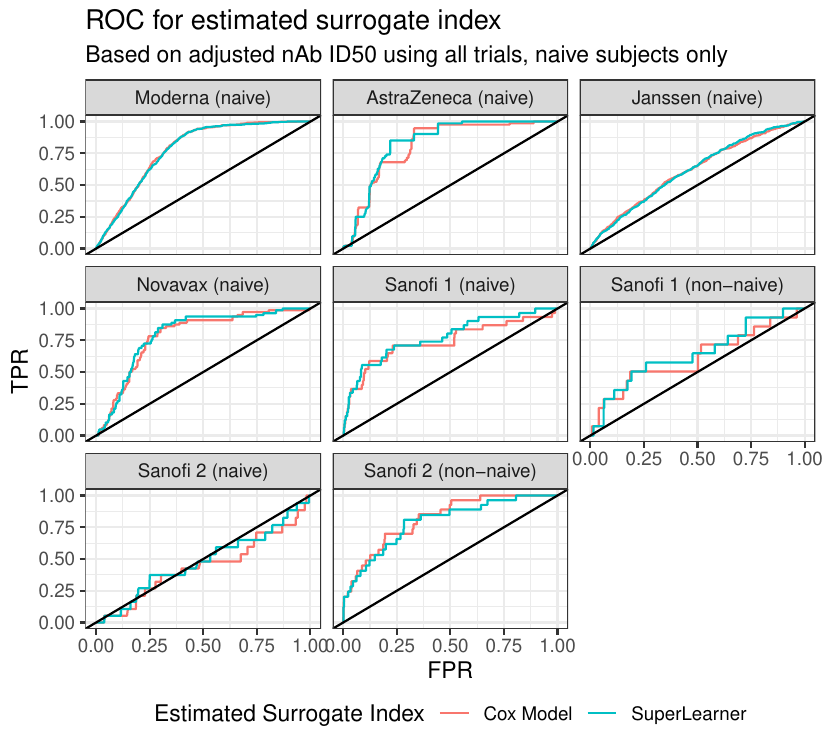"}
		\caption{Receiver operating characteristic (ROC) curves for predicting \texttt{infection\_80d} with the estimated surrogate index based on the adjusted nAb ID50. The surrogate index was estimated using all trials with only naive subjects.}
		\label{fig:roc-pseudoneutid50-adjusted-naive-only}
	\end{figure}

	\begin{figure}
		\centering
		\includegraphics[width=0.8\textwidth]{"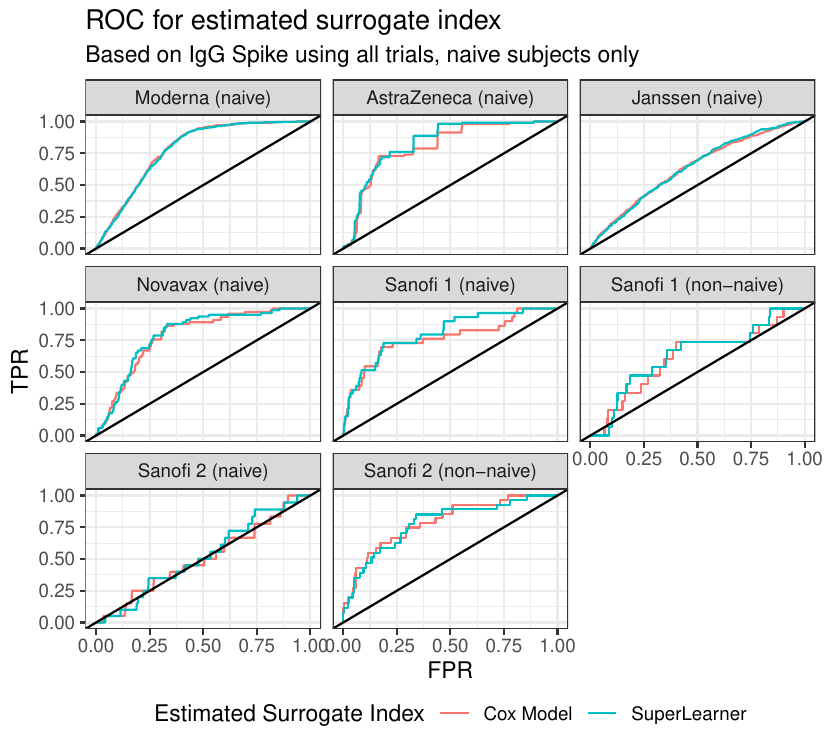"}
		\caption{Receiver operating characteristic (ROC) curves for predicting \texttt{infection\_80d} with the estimated surrogate index based on IgG Spike. The surrogate index was estimated using all trials with only naive subjects. }
		\label{fig:roc-bindspike-naive-only}
	\end{figure}

	\subsection{Additional Results}

	\subsubsection{Additional Non-Parametric Meta-Analysis Results}

	Tables \ref{table:trial-level-correlations-naive-only-sandwich} and \ref{table:trial-level-correlations-pseudo-real-sandwich} present the same results as Tables \ref{table:trial-level-correlations-naive-only} and \ref{table:trial-level-correlations-pseudo-real} in the main text, but using Wald-type confidence intervals based on the sandwich covariance estimator and a $t$ distribution with $N - 1$ degrees of freedom (instead of the BCa bootstrap confidence intervals).
	These confidence intervals are narrower than the BCa bootstrap confidence intervals, but as suggested by the simulations, they may have undercoverage. 

\begin{table}
        \centering
	\caption{Estimated trial-level correlations with 95\% Wald confidence intervals based on the sandwich estimator for the standard errors for the meta-analysis that includes all naive trials.
	The Antibody Marker column corresponds to the meta-analyses where the antibody marker (without any transformation) is evaluated as a trial-level surrogate endpoint.}
        \label{table:trial-level-correlations-naive-only-sandwich}
        \begin{tabular}{lccc}
            \toprule
            Surrogate & Antibody Marker & \multicolumn{2}{c}{Est.~Surrogate Index} \\
            & & Cox Model & SuperLearner \\
            \midrule
            IgG Spike & 0.33 (-0.50, 0.84) & 0.34 (-0.56, 0.87) & 0.33 (-0.53, 0.86) \\
            nAb ID50 & 0.47 (-0.21, 0.84) & 0.56 (-0.10, 0.88) & 0.57 (-0.07, 0.88) \\
            Adjusted nAb ID50 & 0.77 (0.34, 0.93) & 0.88 (0.66, 0.96) & 0.87 (0.60, 0.96) \\
            \bottomrule
        \end{tabular}
    \end{table} 

\begin{table}[htb]
        \centering
	\caption{Estimated trial-level correlations with 95\% Wald confidence intervals based on the sandwich estimator for the standard errors for the meta-analysis with the pseudo-real data.
	The Antibody Marker column corresponds to the meta-analyses where the antibody marker (without any transformation) is evaluated as a trial-level surrogate endpoint.}
        \label{table:trial-level-correlations-pseudo-real-sandwich}
        \begin{tabular}{llccc}
            \toprule
            Pseudo-Real Data & Surrogate & Antibody Marker & \multicolumn{2}{c}{Est.~Surrogate Index} \\
            & & & Cox Model & SuperLearner \\
            \midrule
            \multirow{3}{*}{Increased Precision} 
                & IgG Spike & 0.32 (-0.48, 0.83) & 0.33 (-0.54, 0.86) & 0.31 (-0.51, 0.84) \\
                & nAb ID50 & 0.44 (-0.21, 0.82) & 0.54 (-0.11, 0.86) & 0.54 (-0.08, 0.86) \\
                & Adjusted nAb ID50 & 0.73 (0.31, 0.91) & 0.84 (0.58, 0.94) & 0.84 (0.54, 0.94) \vspace{0.1cm}\\
            \multirow{3}{*}{Increased Number of Trials} 
                & IgG Spike & 0.34 (0.01, 0.60) & 0.35 (-0.01, 0.63) & 0.33 (-0.01, 0.61) \\
                & nAb ID50 & 0.47 (0.23, 0.65) & 0.57 (0.35, 0.73) & 0.58 (0.37, 0.73) \\
                & Adjusted nAb ID50 & 0.77 (0.65, 0.86) & 0.89 (0.83, 0.93) & 0.88 (0.81, 0.92) \\
            \bottomrule
        \end{tabular}
    \end{table}

	\subsubsection{Bayesian Meta-Analysis}\label{appendix:bayesian-analysis}

	We consider a Bayesian analysis based on the bivariate normal model as described in Appendix \ref{appendix:bayesian} and studied in the simulations. We use the same estimated surrogate indices as in the non-parametric frequentist analysis and ignore the fact that $g_N$ depends on the observed data. 
	The Bayesian analyses were performed using \textit{rstan} with 4 chains and 20,000 iterations per chain \parencite{RStan}. The first 5,000 iterations were discarded as burn-in. The posterior distributions for $\rho_{\text{trial}}^{g_N}$ are shown in Figures \ref{fig:posterior-rho-trial} and \ref{fig:posterior-rho-trial-prop-regression}.

	\begin{figure}
		\centering
		\includegraphics[width=0.8\textwidth]{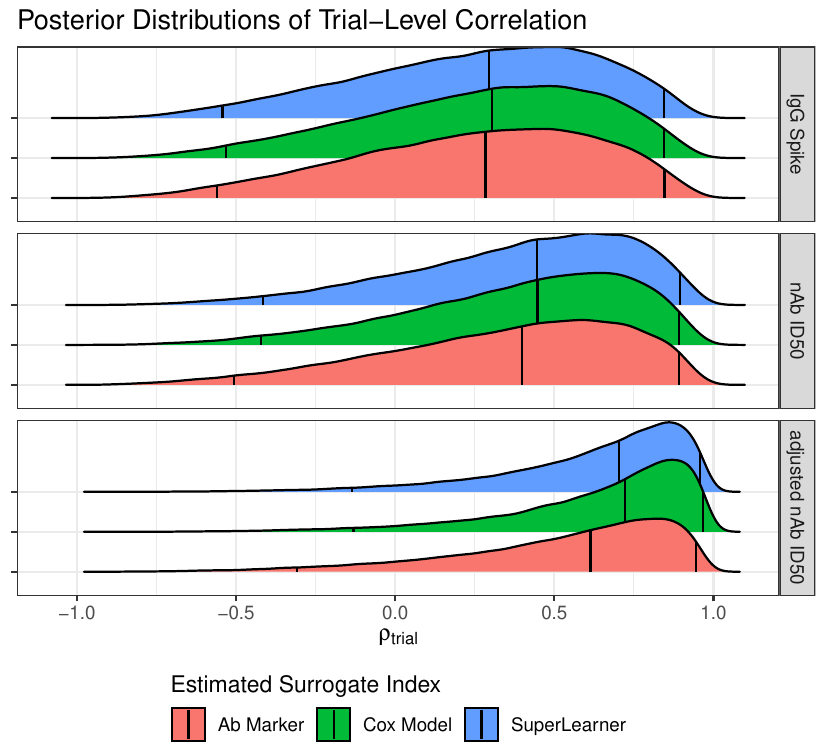}
		\caption{Posterior distributions for $\rho_{\text{trial}}^{g_N}$ from the Bayesian analysis assuming a bivariate normal model. The posterior distributions are shown for the original surrogate and the estimated surrogate indices.
		The first and third vertical lines represent the 95\% equal-tailed credible interval limits, and the middle vertical line represents the posterior median.}
		\label{fig:posterior-rho-trial}
	\end{figure}

	\begin{figure}
		\centering
		\includegraphics[width=0.8\textwidth]{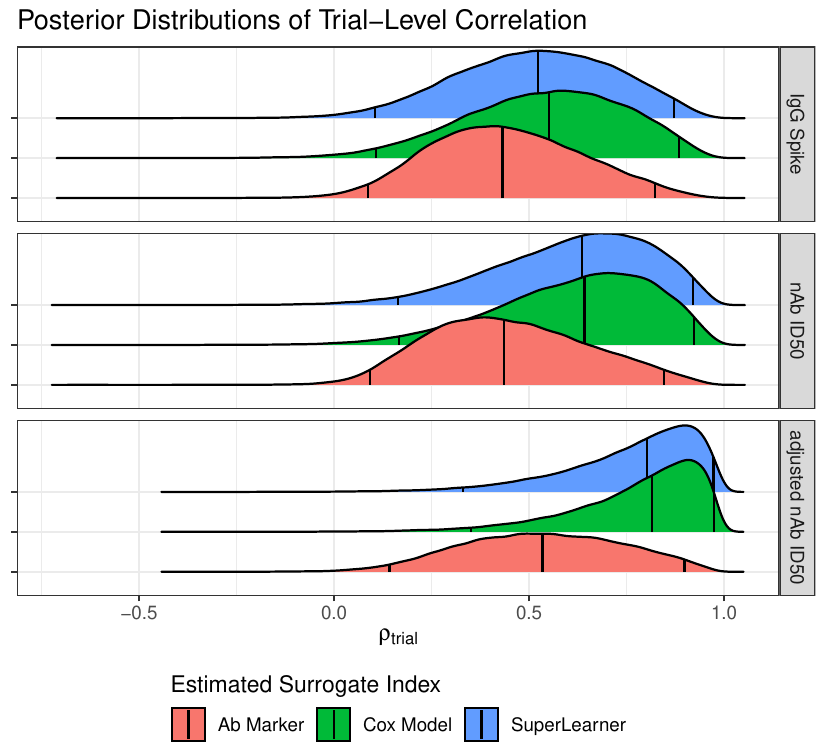}
		\caption{Posterior distributions for $\rho_{\text{trial}}^{g_N}$ from the Bayesian analysis assuming a bivariate normal model with proportional regression of $\beta$ on $\alpha^g$, as described in Appendix \ref{appendix:bayesian}. The posterior distributions are shown for the original surrogate and the estimated surrogate indices.
		The first and third vertical lines represent the 95\% equal-tailed credible interval limits, and the middle vertical line represents the posterior median.}
		\label{fig:posterior-rho-trial-prop-regression}
	\end{figure}
	
\clearpage
\includepdf[pages={-}]{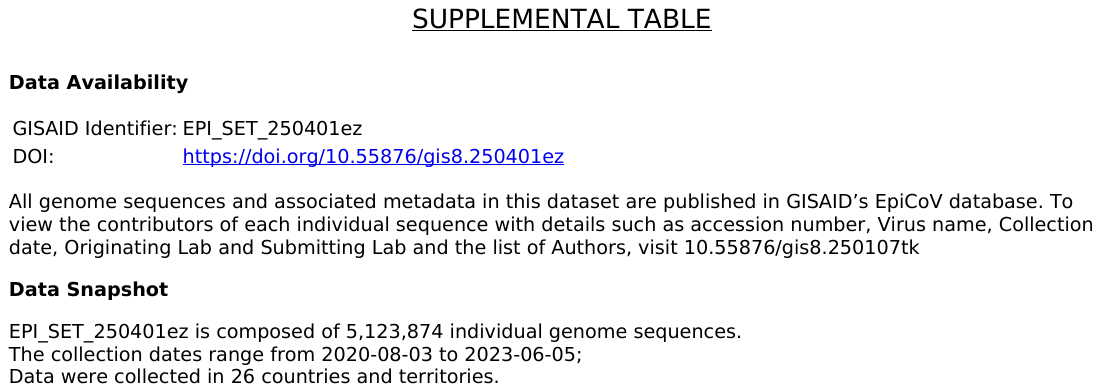}
	
\end{document}